\documentclass[11pt, letter]{article}

\usepackage[english]{babel}
\usepackage[T1]{fontenc}
\usepackage{stmaryrd}
\usepackage{algorithmic}
\usepackage{amssymb}
\usepackage{marginnote}
\usepackage[ruled,vlined,linesnumbered]{algorithm2e}
\usepackage{booktabs,tabularx,threeparttable}
\usepackage[normalem]{ulem}

\usepackage{tikz}
\usetikzlibrary{fit,backgrounds,positioning}
\usepackage{fullpage}
\usepackage{mathrsfs}
\usepackage{amsmath}
\usepackage{bbm}
\usepackage{graphicx}
\usepackage{stackengine}
\usepackage{amsthm}
\usepackage{subcaption}
\usepackage{authblk}

\usepackage{thmtools}
\usepackage{thm-restate}

\usepackage[colorlinks=true, allcolors=blue]{hyperref}

\usepackage{fancyhdr,lastpage}
\usepackage{amsfonts}
\usepackage[dvipsnames]{xcolor}

\usepackage{enumitem}

\newcommand{\mylim}[2][n\rightarrow\infty]{\left(\lim_{#1}#2\right)}
\newcommand{\mysum}[4]{\left(\sum_{#1=#2}^{#3} #4\right)}
\newcommand{\myprod}[4]{\left(\prod_{#1=#2}^{#3} #4\right)}
\newcommand{\fn}[2]{#1\left(#2\right)}

\newcommand{\set}[1]{\left\{#1\right\}}
\newcommand{\floor}[1]{\left\lfloor#1\right\rfloor}
\newcommand{\ceil}[1]{\left\lceil#1\right\rceil}
\newcommand{\angles}[1]{\left\langle #1 \right\rangle}
\newcommand{\abs}[1]{\left| #1 \right|}
\newcommand{\norm}[1]{\left|\left| #1 \right|\right|}

\newcommand{\binomMax}[1][n]{\binom{#1}{\floor{\frac{#1}{2}}}}
\newcommand{\binomSet}[2]{[#1]^{(#2)}}
\newcommand{\unionNat}[1][k]{\bigcup_{#1=1}^{\infty}}

\newcommand{\opInt}[2]{\left(#1,#2\right)}
\newcommand{\clInt}[2]{\left[#1,#2\right]}
\newcommand{\clopInt}[2]{\left[#1,#2\right)}
\newcommand{\opclInt}[2]{\left(#1,#2\right]}

\newcommand{\smallo}[1]{o\left(#1\right)}
\newcommand{\bigO}[1]{O\left(#1\right)}

\DeclareMathOperator*{\amn}{arg\,min}
\DeclareMathOperator*{\amx}{arg\,max}
\DeclareMathOperator*{\argmax}{arg\,max}
\DeclareMathOperator*{\argmin}{arg\,min}
\newcommand{\eqtext}[1]{\ensuremath{\stackrel{\text{#1}}{=}}}

\newcommand{\bb}[0]{\mathbb}
\newcommand{\bbm}[0]{\mathbbm}
\newcommand{\scr}[0]{\mathscr}
\newcommand{\ov}[0]{\overline}

\newcommand{\power}[2]{\left(#1\right)^{#2}}

\newtheorem{theorem}{Theorem}
\newtheorem{claim}[theorem]{Claim}
\newtheorem{remark}[theorem]{Remark}
\newtheorem{conjecture}[theorem]{Conjecture}
\newtheorem{lemma}[theorem]{Lemma}
\newtheorem{definition}[theorem]{Definition}

\newtheorem*{itheorem}{Theorem}

\usetikzlibrary{arrows,shapes,decorations.pathmorphing,automata,backgrounds,petri}
\usetikzlibrary{positioning,calc}
\usetikzlibrary{patterns}
\tikzstyle{legendborder}=[rectangle, draw, black, rounded corners, thin, top color=white, text=black, minimum width=2.5cm, text width=4.5cm]
\tikzstyle{legendnoborder}=[rectangle, draw, white, rounded corners, thin, top color=white, text=black, minimum width=2.5cm]
\tikzstyle{selected edge} = [draw,line width=1pt,black]
\usetikzlibrary{shapes,backgrounds,plothandlers,plotmarks,calc,arrows,fadings,decorations.pathreplacing,decorations.pathmorphing}
\newcommand{\midarrow}{\tikz \draw[-triangle 90] (0,0) -- +(.1,0);}
\tikzstyle{v}=[circle,fill=black,draw=black!75,inner sep=0pt,minimum size=0.3em]
\tikzstyle{I}=[circle,draw=black!75,inner sep=0pt,minimum size=0.8em]
\tikzstyle{J}=[rectangle,draw=black!75,inner sep=0pt,minimum size=0.7em]
\tikzstyle{vertex}=[circle,inner sep=2,minimum size =2mm,semithick,fill=white!80!blue, draw=black]

\newcommand{\mc}{\mathcal}
\newcommand{\mb}{\mathbb}
\newcommand{\WW}{\textsf{W}}
\newcommand{\WT}{\textsf{W[2]}}
\newcommand{\WO}{\textsf{W[1]}}
\newcommand{\PP}{\textsf{P}}
\newcommand{\PPSPACE}{\textsf{PSPACE}}
\newcommand{\NPP}{\textsf{NP}}
\newcommand{\FPT}{\textsf{FPT}}

\newcommand{\OO}{\mathcal{O}}
\newcommand{\A}{\mathcal{A}}
\newcommand{\B}{\mathcal{B}}
\newcommand{\C}{\mathcal{C}}
\newcommand{\D}{\mathcal{D}}
\newcommand{\E}{\mathcal{E}}
\newcommand{\F}{\mathcal{F}}
\newcommand{\G}{\mathcal{G}}
\newcommand{\M}{\mathcal{M}}
\newcommand{\head}{\mathcal{H}}
\newcommand{\R}{\mathcal{R}}
\newcommand{\T}{\mathcal{T}}
\renewcommand{\S}{\mathcal{S}}
\newcommand{\pwtw}{w}
\newcommand{\fvs}{f}
\newcommand{\degen}{d}
\newcommand{\numb}{r}
\newcommand{\struct}{s}

\newcommand{\Q}{\textsc{Multi-Tape-Rec}}
\newcommand{\AutRec}{\textsc{Tape-Rec}}
\newcommand{\SyncAutRec}{\textsc{Sync-Tape-Rec}}
\newcommand{\SyncPathAutRec}{\textsc{Sync-Path-Tape-Rec}}
\newcommand{\SyncQ}{\textsc{Sync-Multi-Tape-Rec}}
\newcommand{\PathAutRec}{\textsc{Path-Tape-Rec}}
\newcommand{\TSDSR}{\textsc{TS-DSR}}
\newcommand{\TJDSR}{\textsc{TJ-DSR}}

\newcommand{\tw}{\mathrm{tw}}
\newcommand{\Cc}{\mathscr{C}}
\newcommand{\ie}{i.e., }
\newcommand{\eg}{e.g., }
\newcommand{\dsr}{\textsc{TS-DSR}}
\newcommand{\dscr}{\textsc{TS-DCR}}
\newcommand{\dsrfullts}{\textsc{Dominating Set Reconfiguration under Token Sliding }}
\newcommand{\bad}[1]{\tau(#1)}

\newcommand{\EFO}{EF1\textsubscript{outer} }

\begin{document}

\title{Fair Division of Graphs: Beyond Traceability\footnote{The first and fourth author are supported by ANR project ENEDISC (ANR-24-CE48-7768-01).}}

\author[1]{Nicolas Bousquet}
\author[2]{Frank Connor}
\author[3]{Agn\`es Totschnig}
\author[1]{S\'ebastien Zeitoun}

\affil[1]{CNRS, Universit\'e Claude Bernard Lyon 1, INSA Lyon, LIRIS, UMR5205, F-69622 Villeurbanne, France}
\affil[2]{McGill University, Montreal, Canada}
\affil[3]{Massachusetts Institute of Technology, Cambridge, United States}

\date{}

\maketitle

\begin{abstract}
In this paper, we study fair division problems in which resources are structured as graphs and agents must receive connected bundles. This connectivity requirement fundamentally alters the problem, making it significantly more challenging than its classical counterpart. We focus on the fairness notion of EF1\textsubscript{outer}, where envy can be eliminated by removing at most one vertex whose deletion does not disconnect the bundle -- a critical constraint for applications such as land division and network allocation.

Our first result extends prior work by establishing the existence of EF1\textsubscript{outer} allocations for an infinite family of non-traceable graphs (that is, graphs that do not admit a Hamiltonian path), answering a central open question and generalizing Bilò et al.'s result for traceable graphs. We then make progress on a conjecture concerning the EF1\textsubscript{outer} spectrum of trees due to Chen and Zwicker. 
Finally, we complement our structural results with algorithmic insights, showing that deciding the existence of an EF1\textsubscript{outer} allocation is NP-complete even for binary additive valuations, thereby resolving an open complexity question. Taken together, our results deepen the connection between graph theory and fair division, and offer new tools for studying fairness in structured resource environments.
\end{abstract}

\section{Introduction}

Fair division asks how to divide resources among several agents so that each party receives their due share. A central fairness criterion is \emph{envy-freeness} (EF), which requires that no agent prefer another agent's bundle to their own~\cite{Foley1967}; for an overview of fairness notions, see~\cite{Amanatidis2023}. Since EF allocations need not exist with indivisible items, Budish~\cite{Budish2011} introduced \emph{envy-freeness up to one good} (EF1), which requires that any envy can be eliminated by removing at most one item from the envied bundle. This relaxation is always attainable in the standard model of fair division~\cite{Lipton2004}.

Many applications impose feasibility constraints on the allocated bundles, particularly \emph{connectivity}, as in land division and network allocation~\cite{Igarashi2024}. Motivated by these settings, Bouveret et al.~\cite{Bouveret2017} introduced fair division on graphs, where items are represented by vertices and each agent must receive a connected subgraph. A natural adaptation of EF1 to this model is envy-freeness up to one \emph{outer} good (\EFO!), where the removed item may not disconnect the envied bundle~\cite{Bilo2022}. Bilò et al.~\cite{Bilo2022} gave the first structural results for two agents and proved existence on traceable graphs\footnote{A graph is called traceable if it admits a Hamiltonian path.} for up to four agents with monotone preferences. Using topological methods, Igarashi~\cite{Igarashi2023} extended this guarantee to an arbitrary number of agents. Characterizing which graph classes and valuation functions admit \EFO allocations remains an active line of research.

\subsection{Our Contributions}
Whether guarantees similar to the result of Igarashi~\cite{Igarashi2023} persist beyond traceable graphs has remained a central open problem in \EFO graph fair division, even for restricted valuation functions such as common additive valuations. This question is highlighted by Chen and Zwicker~\cite{Chen2024} and Bil{`o} et al.~\cite{Bilo2022}, with the former describing it as ``\emph{the most important open question in the study of \EFO graph fair division}.'' We make progress on this question by showing that traceability is not necessary: there exists an infinite family of non-traceable graphs for which \EFO allocations always exist under common monotone valuations.

\begin{itheorem}[Informal version of Theorem \ref{thm:double_suns}]
    There exists an infinite family of non-traceable graphs for which an \EFO allocation always exists for common monotone valuation functions.
\end{itheorem} 

Our family of graphs is almost traceable, which ensures the existence of several orderings of the vertices of the graph such that almost all consecutive subsets of vertices are connected. We prove that, by wisely choosing this order depending on the valuation function and the number of agents, we can always find an \EFO allocation in polynomial time. 

Chen and Zwicker~\cite{Chen2024} introduced the \EFO spectrum of a graph, whose $k$-th entry is yes precisely when every instance with $k$ agents in the given valuation class admits an \EFO allocation. Chen and Zwicker conjectured that the following holds:

\begin{conjecture} \cite{Chen2024}\label{conj:spectrum}
    The \EFO spectrum of any connected graph must be composed of an initial yes string, followed by a (possibly empty) no string, followed by an infinite string of yes.
\end{conjecture} 

They give some examples to support this conjecture and define a notion of \emph{generalized cutset}, whose presence in a graph ensures the contiguity of the no string. 
In the particular case of traceable graphs, the result of Igarashi~\cite{Igarashi2023} can be restated as follows: the spectrum of traceable graphs is an infinite string of yes. 
The minimum $k$ for which there always exists an \EFO allocation for any number of agents $k' \ge k$ is called the \emph{threshold} of the spectrum\footnote{It is easy to see that this value exists and is always smaller than the number of vertices of the graph.}.
The second main result of this paper lends credence to this conjecture by proving that it holds for common binary additive valuation functions on trees:

\begin{itheorem}[Informal version of Theorem \ref{thm:binary}]
    Conjecture~\ref{conj:spectrum} holds for common binary additive valuation functions on trees.
\end{itheorem}

Our result provides, to the best of our knowledge, the first general step toward this conjecture since the work of \cite{Igarashi2023}. Rather than merely establishing the conjecture for common binary additive valuations, our proof gives a precise characterization of certain subsets of vertices, called bad sets, whose presence ensures the existence of a no instance. A bad set is a set such that the number of components incident to it is at least two more than the size of its border. While it is fairly simple to construct a no instance of the desired size if there is a bad set, the hard part of the proof consists in proving that, if an instance on the tree admits no \EFO allocation, then we can find a bad set of a prescribed size --- ensuring that the lower bound matches the upper bound. Moreover, we complete this result with a polynomial-time dynamic programming algorithm computing the maximum size of a bad set, that is the threshold of the spectrum, for common binary valuations on trees. Bad sets are defined in a purely graph-theoretic manner, highlighting that tools from structural graph theory can play a central role in the study of \EFO allocations. We believe that bad sets constitute a natural tool for tackling this conjecture and may be instrumental in determining the spectrum threshold, well beyond the setting of binary valuations and trees.

The only tree which is traceable is a path. 
As soon as a tree contains a branching node, one can observe that, even for binary additive valuation functions, the \EFO guarantee already fails for three agents. More generally, as the number of leaves grows, one can construct bad instances, that is, instances that do not admit any \EFO allocation even for many agents. This naturally raises the question of whether this intuition can be formalized. We show that it is indeed the case, as the following result demonstrates:

\begin{itheorem}[Informal version of Theorem \ref{thm:badset_leaves}]
    There exists a constant $C$ such that, for every $n$-vertex tree $T$ with $\ell$ leaves, the threshold of the spectrum is at least $(1-C \cdot \frac 1\ell)n$ even for common binary additive valuation functions. 
    
    In other words, there always exist a valuation function with no \EFO allocation for $(1 - \frac C\ell)n-1$ agents.
\end{itheorem}

The proof relies on the fact that every tree with $\ell$ leaves contains at least one of three structures of small linear size, which ensure the existence of a large bad set, allowing us to conclude by our above results.

This result ensures that the right parameter to evaluate the threshold, in the case of trees, is the number of leaves, rather than the number of branching nodes. More broadly, it suggests that, on general graphs, the number of connected components attached to a small separator might provide an approximation of the threshold of the spectrum.

Finally, we complete our contributions with algorithmic results which have been underexplored so far. Igarashi \cite{Igarashi2023} left open the complexity question of deciding the existence of an allocation even for "\emph{simple classes of valuation functions, e.g. binary additive valuations}". We consider the \textsc{Outer Connected Fair Division of Graphs} problem which takes as input a graph, a set of agents and their respective valuation functions, and outputs yes if and only if there exists an \EFO allocation. We answer Igarashi's question by proving that the following holds:

\begin{itheorem}[Informal version of Theorem \ref{thm:np_result}]
    The \textsc{Outer Connected Fair Division of Graphs} problem is NP-complete for common additive valuation functions even restricted to very structured and simple classes of graphs such as: variants of stars, graphs with bounded bandwidth\footnote{Bandwidth is, in particular, a very strong restriction of treewidth. A formal definition will be given later on.}, planar graphs.
\end{itheorem}

\subsection{Related Work}

Fair division has been studied in a wide range of settings, spanning both indivisible and divisible resources and encompassing many fairness notions, with envy-freeness playing a particularly prominent role. 
For general background and recent progress on fair division of indivisible goods, we refer to the surveys and overviews~\cite{Thomson2016,Bouveret2015,Markakis2017,Moulin2019,Amanatidis2023}. 
A broad perspective on \emph{constrained} fair division is given by Suksompong~\cite{Suksompong2021}, who surveys numerous feasibility constraints and highlights connectivity as one of the most frequently studied. 
Related constrained variants include contiguity on a line for indivisible items~\cite{Suksompong2019}, as well as extensions to chores and mixed goods/chores models~\cite{Aziz2019,bouveret2019,Hohne2021,segalhalevi2018}. 
In the continuous setting, classic cake-cutting studies dividing an interval~\cite{Stromquist1980,brams1996,segalhalevi2018}, and more recent work considers connectivity constraints on continuous graphs (``graphical cakes'')~\cite{Bei2021,Igarashi2024}.

The graph model central to this paper, where \emph{items are vertices} of a graph and each agent must receive a connected bundle, was introduced by Bouveret et al.~\cite{Bouveret2017}. 
They focus on additive valuations and study several fairness notions, including envy-freeness and maximin share (MMS), and establish strong hardness results for deciding the existence of complete envy-free allocations even on restricted graph classes. 
Further computational results for MMS on graphs are given by Greco and Scarcello~\cite{Greco2020}, while Bei et al.~\cite{Bei2022} study the \emph{price of connectivity}, quantifying how requiring connected bundles impacts achievable fairness guarantees and identifying optimal relaxations in several settings.

A key challenge under connectivity constraints is that the standard EF1 relaxation may be ill-suited: removing an arbitrary item to eliminate envy can disconnect a bundle and thereby violate feasibility. 
Bil\`{o} et al.~\cite{Bilo2022} address this by introducing envy-freeness up to one \emph{outer} good (\EFO), where the removed vertex must be outer (its removal preserves connectivity). 
They provide structural characterizations for two agents and prove existence guarantees on paths for small numbers of agents using cut-and-choose ideas together with Sperner-type topological arguments. 
Chen and Zwicker~\cite{Chen2024} develop additional structural tools (via cutsets) for understanding EF1-type notions on graphs and motivate questions about extending universal guarantees beyond traceable graphs.

Algorithmic and complexity aspects of connected fair division have been further explored by Deligkas et al.~\cite{Deligkas2021}, who study connected EF1 without the outer restriction (CEF1). 
They prove NP-hardness on restricted graph classes and give XP-time algorithms under structural parameterizations such as clique-width and bounded numbers of agents, using dynamic programming over graph decompositions. 
A complementary motivation for imposing structure comes from limited value access: Oh et al.~\cite{Oh2019} study query complexity for computing approximately fair allocations, providing additional evidence that structured settings can lead to more efficient procedures.

More broadly, graphs are an increasingly common structure in multi-agent systems. 
Beyond item-connectivity constraints, graph structure also enters fair allocation through agent-side constraints, such as distributed allocation models where agents interact or trade only with neighbors~\cite{Chevaleyre2017}, and social-comparison models where fairness is evaluated locally with respect to neighbors~\cite{Abebe2017,Eiben2020}. 
A different but related edge-based viewpoint is envy-free orientations, where agents are vertices and value only incident items; Li et al.~\cite{Li2025} study minimum-subsidy envy-free orientations and tight worst-case subsidy bounds in that framework.

\subsection{Paper Outline}
After introducing preliminary definitions and notations in Section~\ref{sec:prelim}, we present the non-traceable double sun graphs, for which we prove universal guarantees in Section~\ref{sec:double_suns}. Next, we give our main characterization for binary valuations on trees in Section~\ref{sec:trees}. We prove that this threshold can moreover be computed in polynomial time in Section~\ref{sec:polytime}. In Section~\ref{sec:manyleaves}, we show that the threshold increases when the number of leaves increases. Finally, in Section~\ref{sec:complexity} we show that the problem is NP-hard even on restricted classes of graphs.

\section{Preliminaries}
\label{sec:prelim}

An instance of the fair division problem consists of a set $N = [n] = \{1, \dots, n\}$ of \emph{agents} with preferences over $m$ \emph{items}, which we associate with the vertices of a graph $G = (V,E)$. Let $X \subseteq V$ be a subset of the vertices. We denote by $G[X]$ the subgraph of $G$ induced by~$X$, obtained by deleting the vertices not in~$X$ and all edges incident to them. 
%We denote by $G \setminus X$ the graph $G[V \setminus X]$. 
We say that a subset $X$ is connected whenever $G[X]$ is connected. We denote by $\mathcal{C}$ the collection of connected subsets of $V$, which we also refer to as \emph{bundles}. The preferences of each agent $i \in N$ are given by a \emph{valuation function} $v_i: \mathcal{C} \rightarrow \mathbb{R}_{\geq 0}$, which satisfies $v_i(\emptyset) = 0$.
A valuation function $v_i$ is called \emph{monotone} if for all bundles $X,Y \in \mathcal{C}$ with $X \subseteq Y$, we have $v_i(X) \leq v_i(Y)$. It is called \emph{additive} if it satisfies $v_i(X) = \sum_{x \in X} v_i(\{x\})$ for each bundle $X \in \mathcal{C}$ and each agent $i \in N$. In particular, additive valuations are monotone. We say that valuations are \emph{common} if $v_i = v_j$ for all agents $i,j \in N$, and \emph{arbitrary} otherwise. The special class of common additive valuations is of central importance in the fair division literature. 

The problem consists in fairly dividing the items among the agents. A (connected) \emph{allocation} $A = \{A_i\}_{i\in N}$ is a partition of $V$ into $n$ parts, which assigns each agent $i$ a connected bundle $A_i$. An allocation $A$ is \emph{envy-free} (EF) if $v_i(A_i) \geq v_i(A_j)$ for all $i,j \in N$, which states that every agent likes their own bundle at least as much as any other allocated bundle. As such allocations might not exist, Budish \cite{Budish2011} introduced the relaxed notion of \emph{envy-freeness up to one good} (EF1), which has been adapted as follows to take into consideration the connectivity constraints:

\begin{definition}[\cite{Bilo2022}]
    An allocation $A$ is \emph{envy-free up to one outer good} \mbox{(\EFO\!)} if for any agents $i,j \in N$, either $A_j = \emptyset$ or there exists an item $x \in A_j$ such that $A_j \setminus \{x\}$ is connected and $v_i(A_i) \geq v_i(A_j \setminus \{x\})$.
\end{definition}

We end this section with some graph-theoretic definitions. Let $T = (V,E)$ be a tree. We say that~$T$ is a~\emph{proper tree} if~$T$ is not a path. A \emph{branching node} in $T$ is a node of degree at least~3. In particular, every proper tree has at least one branching node.
Let $X \subseteq V$, and $u \in X$. Let $T_X$ be the minimal subtree of~$T$ that contains all the vertices of $X$. 
We say that $u$ is \emph{external} in $X$ if $u$ is a leaf of~$T_X$, or equivalently, if at most one connected component of~$T \setminus \{u\}$ contains vertices of $X$.
We say that $u$ is an \emph{external branching node} if it is a branching node and is external with respect to the set of branching nodes. Note that every branching node $u$ is an external branching node if and only if at most one connected component of $T \setminus \{u\}$ is a proper tree.

\section{Universal Guarantee for Double Suns}
\label{sec:double_suns}
	
	Bil{\`o} et al. proved the following result (which was extended to arbitrary monotone functions by Igarashi~\cite{Igarashi2023}):	
	\begin{theorem}[\cite{Bilo2022}]
		\label{thm:ef1path}
		Let~$P$ be an $n$-vertex path and $k \in \{1, \ldots, n\}$. For any common monotone valuation, there exists an \EFO allocation of the vertices of~$P$ to $k$ agents.
	\end{theorem}

    Theorem~\ref{thm:ef1path} implies in particular that every graph that has a Hamiltonian path satisfies the following universal guarantee: for every number of agents, and any monotone valuation, there exists an \EFO allocation.
    We will prove that graphs that satisfy 
    this universal guarantee do not necessarily have a Hamiltonian path, by providing an infinite family of counterexamples.  This positively answers a question of Chen and Zwicker~\cite{Chen2024}, in the case of common valuations. 
	
    Let $n \geqslant 12$. The \emph{double sun of size~$n$} is a clique of size~$n-6$ on which we glued three $K_{2,2}$ on three disjoint pairs of vertices, for a total of $n$ vertices, as represented on Figure~\ref{fig:double_sun}. In this section, our goal is to show that a double sun does not have a Hamiltonian path, but satisfies the previous universal guarantee.
	
	Let us first introduce some notation. For a double sun of size~$n$, we will denote by $(a_1, a_2)$, $(b_1, b_2)$, $(c_1, c_2)$ the three pairs of vertices on which the copies of $K_{2,2}$ are attached. The vertices in these pairs are called the~\emph{connectors}. We denote by $a_1', a_2', b_1', b_2', c_1', c_2'$ the six additional vertices, that are called the~\emph{tips}. For each $x \in \{a,b,c\}$ and $i \in \{1,2\}$, $x_i'$ is a neighbor of both $x_1$ and~$x_2$. We also say that $x_i'$ is the tip of $x_i$, and that $x_i$ is the connector of~$x_i'$. The set $\{x_i,x_i'\}$ is called a~\emph{ray}, and the set $\{x_1, x_2, x_i'\}$ is called a \textit{beam}. 

	Finally, the vertices that belong to the clique of size~$n-6$ and which are not connectors are called the~\emph{core vertices}.
	The double sun of size~$n$ has thus $n-12$ core vertices, $6$ connectors and~$6$ tips.
	See Figure~\ref{fig:double_sun} for an illustration of a double sun. We begin by proving the non-existence of a Hamiltonian path.
	
	\begin{figure}[h!]
	\centering
	\begin{tikzpicture}[x=0.75pt,y=0.75pt,yscale=-1,xscale=1]
		%uncomment if require: \path (0,455); %set diagram left start at 0, and has height of 455
		
		%Shape: Ellipse [id:dp21929241676847366] 
		\draw  [draw opacity=0][fill={rgb, 255:red, 155; green, 155; blue, 155 }  ,fill opacity=0.6 ] (242,149.04) .. controls (242,120.74) and (265.36,97.81) .. (294.17,97.81) .. controls (322.99,97.81) and (346.35,120.74) .. (346.35,149.04) .. controls (346.35,177.33) and (322.99,200.27) .. (294.17,200.27) .. controls (265.36,200.27) and (242,177.33) .. (242,149.04) -- cycle ;
		%Straight Lines [id:da22316896801393737] 
		\draw    (274.66,228.51) -- (274.66,185.57) ;
		%Straight Lines [id:da07948890162351707] 
		\draw    (274.66,185.57) -- (318.25,228.51) ;
		%Straight Lines [id:da17260610539590215] 
		\draw    (274.66,228.51) -- (318.25,185.57) ;
		%Straight Lines [id:da7242001203628535] 
		\draw    (318.25,228.51) -- (318.25,185.57) ;
		%Shape: Circle [id:dp021226520948690042] 
		\draw  [fill={rgb, 255:red, 255; green, 255; blue, 255 }  ,fill opacity=1 ] (272.04,228.51) .. controls (272.04,227.06) and (273.22,225.89) .. (274.66,225.89) .. controls (276.1,225.89) and (277.28,227.06) .. (277.28,228.51) .. controls (277.28,229.95) and (276.1,231.12) .. (274.66,231.12) .. controls (273.22,231.12) and (272.04,229.95) .. (272.04,228.51) -- cycle ;
		%Shape: Circle [id:dp9846954405471801] 
		\draw  [fill={rgb, 255:red, 255; green, 255; blue, 255 }  ,fill opacity=1 ] (272.04,185.57) .. controls (272.04,184.13) and (273.22,182.96) .. (274.66,182.96) .. controls (276.1,182.96) and (277.28,184.13) .. (277.28,185.57) .. controls (277.28,187.02) and (276.1,188.19) .. (274.66,188.19) .. controls (273.22,188.19) and (272.04,187.02) .. (272.04,185.57) -- cycle ;
		%Shape: Ellipse [id:dp6547704432437755] 
		\draw  [fill={rgb, 255:red, 255; green, 255; blue, 255 }  ,fill opacity=1 ] (315.64,228.51) .. controls (315.64,227.06) and (316.81,225.89) .. (318.25,225.89) .. controls (319.7,225.89) and (320.87,227.06) .. (320.87,228.51) .. controls (320.87,229.95) and (319.7,231.12) .. (318.25,231.12) .. controls (316.81,231.12) and (315.64,229.95) .. (315.64,228.51) -- cycle ;
		%Straight Lines [id:da7727743008973983] 
		\draw    (234.76,88.94) -- (271.94,110.41) ;
		%Straight Lines [id:da7388727636142128] 
		\draw    (271.94,110.41) -- (212.96,126.69) ;
		%Straight Lines [id:da28053158189719674] 
		\draw    (234.76,88.94) -- (250.15,148.16) ;
		%Straight Lines [id:da001982176862341145] 
		\draw    (212.96,126.69) -- (250.15,148.16) ;
		%Shape: Ellipse [id:dp2600960007786549] 
		\draw  [fill={rgb, 255:red, 255; green, 255; blue, 255 }  ,fill opacity=1 ] (236.07,86.67) .. controls (237.32,87.4) and (237.75,89) .. (237.02,90.25) .. controls (236.3,91.5) and (234.7,91.93) .. (233.45,91.2) .. controls (232.2,90.48) and (231.77,88.88) .. (232.49,87.63) .. controls (233.22,86.38) and (234.82,85.95) .. (236.07,86.67) -- cycle ;
		%Shape: Ellipse [id:dp5479955038089053] 
		\draw  [fill={rgb, 255:red, 255; green, 255; blue, 255 }  ,fill opacity=1 ] (251.46,145.89) .. controls (252.71,146.62) and (253.14,148.22) .. (252.41,149.47) .. controls (251.69,150.72) and (250.09,151.15) .. (248.84,150.42) .. controls (247.59,149.7) and (247.16,148.1) .. (247.88,146.85) .. controls (248.61,145.6) and (250.21,145.17) .. (251.46,145.89) -- cycle ;
		%Shape: Ellipse [id:dp689720376819663] 
		\draw  [fill={rgb, 255:red, 255; green, 255; blue, 255 }  ,fill opacity=1 ] (273.25,108.14) .. controls (274.5,108.86) and (274.93,110.46) .. (274.21,111.72) .. controls (273.49,112.97) and (271.89,113.39) .. (270.64,112.67) .. controls (269.39,111.95) and (268.96,110.35) .. (269.68,109.1) .. controls (270.4,107.85) and (272,107.42) .. (273.25,108.14) -- cycle ;
		%Shape: Ellipse [id:dp5626249377851468] 
		\draw  [fill={rgb, 255:red, 255; green, 255; blue, 255 }  ,fill opacity=1 ] (214.27,124.42) .. controls (215.52,125.15) and (215.95,126.75) .. (215.23,128) .. controls (214.51,129.25) and (212.91,129.68) .. (211.66,128.96) .. controls (210.4,128.23) and (209.98,126.63) .. (210.7,125.38) .. controls (211.42,124.13) and (213.02,123.7) .. (214.27,124.42) -- cycle ;
		%Straight Lines [id:da37624329187036165] 
		\draw    (374.38,126.69) -- (337.2,148.16) ;
		%Straight Lines [id:da8863145353080276] 
		\draw    (337.2,148.16) -- (352.59,88.94) ;
		%Straight Lines [id:da02594990804790842] 
		\draw    (374.38,126.69) -- (315.4,110.41) ;
		%Straight Lines [id:da9811309136287922] 
		\draw    (352.59,88.94) -- (315.4,110.41) ;
		%Shape: Ellipse [id:dp5459426071189882] 
		\draw  [fill={rgb, 255:red, 255; green, 255; blue, 255 }  ,fill opacity=1 ] (375.69,128.96) .. controls (374.44,129.68) and (372.84,129.25) .. (372.12,128) .. controls (371.4,126.75) and (371.83,125.15) .. (373.08,124.42) .. controls (374.33,123.7) and (375.93,124.13) .. (376.65,125.38) .. controls (377.37,126.63) and (376.94,128.23) .. (375.69,128.96) -- cycle ;
		%Shape: Circle [id:dp9392717939469599] 
		\draw  [fill={rgb, 255:red, 255; green, 255; blue, 255 }  ,fill opacity=1 ] (316.71,112.67) .. controls (315.46,113.39) and (313.86,112.97) .. (313.14,111.72) .. controls (312.42,110.46) and (312.84,108.86) .. (314.1,108.14) .. controls (315.35,107.42) and (316.95,107.85) .. (317.67,109.1) .. controls (318.39,110.35) and (317.96,111.95) .. (316.71,112.67) -- cycle ;
		%Shape: Ellipse [id:dp0036962582315527293] 
		\draw  [fill={rgb, 255:red, 255; green, 255; blue, 255 }  ,fill opacity=1 ] (338.51,150.42) .. controls (337.26,151.15) and (335.66,150.72) .. (334.93,149.47) .. controls (334.21,148.22) and (334.64,146.62) .. (335.89,145.89) .. controls (337.14,145.17) and (338.74,145.6) .. (339.46,146.85) .. controls (340.19,148.1) and (339.76,149.7) .. (338.51,150.42) -- cycle ;
		%Shape: Ellipse [id:dp19703727641574997] 
		\draw  [fill={rgb, 255:red, 255; green, 255; blue, 255 }  ,fill opacity=1 ] (353.9,91.2) .. controls (352.65,91.93) and (351.05,91.5) .. (350.32,90.25) .. controls (349.6,89) and (350.03,87.4) .. (351.28,86.67) .. controls (352.53,85.95) and (354.13,86.38) .. (354.85,87.63) .. controls (355.58,88.88) and (355.15,90.48) .. (353.9,91.2) -- cycle ;
		%Shape: Ellipse [id:dp8762268263406896] 
		\draw  [fill={rgb, 255:red, 255; green, 255; blue, 255 }  ,fill opacity=1 ] (319.56,183.31) .. controls (320.81,184.03) and (321.24,185.63) .. (320.52,186.88) .. controls (319.79,188.13) and (318.19,188.56) .. (316.94,187.84) .. controls (315.69,187.11) and (315.26,185.51) .. (315.99,184.26) .. controls (316.71,183.01) and (318.31,182.58) .. (319.56,183.31) -- cycle ;
		%Shape: Ellipse [id:dp16561792140591458] 
		\draw  [fill={rgb, 255:red, 255; green, 255; blue, 255 }  ,fill opacity=1 ] (296.42,146.77) .. controls (297.67,147.5) and (298.1,149.1) .. (297.38,150.35) .. controls (296.66,151.6) and (295.06,152.03) .. (293.81,151.3) .. controls (292.56,150.58) and (292.13,148.98) .. (292.85,147.73) .. controls (293.57,146.48) and (295.17,146.05) .. (296.42,146.77) -- cycle ;

		{\large
		% Text Node
		\draw (284,74.75) node [anchor=north west][inner sep=0.75pt]    {$K_{7}$};
		}
		
		{\small
		% Text Node
		\draw (276.66,231.91) node [anchor=north west][inner sep=0.75pt]    {$a'_{1}$};
		% Text Node
		\draw (320.25,231.91) node [anchor=north west][inner sep=0.75pt]    {$a'_{2}$};
		% Text Node
		\draw (378.65,119.78) node [anchor=north west][inner sep=0.75pt]    {$b'_{1}$};
		% Text Node
		\draw (357.85,81.03) node [anchor=north west][inner sep=0.75pt]    {$b'_{2}$};
		% Text Node
		\draw (214.66,80.91) node [anchor=north west][inner sep=0.75pt]    {$c'_{1}$};
		% Text Node
		\draw (193.66,119.91) node [anchor=north west][inner sep=0.75pt]    {$c'_{2}$};
		% Text Node
		\draw (267.66,170.91) node [anchor=north west][inner sep=0.75pt]    {$a_{1}$};
		% Text Node
		\draw (310.66,170.91) node [anchor=north west][inner sep=0.75pt]    {$a_{2}$};
		% Text Node
		\draw (253.46,146.29) node [anchor=north west][inner sep=0.75pt]    {$c_{2}$};
		% Text Node
		\draw (272.94,112.81) node [anchor=north west][inner sep=0.75pt]    {$c_{1}$};
		% Text Node
		\draw (304.94,112.81) node [anchor=north west][inner sep=0.75pt]    {$b_{2}$};
		% Text Node
		\draw (318.94,141.81) node [anchor=north west][inner sep=0.75pt]    {$b_{1}$};
		}

	\end{tikzpicture}
	\caption{The double sun of size~$13$. The gray circle represents a clique: all the~$7$ vertices inside this area are pairwise linked by an edge. There is a single core vertex in this clique.}
	\label{fig:double_sun}
	\end{figure}
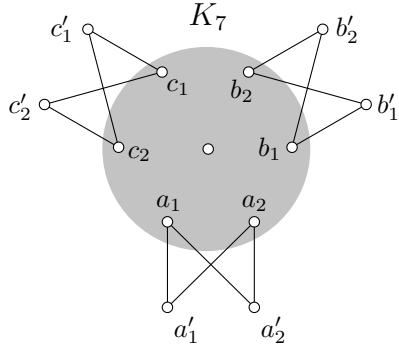

	\begin{lemma}
		Let $n \geqslant 12$. The double sun of size~$n$ does not have a Hamiltonian path.
	\end{lemma}
	
	\begin{proof}
		Assume by contradiction that the double sun of size~$n$ has a Hamiltonian path. Since the vertex $a_1'$ has degree~two, if it is not an endpoint of the path, it should be between its two neighbors in the path, that is, between $a_1$ and~$a_2$. This is similar for~$a_2'$. Since $a_1'$ and $a_2'$ can not be both between $a_1$ and~$a_2$, at least one vertex among~$\{a_1',a_2'\}$ is an endpoint of the Hamiltonian path. Similarly, at least one vertex among $\{b_1',b_2'\}$ should be an endpoint, and at least one vertex among $\{c_1',c_2'\}$ should also be an endpoint. This is a contradiction, since there are only two endpoints.
	\end{proof}
	
	Now, let us state and prove our main result on double suns:
    \begin{theorem}
        \label{thm:double_suns}
        For every $n \geqslant 12$, every $k \in \{1, \ldots, n\}$, and every common monotone valuation, there exists an \EFO allocation of the double sun of size $n$ for~$k$ agents. 
    \end{theorem}
    
   %We give a sketch of the proof below, and a complete proof can be found in Appendix~\ref{appendix:double_suns}. 

\begin{proof}
Let $n \geqslant 12$, and $G = (V,E)$ be a double sun of size~$n$. Let us consider a common monotone valuation, denoted by~$v$, and for every subset of vertices $S \subseteq V$, let $v(S)$ denote the valuation of~$S$. We proceed by distinguishing cases based on $k$.

    \subsection{The case $k \in \{1,2,3\}$}
	
	If $k = 1$, the result is trivial.
	If $k \in \{2,3\}$, let us consider the following order of the vertices, where $(\ast)$ denotes all the core vertices inserted in an arbitrary order:
	\begin{equation}
		\label{eq:order}
		a_1', a_1, a_2', a_2, (\ast), b_1, b_1', b_2', b_2, c_1, c_1', c_2, c_2'
	\end{equation}
	
	Let us remark that this order is not a Hamiltonian path in~$G$ (because there is a missing edge between~$b_1'$ and~$b_2'$), but it satisfies the following property~$\mathcal{P}$: every set formed by consecutive vertices in the order~(\ref{eq:order}) is connected in~$G$, except the set~$\{b_1',b_2'\}$.

	\paragraph*{The case $k = 2$}
	Assume that $k = 2$. Let us consider the path $P$ on the same set of vertices~$V$, in the order~(\ref{eq:order}),  with the same valuation function. By Theorem~\ref{thm:ef1path}, there exists an \EFO allocation of the vertices of~$P$ to two agents. Let us consider this allocation, and let us prove that it is also an \EFO allocation for the double sun. Let us denote by~$X$, $Y$ the two sets of vertices of this allocation. First, since $X$ and~$Y$ form a partition of~$P$ into two connected subgraphs, these two sets are a prefix and a suffix of the order~(\ref{eq:order}). In particular, by property~$\mathcal{P}$, these two sets induce connected subgraphs in~$G$. By symmetry, assume that $v(X) \leqslant v(Y)$. Since this allocation is \EFO in~$P$, there exists a vertex~$u \in Y$ such that $v(Y \setminus \{u\}) \leqslant v(X)$, and $Y \setminus \{u\}$ is connected in~$P$. Thus, $Y \setminus \{u\}$ is a set of consecutive vertices in the order~(\ref{eq:order}). Since $Y$ is a prefix or a suffix, $Y \setminus \{u\}$ is not equal to $\{b_1',b_2'\}$, so by property~$\mathcal{P}$, $Y \setminus \{u\}$ is connected in~$G$. This proves that this allocation is also \EFO in~$G$.

	\paragraph*{The case $k = 3$}
	Assume that $k=3$. For every $x \in \{a,b,c\}$, let $v_x := v(\{x_1, x_1', x_2, x_2'\})$.
	By symmetry, assume that $\min(v_a,v_b,v_c) = v_b$. Consider again the order~(\ref{eq:order}), and the path~$P$ on the same set of vertices, in this order, and the same valuation function.
	By Theorem~\ref{thm:ef1path}, there exists an \EFO allocation of the vertices of~$P$ to three agents. Let us denote by $X,Y,Z$ the three sets of vertices of this allocation. Since these sets are connected in~$P$, each one is a set of consecutive vertices in the order~(\ref{eq:order}). Thus, we can assume that~$X$ is a prefix, $Z$ is a suffix, and~$Y$ is a set of consecutive vertices between~$X$ and~$Z$.
	By property~$\mathcal{P}$, $X$ and~$Z$ are connected in~$G$. Moreover, if~$u$ is the first or the last vertex of $X$ in~$P$, then $X \setminus \{u\}$ is still connected in~$G$, and similarly for~$Z$.
	There now are two cases.
	
	\begin{itemize}
		\item Assume first that~$Y$ satisfies one of the two following conditions:
		\begin{itemize}
			\item $Y$ does not contain both~$b_1'$ and~$b_2'$, or
			\item $Y$ contains both~$b_1'$ and~$b_2'$, and $|Y| \geqslant 4$.
		\end{itemize}
		Then, by property~$\mathcal{P}$, $Y$ is connected in~$G$. Moreover, if for example $v(X) \leqslant v(Y)$, since the allocation is \EFO in~$P$, there exists $u \in Y$ such that $Y \setminus u$ is connected in~$P$ and $v(Y \setminus \{u\}) \leqslant v(X)$. The set $Y \setminus \{u\}$ is a set of consecutive vertices in the order~(\ref{eq:order}), different from~$\{b_1',b_2'\}$ by the assumption made on~$Y$. Thus, by property~$\mathcal{P}$, $Y \setminus \{u\}$ is connected in~$G$. So this allocation is \EFO in~$G$.
		
		\item Assume now that $|Y| \leqslant 3$, and that it contains both~$b_1'$ and~$b_2'$. Then, either $Y = \{b_1',b_2'\}$, or there exists $i \in \{1,2\}$ such that $Y = \{b_i,b_1',b_2'\}$. Consider the following allocation.
        Let $Y' := \{b_1,b_1',b_2',b_2\}$, and $X'$ (resp. $Z'$) be the set of all the vertices before~$b_1$ (resp. after~$b_2$) in the order~(\ref{eq:order}). The sets $X',Y',Z'$ still contain consecutive vertices in the order~(\ref{eq:order}), and by property~$\mathcal{P}$, each of these sets is connected in~$G$.
		Since~$v$ is monotone, we also have~$v(X') \leqslant v(X)$, $v(Z') \leqslant v(Z)$ and $v(Y') \geqslant v(Y)$.
        Moreover, we assumed that $v_b \leqslant v_a$, and we have $v(Y') = v_b$ and $\{a_1,a_1',a_2,a_2'\} \subseteq X'$. This implies, by monotonicity, that $v(Y') \leqslant v(X')$. Similarly, we have $v(Y') \leqslant v(Z')$. By symmetry between $X'$ and $Z'$, we can assume without loss of generality that $v(Y') \leqslant v(Z') \leqslant v(X')$. Let us prove that there is no envy between $X', Y', Z'$.
		
		Since the allocation $X,Y,Z$ is \EFO in~$P$, there exists $u \in X$ such that $X\setminus\{u\}$ is connected in~$P$, and $v(X \setminus \{u\}) \leqslant v(Y)$. If $b_1 \in X$ and $u = b_1$, then $X \setminus \{u\} = X'$, so we get $v(X') \leqslant v(Y) \leqslant v(Y')$. Otherwise, we get $v(X' \setminus \{u\}) \leqslant v(X \setminus \{u\}) \leqslant v(Y) \leqslant v(Y')$, and $X' \setminus \{u\}$ is still connected in~$P$, so it is also connected in~$G$ by property~$\mathcal{P}$. Thus, $Y'$ does not envy~$X'$.
		Similarly, $Y'$ does not envy $Z'$.
		Finally, since $v(Z') \geqslant v(Y')$ and since $Y'$ does not envy $X'$, then $Z'$ does not envy $X'$ neither. Thus, the allocation $X',Y',Z'$ is \EFO.
	\end{itemize}

	\subsection{The case $k \geqslant 4$}
	
	If $k \geqslant 4$, we construct algorithmically an \EFO allocation, with the following procedure which has $n$ steps. For convenience, we assume that each vertex has a unique identifier (this will be useful only to break the symmetry in our procedure in the case where two sets of vertices have the same value).
	At each step of the procedure, we have $k$ disjoint sets, and some vertices that do not belong to any set, that we call \emph{unassigned vertices}. Initially, each set is empty. Each step of the procedure consists in adding an unassigned vertex to one of the~$k$ sets. Thus, after~$n$ steps, there is no unassigned vertex remaining, and the~$k$ sets form a partition of vertices (and we will prove that this partition is an \EFO allocation).
	
	If $X, Y$ are two disjoint and non-empty sets, we say that~\emph{$X$ has higher priority than $Y$} if one of the two following conditions is satisfied: 
	\begin{itemize}
		\item $v(X) < v(Y)$, or
		\item $v(X) = v(Y)$ and the minimum identifier of a vertex in~$X$ is smaller than the minimum identifier of a vertex in~$Y$.
	\end{itemize}
	Note that the relation of having higher priority is transitive. Moreover, two disjoint non-empty sets are always comparable for this relation, and any collection of disjoint non-empty sets has a set with highest priority and a set with lowest priority.
	
	For every \mbox{$s \in \{1, \ldots, n\}$}, we denote by~$U_s$ the set of unassigned vertices at the beginning of step~$s$. We have $|U_s| = n-s+1$. For a set~$X$, we say that $X$ is~\emph{critical} if it is a singleton, and the vertex it contains is a tip.
	Step~$s$ consists in the following.
	\begin{enumerate}[label=(\roman*)]
		\item If at least one set is empty, we select an arbitrary empty set. Else, we select the set with highest priority. \label{item:i}
		
		\item We add a carefully chosen vertex in $U_s$ to the set we selected. \label{item:ii}
	\end{enumerate}
	
	If, at each step, we can maintain the connectivity of each set, then the allocation we obtain is \EFO. Indeed, let us consider two sets $X,Y$ at the end of the procedure such that \mbox{$v(X) > v(Y)$}. If~$u$ is the last vertex which was added to~$X$, removing $u$ from~$X$ leaves a connected set, because the connectivity of each set was maintained through the procedure. Moreover, we have \mbox{$v(X \setminus \{u\}) \leqslant v(Y)$} because the fact that $u$ was added to~$X \setminus \{u\}$ at some step implies that $X \setminus \{u\}$ was the set with minimum value at this step.
	In the following, we detail the two cases $k \geqslant 6$ and $k \in \{4,5\}$ (in this order, by increasing technicality).
	What distinguishes these three cases is~\ref{item:ii}, that is, the way we chose the unassigned vertex to add to the set with highest priority at each step of the procedure.

	\paragraph*{The case $k \geqslant 6$}
	
	Assume that $k \geqslant 6$.
	Let $s \in \{1, \ldots, n\}$. Let us denote by~$X$ the set with highest priority that was selected in~\ref{item:i}, and let us describe how we chose the unassigned vertex to add to~$X$ in~\ref{item:ii} at step~$s$. See Figure~\ref{fig:doublesun_k=6} for an illustration with~$k=6$.

	\begin{enumerate}
		\item If $X$ is a critical set, then~$X$ contains only one vertex which is a tip $x_i'$ for some $x \in \{a,b,c\}$ and $i \in \{1,2\}$. Then, we add its connector~$x_i$ to~$X$ (we will prove thereafter that it is indeed possible because its connector $x_i$ is unassigned). \label{item:1_k>=6}
		
		\item If $X$ is not critical, there are some sub-cases.
		
		\begin{enumerate}
			\item If there is a tip in~$U_s$, we add an arbitrary tip to~$X$. \label{item:a_k>=6}
			
			\item Else, if there is a core vertex in~$U_s$, we add an arbitrary core vertex to~$X$. \label{item:b_k>=6}
			
			\item Else, $U_s$ contains only connectors, so $U_s$ is a subset of~$\{a_1,a_2,b_1,b_2,c_1,c_2\}$. Let $x \in \{a,b,c\}$ and $i \in \{1,2\}$ such that $x_i \in U_s$ and $\{x_i'\}$ is the tip with the lowest priority among them. We add~$x_i$ to~$X$. \label{item:c_k>=6}
		\end{enumerate}
	\end{enumerate}

	\begin{figure}[h!]
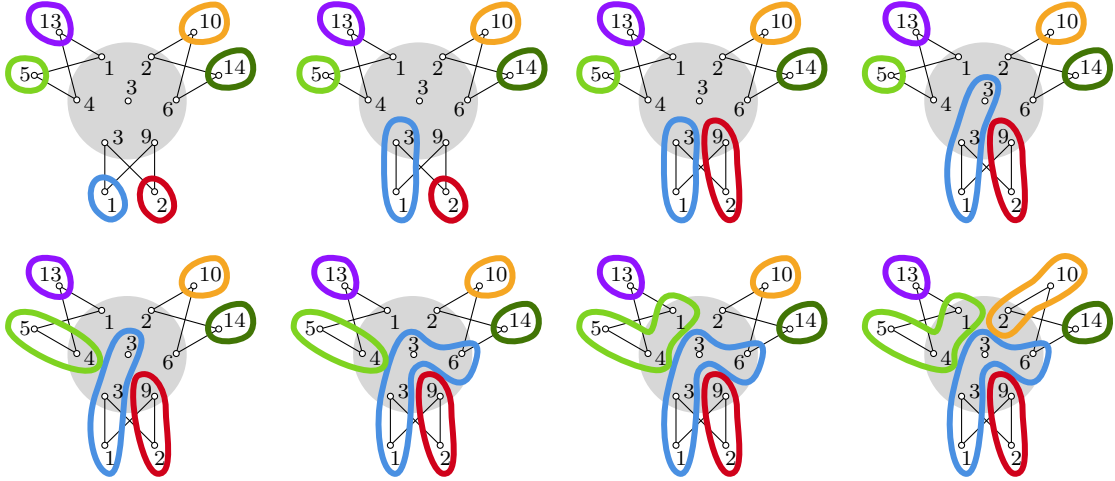

		\centering
		

		
		\caption{Illustration of our procedure for $k=6$ with a double sun of size~$13$. Only the last steps are represented here (the~$6$ first steps consist in adding the tips to the~$6$ initially empty sets). In this example, the valuation is additive, so only the value of each vertex is written. Unique identifiers are not represented to simplify the figure. At the end of the procedure, we obtain an \EFO allocation for~$6$ agents.}
		\label{fig:doublesun_k=6}
		
	\end{figure}

	It remains to prove that the connectivity of each set is maintained at each step, and that in case~\ref{item:1_k>=6}, the connector is indeed unassigned.
	
	Let us first prove that the connectivity is maintained. In Case~\ref{item:1_k>=6}, since~$X$ contains only a tip at the beginning of step~$s$ and we add its connector, then~$X$ is still connected after step~$s$. Note that, at the beginning of the procedure, all the sets are empty, and while at least one set is empty among the~$k$ sets, our procedure selects an empty set in~\ref{item:i}. Since an empty set is not critical, and there are more sets than tips because $k \geqslant 6$, the six first steps consists in adding a tip to an empty set. Thus, in Case~\ref{item:a_k>=6}, the connectivity is maintained, since all tips are added to empty sets. Finally, in Cases~\ref{item:b_k>=6} and~\ref{item:c_k>=6}, the connectivity is also maintained, because $X$ is not critical at the beginning of step~$s$, so it is either empty, or it contains a core vertex or a connector, which is a neighbor of all other core vertices and connectors.
	
	Let us finally prove that in Case~\ref{item:1_k>=6}, the connector is indeed unassigned.
	Assume for the sake of contradiction that Case~\ref{item:1_k>=6} is applied at step~$s$, and that $x_i$ was already assigned. The vertex $x_i$ was thus added to some set at step~$s' < s$, and Case~\ref{item:c_k>=6} was applied at step~$s'$. At the beginning of step~$s'$, all core vertices were already assigned, so $U_{s'}$ contains only connectors. For every $y \in \{a,b,c\}$ and $j \in \{1,2\}$ such that $y_j \in U_{s'}$, one of the~$k$ sets at the beginning of step~$s'$ is the critical set $\{y_j'\}$.
	Let us denote by $\ell$ the size of~$U_{s'}$. At the end of step~$s'$, there are $\ell-1$ remaining steps (because there are $\ell-1$ unassigned vertices remaining). Moreover, at step~$s'$, since $x_i$ was the vertex selected in~\ref{item:c_k>=6}, it implies that $\{x_i'\}$ has the lowest priority among the tips such that $x_i \in U_s$.
	Thus, at the end of step~$s'$, there are $\ell-1$ critical sets that have higher priority than~$\{x_i'\}$, and since there are only~$\ell-1$ steps remaining, the set~$\{x_i'\}$ will never be selected by our procedure in~\ref{item:i}. This is a contradiction since it is selected at step~$s$. Thus, $x_i$ was indeed unassigned at step~$s$.
	This concludes the proof in the case~$k \geqslant 6$.

	\paragraph*{The case $k \in \{4,5\}$}
	
	Assume that $k \in \{4,5\}$.
	By symmetry, assume that the ray with highest priority among all rays is $\{a_1, a_1'\}$.
	Again by symmetry, assume that the ray with highest priority among rays of the form~$\{x_i,x_i'\}$ with $x \in \{b,c\}$ and $i \in \{1,2\}$ is $\{b_1, b_1'\}$. 
	We partition tips into two sets, \emph{starting tips} and \emph{delayed tips}.
	
	\begin{itemize}
		\item If $k = 5$, the set of starting tips is $\{a_1',b_1',b_2',c_1',c_2'\}$, and the set of delayed tips is $\{a_2'\}$.
		
		\item If $k = 4$, the set of starting tips is $\{a_1',b_1',c_1',c_2'\}$, and the set of delayed tips is $\{a_2',b_2'\}$.
	\end{itemize}
	
	Let $s \in \{1, \ldots, n\}$. Let us denote by~$X$ the set with highest priority that was selected in~\ref{item:i}, and let us describe how we chose the unassigned vertex to add to~$X$ in~\ref{item:ii} at step~$s$. See Figure~\ref{fig:doublesun_k=4} for an example. 
	
	\begin{enumerate}
		\item If $X$ is a critical set, then~$X$ contains only one vertex, which is a tip $x_i'$ for some $x \in \{a,b,c\}$ and $i \in \{1,2\}$. We add its connector~$x_i$ to~$X$ (we will prove thereafter that it is indeed possible because its connector $x_i$ is unassigned). \label{item:1_k<=5}
		
		\item If $X$ is not critical, there are some sub-cases. \label{item:2_k<=5}
		
		\begin{enumerate}
			\item If there is a starting tip in~$U_s$, we add an arbitrary starting tip to~$X$. \label{item:a_k<=5}
			
			\item Else, if there exists a delayed tip in $U_s$ adjacent to some vertex of~$X$, we add this delayed tip to~$X$. \label{item:b_k<=5}
			
			\item Else, if there is a connector of a delayed tip in~$U_s$, we add an arbitrary connector of a delayed tip to~$X$. \label{item:c_k<=5}
			
			\item Else, if there is a core vertex in~$U_s$, we add an arbitrary core vertex to~$X$. \label{item:d_k<=5}
			
			\item Else, $U_s$ contains only connectors of starting tips (we will prove this afterwards). Let $x \in \{a,b,c\}$ and $i \in \{1,2\}$ such that $x_i \in U_s$ and $\{x_i'\}$ is the starting tip with the lowest priority among them. We add~$x_i$ to~$X$. \label{item:e_k<=5}
		\end{enumerate}
	\end{enumerate}
	
	\begin{figure}[h!]
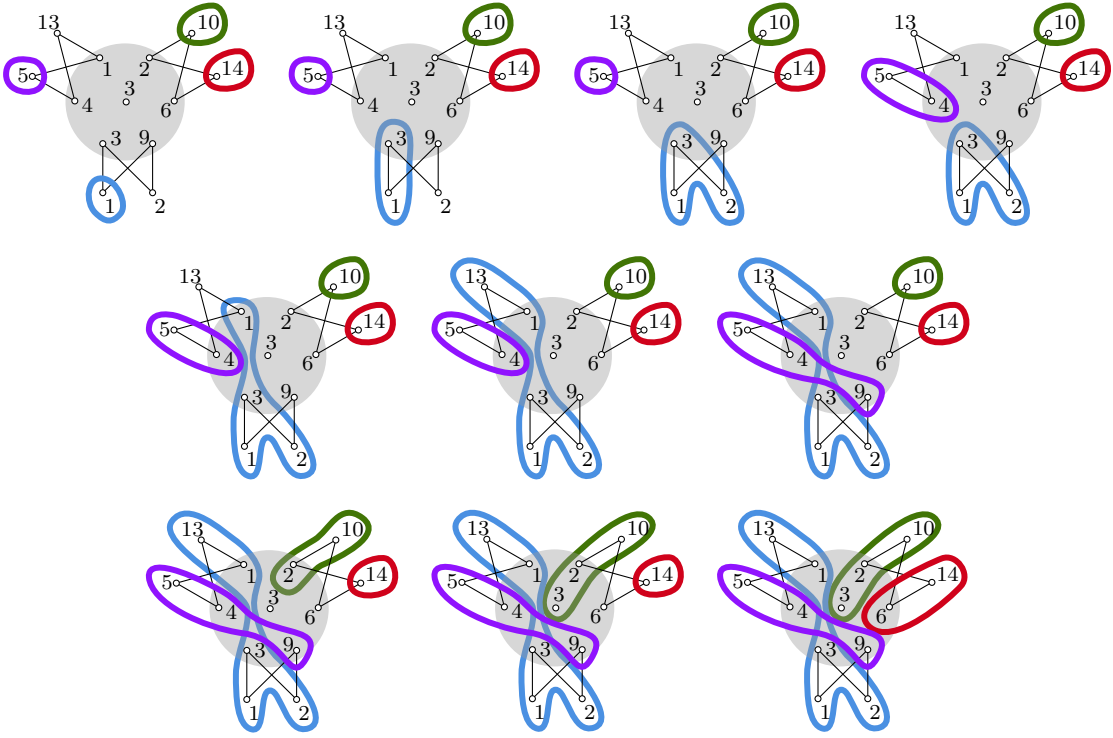

		\centering
		

		\caption{Illustration of our procedure for $k=4$ with a double sun of size~$13$. The starting tips are the tips with value $1,5,10,14$, and the delayed tips are the tips with value $2,13$. Only the last steps are represented here (the~$4$ first steps consist in adding the starting tips to the~$4$ initially empty sets). In this example, the valuation is additive, so only the value of each vertex is written. Unique identifiers are not represented to simplify the figure. At the end of the procedure, we obtain an \EFO allocation for~$4$ agents.}
		\label{fig:doublesun_k=4}
	\end{figure}
	
	It remains to prove that the connectivity of each set is maintained at each step, that in Case~\ref{item:1_k<=5}, the connector is indeed unassigned, and finally, that in Case~\ref{item:e_k<=5}, $U_s$ contains only connectors of starting tips. The proofs of the two first items (namely, the connectivity of each set, and the fact that the connector is unassigned in case~\ref{item:1_k<=5}) are exactly the same as for the case~$k \geqslant 6$, so we detail only the proof of the third item.
	So let us prove that, in case~\ref{item:e_k<=5}, $U_s$ contains only connectors of starting tips. Assume that case~\ref{item:e_k<=5} was applied at step~$s$. Since the previous cases were not applied, $U_s$ do not contain neither starting tips, nor connectors of delayed tips, nor core vertices. It could possibly contain a delayed tip non-adjacent to~$X$, but we will prove that it not the case.
	
	The~$k$ first steps of our procedure consist in adding a starting tip to an empty set. Thus, as soon as a non-empty set is selected, all starting tips are assigned. Moreover, at the beginning of step~$k+1$, all sets are critical, and each contains a starting tip. For every $s \geqslant k+1$, and for every $x \in \{a,b,c\}$, $i \in \{1,2\}$ such that $x_i'$ is a starting tip, we denote by $X_{s}^{x_i'}$ the set that contains the tip $x_i'$ at the beginning of step~$s$.
	
	Let us consider the delayed tip~$a_2'$. Its connector $a_2$ was assigned to some set before step~$s$, at a step for which case~\ref{item:c_k<=5} was applied. Let $s'<s$ be the first step in which case~\ref{item:2_k<=5} was applied and the set~$X$ selected at step~$s'$ is not empty.
	This set~$X$ is thus a ray, and since the ray with higher priority is $\{a_1,a_1'\}$, we have $X = X_{s'}^{a_1'} = \{a_1,a_1'\}$.  At step~$s'$, case~\ref{item:a_k<=5} is not applied (because all starting tips are already assigned), so case~\ref{item:b_k<=5} is applied. Thus, the delayed tip~$a_2'$ is added to~$X$ at step~$s'$, so $a_2' \notin U_s$.
	
	Finally, if $k=4$, let us consider the delayed tip~$b_2'$. Assume by contradiction that $b_2' \in U_s$. First, let us prove that $X_{s}^{b_1'} \subseteq \{b_1,b_1'\}$. We have $X_{k+1}^{b_1'} = \{b_1'\}$. If the set $\{b_1'\}$ is selected by our procedure in~\ref{item:i} at some step to receive a second vertex, case~\ref{item:1_k<=5} is applied, so it receives its connector~$b_1$ and this set becomes~$\{b_1,b_1'\}$. If this set is selected again in~\ref{item:i} to receive a third vertex before step~$s$, since $b_2'$ is unassigned, case~\ref{item:b_k<=5} is applied and it receives $b_2'$ (which is adjacent to~$b_1$). This is not the case, since by assumption $b_2' \in U_s$, so we indeed have $X_{s}^{b_1'} \subseteq \{b_1,b_1'\}$. Since $\{b_1, b_1'\}$ has higher priority than $\{c_1, c_1'\}$ and $\{c_2, c_2'\}$, we also have  $X_{s}^{c_i'} \subseteq \{c_i,c_i'\}$ for each $i \in \{1,2\}$ (indeed, if $X_{s'}^{c_i'}$ is selected to receive a third vertex in~\ref{item:i} at some step $s'<s$, then $X_{s''}^{b_1'}$ should have been selected at an earlier step $s''<s'$ to receive a third vertex, and this is not the case).
	Since $U_s$ do not contain the connector $b_2$, then $b_2$ was previously assigned to some set. Since $|X_{s}^{x_i'}| \subseteq \{x_i,x_i'\}$ for every $x_i \in \{b_1,c_1,c_2\}$, then $b_2 \in X_{s}^{a_1'}$.
	Since case~\ref{item:2_k<=5} was applied at step~$s$, then the set selected at step~$s$ has size at least~$2$. Moreover, again since $\{b_1, b_1'\}$ has higher priority than $\{c_1, c_1'\}$ and $\{c_2, c_2'\}$, the set selected at step~$s$ can not be neither $X_{s}^{c_1'}$ nor $X_{s}^{c_2'}$. Thus, it is either $X_{s}^{a_1'}$ or $X_{s}^{b_1'}$. If it $X_{s}^{a_1'}$, then case~\ref{item:b_k<=5} should have been applied instead of case~\ref{item:e_k<=5}, because $b_2 \in X_{s}^{a_1'}$ and $b_2'$ is adjacent to~$b_2$, a contradiction. If it is $X_{s}^{b_1'}$, then since the set selected in~\ref{item:i} at step~$s$ has size at least~$2$, we have $X_{s}^{b_1'} = \{b_1',b_1\}$. It is again a contradiction, because case~\ref{item:b_k<=5} should have been applied instead of case~\ref{item:e_k<=5}, since $b_1 \in X_{s}^{b_1'}$ and $b_2'$ is adjacent to~$b_1$.
	This concludes the proof. 
\end{proof}

%\vspace{2cm}
%\textcolor{blue}{Remark: this proof shows that if a graph on~$a+b$ vertices has a clique of size $a$, and $b$ additional vertices with $b \leqslant a$ that can be matched to some vertices of the clique, the answer is always positive for $k=b$ (same proof as for $k=6$ previously).}

\section{Spectrum of Trees with Binary Valuations}
\label{sec:trees}
In Section~\ref{sec:double_suns} we established that traceability is not necessary for an all-yes spectrum: there exist non-traceable graphs which admit an \EFO allocation for common monotone valuations for every $k$. We now study the onset of no-instances on trees, an important class of non-traceable graphs, which generalize paths, the focus of most prior work.
We introduce good and bad sets—purely graph-theoretic obstructions that will drive both the characterization and the polynomial-time computation of the spectrum threshold.

Let $G$ be a graph, and $X \subseteq V(G)$. Let $\mathcal{C}_X$ be the set of connected components of $G \setminus X$. For every $C \in \mathcal{C}_X$, the \emph{boundary} of $C$ is defined as the subset of~$X$ containing all the vertices of~$X$ adjacent to at least one vertex of~$C$. We say that $C \in \mathcal{C}_X$ is a \emph{good component} if there is an injective function $f : \mathcal{C}_X \setminus C \to X$ such that, for every $C' \in \mathcal{C}_X \setminus C$, $f(C')$ belongs to the boundary of~$C'$ (in this case, we say that $f$ is the \emph{matching function} for $C$ and \emph{matches} $C'$ to $f(C')$ for every $C' \in \mathcal{C}_X \setminus C$). Finally, we say that $X$ is a \emph{good set} for $G$ if $\mathcal{C}_X$ contains a good component, and a \emph{bad set} otherwise.  See Figure~\ref{fig:example_good_component} and Figure~\ref{fig:example_bad_set} for examples of a good component and a bad set.

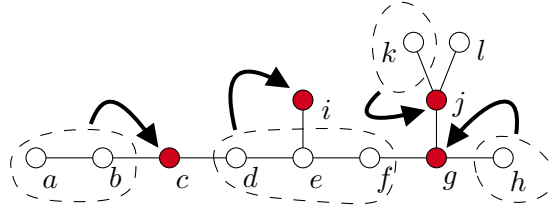
\begin{figure}[h!]
    \centering
    \begin{tikzpicture}[x=0.75pt,y=0.75pt,yscale=-1,xscale=1]

%Straight Lines [id:da8043150739980819] 
\draw    (388,56) -- (376.5,27) ;
%Straight Lines [id:da3782239463714072] 
\draw    (388,56) -- (399.5,27) ;
%Straight Lines [id:da3153434300280563] 
\draw    (187,85) -- (220.5,85) ;
%Straight Lines [id:da8263634859194179] 
\draw    (220.5,85) -- (254,85) ;
%Straight Lines [id:da6841629172908128] 
\draw    (254,85) -- (287.5,85) ;
%Straight Lines [id:da2500160982119267] 
\draw    (287.5,85) -- (321,85) ;
%Straight Lines [id:da3373777248163913] 
\draw    (321,85) -- (354.5,85) ;
%Straight Lines [id:da38962727122151775] 
\draw    (354.5,85) -- (388,85) ;
%Straight Lines [id:da10864763457286819] 
\draw    (388,85) -- (421.5,85) ;
%Straight Lines [id:da33192928128839416] 
\draw    (321,85) -- (321,56) ;
%Straight Lines [id:da9616805020039595] 
\draw    (388,85) -- (388,56) ;
%Shape: Circle [id:dp23129232623807305] 
\draw  [fill={rgb, 255:red, 255; green, 255; blue, 255 }  ,fill opacity=1 ] (182,85) .. controls (182,82.24) and (184.24,80) .. (187,80) .. controls (189.76,80) and (192,82.24) .. (192,85) .. controls (192,87.76) and (189.76,90) .. (187,90) .. controls (184.24,90) and (182,87.76) .. (182,85) -- cycle ;
%Shape: Circle [id:dp833356915964317] 
\draw  [fill={rgb, 255:red, 255; green, 255; blue, 255 }  ,fill opacity=1 ] (215.5,85) .. controls (215.5,82.24) and (217.74,80) .. (220.5,80) .. controls (223.26,80) and (225.5,82.24) .. (225.5,85) .. controls (225.5,87.76) and (223.26,90) .. (220.5,90) .. controls (217.74,90) and (215.5,87.76) .. (215.5,85) -- cycle ;
%Shape: Circle [id:dp48059563893673196] 
\draw  [fill={rgb, 255:red, 208; green, 2; blue, 27 }  ,fill opacity=1 ] (249,85) .. controls (249,82.24) and (251.24,80) .. (254,80) .. controls (256.76,80) and (259,82.24) .. (259,85) .. controls (259,87.76) and (256.76,90) .. (254,90) .. controls (251.24,90) and (249,87.76) .. (249,85) -- cycle ;
%Shape: Circle [id:dp002537906028607484] 
\draw  [fill={rgb, 255:red, 255; green, 255; blue, 255 }  ,fill opacity=1 ] (282.5,85) .. controls (282.5,82.24) and (284.74,80) .. (287.5,80) .. controls (290.26,80) and (292.5,82.24) .. (292.5,85) .. controls (292.5,87.76) and (290.26,90) .. (287.5,90) .. controls (284.74,90) and (282.5,87.76) .. (282.5,85) -- cycle ;
%Shape: Circle [id:dp32691542837669507] 
\draw  [fill={rgb, 255:red, 255; green, 255; blue, 255 }  ,fill opacity=1 ] (316,85) .. controls (316,82.24) and (318.24,80) .. (321,80) .. controls (323.76,80) and (326,82.24) .. (326,85) .. controls (326,87.76) and (323.76,90) .. (321,90) .. controls (318.24,90) and (316,87.76) .. (316,85) -- cycle ;
%Shape: Circle [id:dp8955523622970727] 
\draw  [fill={rgb, 255:red, 255; green, 255; blue, 255 }  ,fill opacity=1 ] (349.5,85) .. controls (349.5,82.24) and (351.74,80) .. (354.5,80) .. controls (357.26,80) and (359.5,82.24) .. (359.5,85) .. controls (359.5,87.76) and (357.26,90) .. (354.5,90) .. controls (351.74,90) and (349.5,87.76) .. (349.5,85) -- cycle ;
%Shape: Circle [id:dp2952860485589799] 
\draw  [fill={rgb, 255:red, 208; green, 2; blue, 27 }  ,fill opacity=1 ] (383,85) .. controls (383,82.24) and (385.24,80) .. (388,80) .. controls (390.76,80) and (393,82.24) .. (393,85) .. controls (393,87.76) and (390.76,90) .. (388,90) .. controls (385.24,90) and (383,87.76) .. (383,85) -- cycle ;
%Shape: Circle [id:dp5782729803993517] 
\draw  [fill={rgb, 255:red, 255; green, 255; blue, 255 }  ,fill opacity=1 ] (416.5,85) .. controls (416.5,82.24) and (418.74,80) .. (421.5,80) .. controls (424.26,80) and (426.5,82.24) .. (426.5,85) .. controls (426.5,87.76) and (424.26,90) .. (421.5,90) .. controls (418.74,90) and (416.5,87.76) .. (416.5,85) -- cycle ;
%Shape: Circle [id:dp23320221219323] 
\draw  [fill={rgb, 255:red, 208; green, 2; blue, 27 }  ,fill opacity=1 ] (316,56) .. controls (316,53.24) and (318.24,51) .. (321,51) .. controls (323.76,51) and (326,53.24) .. (326,56) .. controls (326,58.76) and (323.76,61) .. (321,61) .. controls (318.24,61) and (316,58.76) .. (316,56) -- cycle ;
%Shape: Circle [id:dp07333799177269418] 
\draw  [fill={rgb, 255:red, 208; green, 2; blue, 27 }  ,fill opacity=1 ] (383,56) .. controls (383,53.24) and (385.24,51) .. (388,51) .. controls (390.76,51) and (393,53.24) .. (393,56) .. controls (393,58.76) and (390.76,61) .. (388,61) .. controls (385.24,61) and (383,58.76) .. (383,56) -- cycle ;
%Shape: Circle [id:dp7441304212654724] 
\draw  [fill={rgb, 255:red, 255; green, 255; blue, 255 }  ,fill opacity=1 ] (394.5,27) .. controls (394.5,24.24) and (396.74,22) .. (399.5,22) .. controls (402.26,22) and (404.5,24.24) .. (404.5,27) .. controls (404.5,29.76) and (402.26,32) .. (399.5,32) .. controls (396.74,32) and (394.5,29.76) .. (394.5,27) -- cycle ;
%Shape: Circle [id:dp03285159237149271] 
\draw  [fill={rgb, 255:red, 255; green, 255; blue, 255 }  ,fill opacity=1 ] (371.5,27) .. controls (371.5,24.24) and (373.74,22) .. (376.5,22) .. controls (379.26,22) and (381.5,24.24) .. (381.5,27) .. controls (381.5,29.76) and (379.26,32) .. (376.5,32) .. controls (373.74,32) and (371.5,29.76) .. (371.5,27) -- cycle ;
%Shape: Polygon Curved [id:ds7807827823734887] 
\draw  [dash pattern={on 4.5pt off 4.5pt}] (276.5,87) .. controls (272.5,70) and (303.5,71) .. (321,70.5) .. controls (338.5,70) and (369.5,68) .. (367.5,89) .. controls (365.5,110) and (358.5,109) .. (329.5,107) .. controls (300.5,105) and (280.5,104) .. (276.5,87) -- cycle ;
%Shape: Polygon Curved [id:ds01804151278026589] 
\draw  [dash pattern={on 4.5pt off 4.5pt}] (174.5,90) .. controls (172.5,79) and (187.5,72) .. (205.5,71) .. controls (223.5,70) and (237.5,72) .. (237.25,85) .. controls (237,98) and (238.5,107) .. (208.5,107) .. controls (178.5,107) and (176.5,101) .. (174.5,90) -- cycle ;
%Shape: Polygon Curved [id:ds3238785955213135] 
\draw  [dash pattern={on 4.5pt off 4.5pt}] (408.5,94) .. controls (417,118) and (458.5,106) .. (444.5,87) .. controls (430.5,68) and (400,70) .. (408.5,94) -- cycle ;
%Shape: Polygon Curved [id:ds7433946876416038] 
\draw  [dash pattern={on 4.5pt off 4.5pt}] (358.5,43) .. controls (368,68) and (399.5,36) .. (383.5,15) .. controls (367.5,-6) and (349,18) .. (358.5,43) -- cycle ;
%Curve Lines [id:da1282352573747737] 
\draw [line width=1.5]    (359.5,52) .. controls (345.25,65.3) and (355.37,67.77) .. (378.71,61.98) ;
\draw [shift={(382.5,61)}, rotate = 164.93] [fill={rgb, 255:red, 0; green, 0; blue, 0 }  ][line width=0.08]  [draw opacity=0] (11.61,-5.58) -- (0,0) -- (11.61,5.58) -- cycle    ;
%Curve Lines [id:da8673804573859186] 
\draw [line width=1.5]    (428,73) .. controls (431.8,44.5) and (410.33,58.42) .. (395.31,77.03) ;
\draw [shift={(393,80)}, rotate = 306.87] [fill={rgb, 255:red, 0; green, 0; blue, 0 }  ][line width=0.08]  [draw opacity=0] (11.61,-5.58) -- (0,0) -- (11.61,5.58) -- cycle    ;
%Curve Lines [id:da7912875571236077] 
\draw [line width=1.5]    (287,71) .. controls (278.97,37.92) and (291.92,37) .. (311.09,47.97) ;
\draw [shift={(314.5,50)}, rotate = 211.76] [fill={rgb, 255:red, 0; green, 0; blue, 0 }  ][line width=0.08]  [draw opacity=0] (11.61,-5.58) -- (0,0) -- (11.61,5.58) -- cycle    ;
%Curve Lines [id:da8760396510477308] 
\draw [line width=1.5]    (214.5,68) .. controls (221.22,49.76) and (229.78,54.56) .. (246.38,77.09) ;
\draw [shift={(248.5,80)}, rotate = 234.25] [fill={rgb, 255:red, 0; green, 0; blue, 0 }  ][line width=0.08]  [draw opacity=0] (11.61,-5.58) -- (0,0) -- (11.61,5.58) -- cycle    ;

% Text Node
\draw (189,92.4) node [anchor=north west][inner sep=0.75pt]    {$a$};
% Text Node
\draw (323,92.4) node [anchor=north west][inner sep=0.75pt]    {$e$};
% Text Node
\draw (356.5,88.4) node [anchor=north west][inner sep=0.75pt]    {$f$};
% Text Node
\draw (390,90.4) node [anchor=north west][inner sep=0.75pt]    {$g$};
% Text Node
\draw (222.5,88.4) node [anchor=north west][inner sep=0.75pt]    {$b$};
% Text Node
\draw (256,92.4) node [anchor=north west][inner sep=0.75pt]    {$c$};
% Text Node
\draw (289.5,88.4) node [anchor=north west][inner sep=0.75pt]    {$d$};
% Text Node
\draw (423.5,90.4) node [anchor=north west][inner sep=0.75pt]    {$h$};
% Text Node
\draw (329,54.4) node [anchor=north west][inner sep=0.75pt]    {$i$};
% Text Node
\draw (407,24.4) node [anchor=north west][inner sep=0.75pt]    {$l$};
% Text Node
\draw (359,25.4) node [anchor=north west][inner sep=0.75pt]    {$k$};
% Text Node
\draw (395,50.4) node [anchor=north west][inner sep=0.75pt]    {$j$};

\end{tikzpicture}

    \caption{Let us consider the set of red vertices $X := \{c, g, i, j\}$. $\mathcal{C}_X$ contains five components (which are $\{a,b\}, \{d,e,f\}, \{h\}, \{k\}, \{l\}$). The component $\{l\} \in \mathcal{C}_X$ is a good component, because it is possible to map injectively each other component to a vertex in its boundary.}
    \label{fig:example_good_component}
\end{figure}

\begin{figure}[h!]
    \centering
    
    \begin{tikzpicture}[x=0.75pt,y=0.75pt,yscale=-1,xscale=1]

%Straight Lines [id:da8043150739980819] 
\draw    (388,56) -- (376.5,27) ;
%Straight Lines [id:da3782239463714072] 
\draw    (388,56) -- (399.5,27) ;
%Straight Lines [id:da3153434300280563] 
\draw    (187,85) -- (220.5,85) ;
%Straight Lines [id:da8263634859194179] 
\draw    (220.5,85) -- (254,85) ;
%Straight Lines [id:da6841629172908128] 
\draw    (254,85) -- (287.5,85) ;
%Straight Lines [id:da2500160982119267] 
\draw    (287.5,85) -- (321,85) ;
%Straight Lines [id:da3373777248163913] 
\draw    (321,85) -- (354.5,85) ;
%Straight Lines [id:da38962727122151775] 
\draw    (354.5,85) -- (388,85) ;
%Straight Lines [id:da10864763457286819] 
\draw    (388,85) -- (421.5,85) ;
%Straight Lines [id:da33192928128839416] 
\draw    (321,85) -- (321,56) ;
%Straight Lines [id:da9616805020039595] 
\draw    (388,85) -- (388,56) ;
%Shape: Circle [id:dp23129232623807305] 
\draw  [fill={rgb, 255:red, 255; green, 255; blue, 255 }  ,fill opacity=1 ] (182,85) .. controls (182,82.24) and (184.24,80) .. (187,80) .. controls (189.76,80) and (192,82.24) .. (192,85) .. controls (192,87.76) and (189.76,90) .. (187,90) .. controls (184.24,90) and (182,87.76) .. (182,85) -- cycle ;
%Shape: Circle [id:dp833356915964317] 
\draw  [fill={rgb, 255:red, 208; green, 2; blue, 27 }  ,fill opacity=1 ] (215.5,85) .. controls (215.5,82.24) and (217.74,80) .. (220.5,80) .. controls (223.26,80) and (225.5,82.24) .. (225.5,85) .. controls (225.5,87.76) and (223.26,90) .. (220.5,90) .. controls (217.74,90) and (215.5,87.76) .. (215.5,85) -- cycle ;
%Shape: Circle [id:dp48059563893673196] 
\draw  [fill={rgb, 255:red, 255; green, 255; blue, 255 }  ,fill opacity=1 ] (249,85) .. controls (249,82.24) and (251.24,80) .. (254,80) .. controls (256.76,80) and (259,82.24) .. (259,85) .. controls (259,87.76) and (256.76,90) .. (254,90) .. controls (251.24,90) and (249,87.76) .. (249,85) -- cycle ;
%Shape: Circle [id:dp002537906028607484] 
\draw  [fill={rgb, 255:red, 255; green, 255; blue, 255 }  ,fill opacity=1 ] (282.5,85) .. controls (282.5,82.24) and (284.74,80) .. (287.5,80) .. controls (290.26,80) and (292.5,82.24) .. (292.5,85) .. controls (292.5,87.76) and (290.26,90) .. (287.5,90) .. controls (284.74,90) and (282.5,87.76) .. (282.5,85) -- cycle ;
%Shape: Circle [id:dp32691542837669507] 
\draw  [fill={rgb, 255:red, 208; green, 2; blue, 27 }  ,fill opacity=1 ] (316,85) .. controls (316,82.24) and (318.24,80) .. (321,80) .. controls (323.76,80) and (326,82.24) .. (326,85) .. controls (326,87.76) and (323.76,90) .. (321,90) .. controls (318.24,90) and (316,87.76) .. (316,85) -- cycle ;
%Shape: Circle [id:dp8955523622970727] 
\draw  [fill={rgb, 255:red, 208; green, 2; blue, 27 }  ,fill opacity=1 ] (349.5,85) .. controls (349.5,82.24) and (351.74,80) .. (354.5,80) .. controls (357.26,80) and (359.5,82.24) .. (359.5,85) .. controls (359.5,87.76) and (357.26,90) .. (354.5,90) .. controls (351.74,90) and (349.5,87.76) .. (349.5,85) -- cycle ;
%Shape: Circle [id:dp2952860485589799] 
\draw  [fill={rgb, 255:red, 208; green, 2; blue, 27 }  ,fill opacity=1 ] (383,85) .. controls (383,82.24) and (385.24,80) .. (388,80) .. controls (390.76,80) and (393,82.24) .. (393,85) .. controls (393,87.76) and (390.76,90) .. (388,90) .. controls (385.24,90) and (383,87.76) .. (383,85) -- cycle ;
%Shape: Circle [id:dp5782729803993517] 
\draw  [fill={rgb, 255:red, 255; green, 255; blue, 255 }  ,fill opacity=1 ] (416.5,85) .. controls (416.5,82.24) and (418.74,80) .. (421.5,80) .. controls (424.26,80) and (426.5,82.24) .. (426.5,85) .. controls (426.5,87.76) and (424.26,90) .. (421.5,90) .. controls (418.74,90) and (416.5,87.76) .. (416.5,85) -- cycle ;
%Shape: Circle [id:dp23320221219323] 
\draw  [fill={rgb, 255:red, 255; green, 255; blue, 255 }  ,fill opacity=1 ] (316,56) .. controls (316,53.24) and (318.24,51) .. (321,51) .. controls (323.76,51) and (326,53.24) .. (326,56) .. controls (326,58.76) and (323.76,61) .. (321,61) .. controls (318.24,61) and (316,58.76) .. (316,56) -- cycle ;
%Shape: Circle [id:dp07333799177269418] 
\draw  [fill={rgb, 255:red, 255; green, 255; blue, 255 }  ,fill opacity=1 ] (383,56) .. controls (383,53.24) and (385.24,51) .. (388,51) .. controls (390.76,51) and (393,53.24) .. (393,56) .. controls (393,58.76) and (390.76,61) .. (388,61) .. controls (385.24,61) and (383,58.76) .. (383,56) -- cycle ;
%Shape: Circle [id:dp7441304212654724] 
\draw  [fill={rgb, 255:red, 255; green, 255; blue, 255 }  ,fill opacity=1 ] (394.5,27) .. controls (394.5,24.24) and (396.74,22) .. (399.5,22) .. controls (402.26,22) and (404.5,24.24) .. (404.5,27) .. controls (404.5,29.76) and (402.26,32) .. (399.5,32) .. controls (396.74,32) and (394.5,29.76) .. (394.5,27) -- cycle ;
%Shape: Circle [id:dp03285159237149271] 
\draw  [fill={rgb, 255:red, 255; green, 255; blue, 255 }  ,fill opacity=1 ] (371.5,27) .. controls (371.5,24.24) and (373.74,22) .. (376.5,22) .. controls (379.26,22) and (381.5,24.24) .. (381.5,27) .. controls (381.5,29.76) and (379.26,32) .. (376.5,32) .. controls (373.74,32) and (371.5,29.76) .. (371.5,27) -- cycle ;

% Text Node
\draw (189,92.4) node [anchor=north west][inner sep=0.75pt]    {$a$};
% Text Node
\draw (323,92.4) node [anchor=north west][inner sep=0.75pt]    {$e$};
% Text Node
\draw (356.5,88.4) node [anchor=north west][inner sep=0.75pt]    {$f$};
% Text Node
\draw (390,90.4) node [anchor=north west][inner sep=0.75pt]    {$g$};
% Text Node
\draw (222.5,88.4) node [anchor=north west][inner sep=0.75pt]    {$b$};
% Text Node
\draw (256,92.4) node [anchor=north west][inner sep=0.75pt]    {$c$};
% Text Node
\draw (289.5,88.4) node [anchor=north west][inner sep=0.75pt]    {$d$};
% Text Node
\draw (423.5,90.4) node [anchor=north west][inner sep=0.75pt]    {$h$};
% Text Node
\draw (329,54.4) node [anchor=north west][inner sep=0.75pt]    {$i$};
% Text Node
\draw (407,24.4) node [anchor=north west][inner sep=0.75pt]    {$l$};
% Text Node
\draw (359,25.4) node [anchor=north west][inner sep=0.75pt]    {$k$};
% Text Node
\draw (394,54.4) node [anchor=north west][inner sep=0.75pt]    {$j$};

\end{tikzpicture}

    \caption{Let us consider the set of red vertices $X:=\{b, e, f, g\}$. $\mathcal{C}_X$ contains five components (which are $\{a\}, \{c,d\}, \{i\}, \{j,k,l\}, \{h\}$) and there are only three vertices in the union of their boundaries (which are $b, e, g$). Thus, by a simple cardinality argument, $\mathcal{C}_X$ does not contain any good component, so $X$ is a bad set.}
    \label{fig:example_bad_set}
\end{figure}
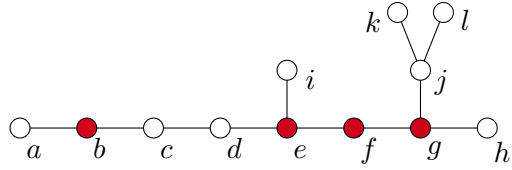

\begin{definition}
\label{def:tau}
For every graph~$G$, we denote by~$\bad{G}$ the maximum size of a bad set in~$G$ (and if $G$ does not have any bad set, we set $\bad{G} := 0$ by convention).     
\end{definition}

The aim of this section is to prove the following result:
\begin{theorem}
    \label{thm:binary}
    Let $T$ be a $n$-vertex tree and $k \in \{2, \ldots, n\}$. There exists an \EFO allocation for $k$ agents for every common additive binary valuation if and only if $k \geqslant \bad{T}+2$.
\end{theorem}

Note that paths do not have any bad set. Thus, if~$T$ is a path, we have~$\bad{T}=0$ and the result of Theorem~\ref{thm:binary} holds, because paths always admit an \EFO allocation for monotone valuations~\cite{Igarashi2023}. It thus suffices to prove Theorem~\ref{thm:binary} for proper trees only.

%\subsection{First implication of Theorem~\ref{thm:binary}}
\subsection{Negative Instances in the Presence of Bad Sets}

In this section, we prove the first implication of Theorem~\ref{thm:binary}, as a direct consequence of the two following lemmas. First, Lemma \ref{lem:smaller_bad_set} ensures that if we have a bad set of some size $k$, there is a bad set of size $k-1$. Second, Lemma \ref{lem:bad_valuation} ensures that if we have a bad set of size $k$, then there is a negative instance for $k+1$ agents. The combination of these two lemmas proves that, for every $k \in \{2, \ldots, \bad{T}+1\}$, there exists a negative instance. 

\begin{restatable}{lemma}{LemmaSmallerBadSet}
    \label{lem:smaller_bad_set}
    Let $T$ be a tree and $X$ be a bad set for $T$. If $|X| \geqslant 2$, then there exists $u \in X$ such that a $Y := X \setminus \{u\}$ is still a bad set.
\end{restatable} 

\begin{proof}
    If $S$ is a set of vertices of~$T$ and~$u \in V(T) \setminus S$, we denote by $C_u^S$ the connected component of $T \setminus S$ to which~$u$ belongs (so we have $C_u^S \in \C_S$). We distinguish several cases.
\smallskip

\noindent\textbf{Case 1.}
If $X$ contains a degree-1 vertex~$u$, let~$v$ be the neighbor of~$u$. Let us prove that $Y := X \setminus \{u\}$ is a bad set. We distinguish two cases based on whether $v$ belongs to $X$ or not. If $v \notin X$, removing $u$ from $X$ just adds one vertex in~$C_v^X$ and removes $u$ from its boundary, so $Y$ is still a bad set. If $v \in X$, then removing $u$ from $X$ just adds one connected component to $T \setminus X$ and does not change all the other ones and their boundaries, so $Y$ is also a bad set.
\smallskip

\noindent\textbf{Case 2.} If $X$ contains a degree-2 vertex~$u$, let~$v$ and~$w$ be its two neighbors. Let us prove that $Y := X \setminus \{u\}$ is a bad set. There are several subcases, see Figure~\ref{fig:cases_degree2} for an illustration.

        \begin{figure}[h!]
            \centering
            
\begin{tikzpicture}[x=0.75pt,y=0.75pt,yscale=-1,xscale=1]

%Straight Lines [id:da16970228181221714] 
\draw    (49,86) -- (82.5,86) ;
%Straight Lines [id:da014400691949230104] 
\draw    (82.5,86) -- (116,86) ;
%Shape: Circle [id:dp3329177072823849] 
\draw  [fill={rgb, 255:red, 208; green, 2; blue, 27 }  ,fill opacity=1 ] (44,86) .. controls (44,83.24) and (46.24,81) .. (49,81) .. controls (51.76,81) and (54,83.24) .. (54,86) .. controls (54,88.76) and (51.76,91) .. (49,91) .. controls (46.24,91) and (44,88.76) .. (44,86) -- cycle ;
%Shape: Circle [id:dp22781048498921497] 
\draw  [fill={rgb, 255:red, 208; green, 2; blue, 27 }  ,fill opacity=1 ] (77.5,86) .. controls (77.5,83.24) and (79.74,81) .. (82.5,81) .. controls (85.26,81) and (87.5,83.24) .. (87.5,86) .. controls (87.5,88.76) and (85.26,91) .. (82.5,91) .. controls (79.74,91) and (77.5,88.76) .. (77.5,86) -- cycle ;
%Shape: Circle [id:dp7860815594208476] 
\draw  [fill={rgb, 255:red, 208; green, 2; blue, 27 }  ,fill opacity=1 ] (111,86) .. controls (111,83.24) and (113.24,81) .. (116,81) .. controls (118.76,81) and (121,83.24) .. (121,86) .. controls (121,88.76) and (118.76,91) .. (116,91) .. controls (113.24,91) and (111,88.76) .. (111,86) -- cycle ;
%Straight Lines [id:da27391924913802657] 
\draw    (212.5,86) -- (246,86) ;
%Straight Lines [id:da5705779011585469] 
\draw    (246,86) -- (279.5,86) ;
%Shape: Circle [id:dp05286008434871248] 
\draw  [fill={rgb, 255:red, 255; green, 255; blue, 255 }  ,fill opacity=1 ] (207.5,86) .. controls (207.5,83.24) and (209.74,81) .. (212.5,81) .. controls (215.26,81) and (217.5,83.24) .. (217.5,86) .. controls (217.5,88.76) and (215.26,91) .. (212.5,91) .. controls (209.74,91) and (207.5,88.76) .. (207.5,86) -- cycle ;
%Shape: Circle [id:dp13683972709446823] 
\draw  [fill={rgb, 255:red, 208; green, 2; blue, 27 }  ,fill opacity=1 ] (241,86) .. controls (241,83.24) and (243.24,81) .. (246,81) .. controls (248.76,81) and (251,83.24) .. (251,86) .. controls (251,88.76) and (248.76,91) .. (246,91) .. controls (243.24,91) and (241,88.76) .. (241,86) -- cycle ;
%Shape: Circle [id:dp5911508373558654] 
\draw  [fill={rgb, 255:red, 208; green, 2; blue, 27 }  ,fill opacity=1 ] (274.5,86) .. controls (274.5,83.24) and (276.74,81) .. (279.5,81) .. controls (282.26,81) and (284.5,83.24) .. (284.5,86) .. controls (284.5,88.76) and (282.26,91) .. (279.5,91) .. controls (276.74,91) and (274.5,88.76) .. (274.5,86) -- cycle ;
%Straight Lines [id:da4970846492905965] 
\draw    (376,86) -- (409.5,86) ;
%Straight Lines [id:da25592309918037404] 
\draw    (409.5,86) -- (443,86) ;
%Shape: Circle [id:dp9870919248408667] 
\draw  [fill={rgb, 255:red, 255; green, 255; blue, 255 }  ,fill opacity=1 ] (371,86) .. controls (371,83.24) and (373.24,81) .. (376,81) .. controls (378.76,81) and (381,83.24) .. (381,86) .. controls (381,88.76) and (378.76,91) .. (376,91) .. controls (373.24,91) and (371,88.76) .. (371,86) -- cycle ;
%Shape: Circle [id:dp15601507365025213] 
\draw  [fill={rgb, 255:red, 208; green, 2; blue, 27 }  ,fill opacity=1 ] (404.5,86) .. controls (404.5,83.24) and (406.74,81) .. (409.5,81) .. controls (412.26,81) and (414.5,83.24) .. (414.5,86) .. controls (414.5,88.76) and (412.26,91) .. (409.5,91) .. controls (406.74,91) and (404.5,88.76) .. (404.5,86) -- cycle ;
%Shape: Circle [id:dp38770203924145075] 
\draw  [fill={rgb, 255:red, 255; green, 255; blue, 255 }  ,fill opacity=1 ] (438,86) .. controls (438,83.24) and (440.24,81) .. (443,81) .. controls (445.76,81) and (448,83.24) .. (448,86) .. controls (448,88.76) and (445.76,91) .. (443,91) .. controls (440.24,91) and (438,88.76) .. (438,86) -- cycle ;

% Text Node
\draw (84.5,89.4) node [anchor=north west][inner sep=0.75pt]    {$u$};
% Text Node
\draw (51,89.4) node [anchor=north west][inner sep=0.75pt]    {$v$};
% Text Node
\draw (118,89.4) node [anchor=north west][inner sep=0.75pt]    {$w$};
% Text Node
\draw (248,89.4) node [anchor=north west][inner sep=0.75pt]    {$u$};
% Text Node
\draw (214.5,89.4) node [anchor=north west][inner sep=0.75pt]    {$v$};
% Text Node
\draw (281.5,89.4) node [anchor=north west][inner sep=0.75pt]    {$w$};
% Text Node
\draw (411.5,89.4) node [anchor=north west][inner sep=0.75pt]    {$u$};
% Text Node
\draw (378,89.4) node [anchor=north west][inner sep=0.75pt]    {$v$};
% Text Node
\draw (445,89.4) node [anchor=north west][inner sep=0.75pt]    {$w$};
% Text Node
\draw (73,110) node [anchor=north west][inner sep=0.75pt]   [align=left] {(a)};
% Text Node
\draw (236,110) node [anchor=north west][inner sep=0.75pt]   [align=left] {(b)};
% Text Node
\draw (400,110) node [anchor=north west][inner sep=0.75pt]   [align=left] {(c)};

\end{tikzpicture}

            \caption{The three subcases. Red vertices are vertices of $X$. Only $N[u]$ is represented here.}
            \label{fig:cases_degree2}
        \end{figure}
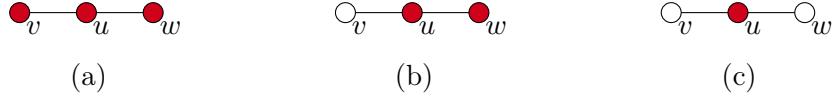
        \begin{enumerate}
            \item If $v \in X$ and~$w \in X$, then removing~$u$ from~$X$ only creates a new connected component in $T \setminus X$ and does not affect the boundaries of the other components, so $Y$ is a bad set.
            
            \item If $v \notin X$ and~$w \in X$, removing $u$ from~$X$ adds $u$ to $C_v^X$. In other words, we have $C_v^Y = C_v^X \cup \{u\}$. By contradiction, assume that $Y$ is not a bad set, and let us show that $X$ is not a bad set either. Let $C \in \mathcal{C}_Y$ be a good component, and $f : \mathcal{C}_Y \setminus C \to Y$ be its matching function. If $C = C_{v}^Y$, then $C_v^X$ would be a good component in $\mathcal{C}_X$ with the same matching function $f$. If $C \neq C_{v}^Y$, then $C$ is a good component in $\mathcal{C}_X$ with the matching function~$g:\mathcal{C}_X \setminus C \to X$ defined as follows: we set $g(C_v^X) := u$, and  $g(C') = f(C')$ for every $C' \in \mathcal{C}_X \setminus \{C, C_v^X\}$. This is a contradiction, because $X$ is a bad set. So $Y$ is a bad set too.
            \item If $v \notin X$ and~$w \notin X$, removing $u$ from~$X$ merges~$u$ with $C_v$ and $C_w$. In other words, we have $C_v^Y = C_w^Y = C_v^X \cup C_w^X \cup \{u\}$. By contradiction, assume that $Y$ is not a bad set, and let us show that $X$ is not a bad set either. Let $C \in \mathcal{C}_Y$ be a good component, and $f : \mathcal{C}_Y \setminus C \to Y$ be its matching function. If $C = C_{v}^Y$, then $C_v^X$ is a good component in $\mathcal{C}_X$ by taking the same matching function $f$ and matching $C_w^X$ to~$u$. If $C \neq C_{v}^Y$, then $C$ is a good component in $\mathcal{C}_X$ with the matching function~$g:\mathcal{C}_X \setminus C \to X$ defined as follows: we set $g(C_v^X) := f(C_{v}^Y)$, $g(C_w^X):=u$, and  $g(C') = f(C')$ for every $C' \in \mathcal{C}_X \setminus \{C, C_v^X,C_w^X\}$. This is a contradiction, because $X$ is a bad set. So $Y$ is a bad set too.
        \end{enumerate}

\noindent\textbf{Case 3.} If $X$ contains only vertices of degree at least~$3$, it is sufficient to prove the following claim: 
\begin{center}
If $Y$ is a non-empty set containing only nodes of degree at least~3, then $Y$ is a bad set.
\end{center}
    Using this claim, we get that for every $u \in X$, $X \setminus u$ is still a bad set. Let us finally prove this claim. By a simple cardinality argument, it is sufficient to prove that for every non-empty set~$Y$ containing only vertices of degree at least~3, we have $|\mathcal{C}_Y| \geqslant |Y| + 2$.
    Let us prove this by induction on~$|Y|$. If $|Y| = 1$, we have $|\mathcal{C}_Y| \geqslant 3$ because the vertex in~$Y$ has degree at least~3. Assume now that $|Y| \geqslant 2$, and let $u \in Y$ be an external vertex in $Y$. Let $Y' := Y \setminus \{u\}$. Since $u$ has degree at least~3, there are at least two connected components of $T \setminus\{u\}$ that do not intersect~$Y'$. Thus, adding $u$ to~$Y'$ splits $C_u^{Y'}$ in at least two connected components. In other words, $C_u^{Y'}$ is the union of $\{u\}$ and at least two connected components in $\C_Y$. So we finally get $|\mathcal{C}_Y| \geqslant |\mathcal{C}_{Y'}|+1 \geqslant |Y'| + 3 \geqslant |Y| + 2$ (where the second inequality holds by the induction hypothesis on~$Y'$).
\end{proof}

\begin{lemma}
    \label{lem:bad_valuation}
    Let $G$ be a graph and $X$ be a bad set for $G$. Then, there exists a binary additive valuation that does not admit any \EFO allocation with $k:=|X|+1$ agents.
\end{lemma}

\begin{proof}
    Let us consider the following binary additive valuation: vertices in~$X$ get value~$1$, and vertices in $V(G) \setminus X$ get value~$0$. Let us prove that this valuation does not admit any \EFO allocation with $k := |X|+1$ agents.
    Consider any allocation to $k$ agents. By the pigeonhole principle, there exists at least one agent that has only vertices in $V(G) \setminus X$. Let us call such an agent a \emph{deprived agent}. 
    There are now two cases.

    \begin{itemize}
        \item If there are at least two deprived agents, by the pigeonhole principle, there exists another agent that has two vertices in~$X$. Thus, by definition of the valuation, this allocation is not \EFO. Indeed, some agent receives a set of value two which is envied up to more than one item by the deprived agents.

        \item If there is a unique deprived agent, all its vertices are contained in some connected component of $G \setminus X$. Let $C \in \mathcal{C}_X$ be this component.
        Consider the function $f : \mathcal{C}_X \setminus C \to X$ defined as follows. For any component of $C' \in \mathcal{C}_X \setminus C$, consider an arbitrary $u \in C'$. There exists a non-deprived agent that receives the vertex~$u$ in its subtree. Since the agent is not deprived, it also receives a vertex $v \in X$. We set $f(C'):=v$. 
        
        By definition of a bad set, $C$~is a bad component, so $f$ is not injective. Thus, there exists two distinct components $C_1, C_2 \in \mathcal{C}_X \setminus C$ and $v \in X$ such that $f(C_1) = f(C_2) = v$. By definition of~$f$, there exists an agent that receives a set containing $v$, and at least one vertex from $C_1$ and at least one vertex from~$C_2$. In particular, removing $v$ from its subtree disconects it. Thus, the removal of any leaf of the set of $v$ leaves a set of value one. So~this agent is envied by the deprived agent, and this allocation is not \EFO. \qedhere
    \end{itemize}
    
\end{proof}

Lemma~\ref{lem:smaller_bad_set} and Lemma~\ref{lem:bad_valuation} immediately prove one of the two directions of Theorem~\ref{thm:binary}. Indeed, let $T$ be a proper tree, and assume that $2 \leqslant k \leqslant \bad{T} + 1$. By Lemma~\ref{lem:smaller_bad_set}, there exists a bad set~$X$ such that $|X| = k-1$. Then, Lemma~\ref{lem:bad_valuation} ensures the existence of a valuation that does not admit any \EFO allocation.
\medskip

%\subsection{Second implication of Theorem~\ref{thm:binary}}
\subsection{Positive Allocations in the Absence of Bad Sets}

Let us now prove the other implication of Theorem~\ref{thm:binary}. Namely, let $T$ be an $n$-vertex proper tree, assume that $k \geqslant \bad{T} + 2$, consider a common additive binary valuation, and let us show that there exists an \EFO allocation. Let us denote by $V_0$ (resp. by $V_1$) the subset of $V(T)$ containing all the vertices with value~$0$ (resp. with value~$1$). Let~$n_0 := |V_0|$ and $n_1 := |V_1|$. We have $n = n_0 + n_1$.

Assume first that $n_1 \leqslant k-1$. In this case, let~$X \subseteq V(T)$ be a set of vertices containing all vertices in $V_1$ plus $k-n_1-1$ other additional vertices in $V_0$, chosen arbitrarily.
We have $V_1 \subseteq X$ and $|X| = k-1 \geqslant \bad{T}+1$, so $X$ is not a bad set. Let $C \in \mathcal{C}_X$ be a good component, and let $f: \mathcal{C}_X \setminus C \to X$ be its matching function. Let us denote by $u_1, \ldots, u_{k-1}$ the vertices of $X$.
Let us consider $k$ agents $A_1, \ldots, A_k$ and the following allocation.
For each $i \in \{1, \ldots, k-1\}$, the vertex $u_i$ is allocated to $A_i$. If there exists a component $C' \in \mathcal{C}_X \setminus C$ such that $f(C')=u_i$, $A_i$ also receives all the vertices of $C'$ (in this case, such a $C'$ is unique since $f$ is injective).
Finally, $A_k$ receives all the vertices in $C$. Let us prove that this allocation is~\EFO. Note that, for every $i \in \{1, \ldots, k\}$, the set of vertices allocated to $A_i$ is connected. Moreover, the total value of the vertices allocated to $A_k$ is~$0$ (because $C \subseteq V(T) \setminus X \subseteq V_0$), and for each $i \in \{1, \ldots, k-1\}$, the total value of the vertices allocated to $A_i$ is~$1$ if $u_i \in V_1$, and~$0$ otherwise.
Since for each $i \in \{1, \ldots, k-1\}$, $u_i$ is not a cut-vertex in the part of $A_i$, this ensures that this allocation is \EFO.

\begin{figure}[h!]
    \centering
    \begin{tikzpicture}[x=0.75pt,y=0.75pt,yscale=-1,xscale=1]

%Straight Lines [id:da7091076788737796] 
\draw    (482.5,81) -- (482.5,105.5) ;
%Straight Lines [id:da8160536866558529] 
\draw    (482.5,130) -- (469,154.5) ;
%Straight Lines [id:da7545023651081156] 
\draw    (482.5,130) -- (496,154.5) ;
%Straight Lines [id:da40880187708550997] 
\draw    (482.5,105.5) -- (482.5,130) ;
%Straight Lines [id:da6161795296396849] 
\draw    (496,154.5) -- (509.5,179) ;
%Straight Lines [id:da15886589476693036] 
\draw    (407.5,105.5) -- (407.5,130) ;
%Straight Lines [id:da024881774345310803] 
\draw    (407.5,130) -- (394,154.5) ;
%Straight Lines [id:da6314519054956839] 
\draw    (407.5,130) -- (421,154.5) ;
%Straight Lines [id:da5467423090426149] 
\draw    (286.5,109) -- (286.5,130) ;
%Straight Lines [id:da6368684461988834] 
\draw    (286.5,88) -- (286.5,109) ;
%Straight Lines [id:da8326877328625868] 
\draw    (286.5,67) -- (286.5,88) ;
%Straight Lines [id:da1520797020719069] 
\draw    (286.5,130) -- (286.5,151) ;
%Straight Lines [id:da26622265574067117] 
\draw    (286.5,151) -- (286.5,172) ;
%Straight Lines [id:da21632096896824216] 
\draw    (286.5,130) -- (309.5,130) ;
%Straight Lines [id:da6216552870081155] 
\draw    (309.5,130) -- (332.5,130) ;
%Straight Lines [id:da8377536818518887] 
\draw    (332.5,109) -- (332.5,130) ;
%Straight Lines [id:da8904327569114151] 
\draw    (332.5,88) -- (332.5,109) ;
%Straight Lines [id:da7398928838127611] 
\draw    (332.5,130) -- (332.5,151) ;
%Straight Lines [id:da5218374443197644] 
\draw    (177,130) -- (177,163) ;
%Straight Lines [id:da4479932125656946] 
\draw    (177,97) -- (177,130) ;
%Straight Lines [id:da9219319311005134] 
\draw    (211.5,130) -- (177,130) ;
%Straight Lines [id:da039838313478612664] 
\draw    (211.5,130) -- (211.5,163) ;
%Straight Lines [id:da6037028046956741] 
\draw    (211.5,97) -- (211.5,130) ;
%Shape: Circle [id:dp5604140176243656] 
\draw  [fill={rgb, 255:red, 255; green, 255; blue, 255 }  ,fill opacity=1 ] (172,163) .. controls (172,160.24) and (174.24,158) .. (177,158) .. controls (179.76,158) and (182,160.24) .. (182,163) .. controls (182,165.76) and (179.76,168) .. (177,168) .. controls (174.24,168) and (172,165.76) .. (172,163) -- cycle ;
%Shape: Circle [id:dp7787762887091321] 
\draw  [fill={rgb, 255:red, 255; green, 255; blue, 255 }  ,fill opacity=1 ] (206.5,97) .. controls (206.5,94.24) and (208.74,92) .. (211.5,92) .. controls (214.26,92) and (216.5,94.24) .. (216.5,97) .. controls (216.5,99.76) and (214.26,102) .. (211.5,102) .. controls (208.74,102) and (206.5,99.76) .. (206.5,97) -- cycle ;
%Shape: Circle [id:dp1629349086771925] 
\draw  [fill={rgb, 255:red, 255; green, 255; blue, 255 }  ,fill opacity=1 ] (206.5,163) .. controls (206.5,160.24) and (208.74,158) .. (211.5,158) .. controls (214.26,158) and (216.5,160.24) .. (216.5,163) .. controls (216.5,165.76) and (214.26,168) .. (211.5,168) .. controls (208.74,168) and (206.5,165.76) .. (206.5,163) -- cycle ;
%Shape: Circle [id:dp3954198385429871] 
\draw  [fill={rgb, 255:red, 255; green, 255; blue, 255 }  ,fill opacity=1 ] (206.5,130) .. controls (206.5,127.24) and (208.74,125) .. (211.5,125) .. controls (214.26,125) and (216.5,127.24) .. (216.5,130) .. controls (216.5,132.76) and (214.26,135) .. (211.5,135) .. controls (208.74,135) and (206.5,132.76) .. (206.5,130) -- cycle ;
%Shape: Circle [id:dp43467685125867805] 
\draw  [fill={rgb, 255:red, 255; green, 255; blue, 255 }  ,fill opacity=1 ] (172,97) .. controls (172,94.24) and (174.24,92) .. (177,92) .. controls (179.76,92) and (182,94.24) .. (182,97) .. controls (182,99.76) and (179.76,102) .. (177,102) .. controls (174.24,102) and (172,99.76) .. (172,97) -- cycle ;
%Shape: Circle [id:dp2809927085452779] 
\draw  [fill={rgb, 255:red, 255; green, 255; blue, 255 }  ,fill opacity=1 ] (172,130) .. controls (172,127.24) and (174.24,125) .. (177,125) .. controls (179.76,125) and (182,127.24) .. (182,130) .. controls (182,132.76) and (179.76,135) .. (177,135) .. controls (174.24,135) and (172,132.76) .. (172,130) -- cycle ;
%Shape: Circle [id:dp3105028930646696] 
\draw  [fill={rgb, 255:red, 255; green, 255; blue, 255 }  ,fill opacity=1 ] (281.5,88) .. controls (281.5,85.24) and (283.74,83) .. (286.5,83) .. controls (289.26,83) and (291.5,85.24) .. (291.5,88) .. controls (291.5,90.76) and (289.26,93) .. (286.5,93) .. controls (283.74,93) and (281.5,90.76) .. (281.5,88) -- cycle ;
%Shape: Circle [id:dp7142070306001946] 
\draw  [fill={rgb, 255:red, 255; green, 255; blue, 255 }  ,fill opacity=1 ] (281.5,67) .. controls (281.5,64.24) and (283.74,62) .. (286.5,62) .. controls (289.26,62) and (291.5,64.24) .. (291.5,67) .. controls (291.5,69.76) and (289.26,72) .. (286.5,72) .. controls (283.74,72) and (281.5,69.76) .. (281.5,67) -- cycle ;
%Shape: Circle [id:dp3537620247864416] 
\draw  [fill={rgb, 255:red, 255; green, 255; blue, 255 }  ,fill opacity=1 ] (281.5,172) .. controls (281.5,169.24) and (283.74,167) .. (286.5,167) .. controls (289.26,167) and (291.5,169.24) .. (291.5,172) .. controls (291.5,174.76) and (289.26,177) .. (286.5,177) .. controls (283.74,177) and (281.5,174.76) .. (281.5,172) -- cycle ;
%Shape: Circle [id:dp8167923727110462] 
\draw  [fill={rgb, 255:red, 255; green, 255; blue, 255 }  ,fill opacity=1 ] (281.5,151) .. controls (281.5,148.24) and (283.74,146) .. (286.5,146) .. controls (289.26,146) and (291.5,148.24) .. (291.5,151) .. controls (291.5,153.76) and (289.26,156) .. (286.5,156) .. controls (283.74,156) and (281.5,153.76) .. (281.5,151) -- cycle ;
%Shape: Circle [id:dp3913376843963784] 
\draw  [fill={rgb, 255:red, 255; green, 255; blue, 255 }  ,fill opacity=1 ] (281.5,109) .. controls (281.5,106.24) and (283.74,104) .. (286.5,104) .. controls (289.26,104) and (291.5,106.24) .. (291.5,109) .. controls (291.5,111.76) and (289.26,114) .. (286.5,114) .. controls (283.74,114) and (281.5,111.76) .. (281.5,109) -- cycle ;
%Shape: Circle [id:dp4951508849310634] 
\draw  [fill={rgb, 255:red, 255; green, 255; blue, 255 }  ,fill opacity=1 ] (281.5,130) .. controls (281.5,127.24) and (283.74,125) .. (286.5,125) .. controls (289.26,125) and (291.5,127.24) .. (291.5,130) .. controls (291.5,132.76) and (289.26,135) .. (286.5,135) .. controls (283.74,135) and (281.5,132.76) .. (281.5,130) -- cycle ;
%Shape: Circle [id:dp7059662446994728] 
\draw  [fill={rgb, 255:red, 255; green, 255; blue, 255 }  ,fill opacity=1 ] (304.5,130) .. controls (304.5,127.24) and (306.74,125) .. (309.5,125) .. controls (312.26,125) and (314.5,127.24) .. (314.5,130) .. controls (314.5,132.76) and (312.26,135) .. (309.5,135) .. controls (306.74,135) and (304.5,132.76) .. (304.5,130) -- cycle ;
%Shape: Circle [id:dp908214381091036] 
\draw  [fill={rgb, 255:red, 255; green, 255; blue, 255 }  ,fill opacity=1 ] (327.5,130) .. controls (327.5,127.24) and (329.74,125) .. (332.5,125) .. controls (335.26,125) and (337.5,127.24) .. (337.5,130) .. controls (337.5,132.76) and (335.26,135) .. (332.5,135) .. controls (329.74,135) and (327.5,132.76) .. (327.5,130) -- cycle ;
%Shape: Circle [id:dp1582673878681473] 
\draw  [fill={rgb, 255:red, 255; green, 255; blue, 255 }  ,fill opacity=1 ] (327.5,109) .. controls (327.5,106.24) and (329.74,104) .. (332.5,104) .. controls (335.26,104) and (337.5,106.24) .. (337.5,109) .. controls (337.5,111.76) and (335.26,114) .. (332.5,114) .. controls (329.74,114) and (327.5,111.76) .. (327.5,109) -- cycle ;
%Shape: Circle [id:dp7510712383216565] 
\draw  [fill={rgb, 255:red, 255; green, 255; blue, 255 }  ,fill opacity=1 ] (327.5,88) .. controls (327.5,85.24) and (329.74,83) .. (332.5,83) .. controls (335.26,83) and (337.5,85.24) .. (337.5,88) .. controls (337.5,90.76) and (335.26,93) .. (332.5,93) .. controls (329.74,93) and (327.5,90.76) .. (327.5,88) -- cycle ;
%Shape: Circle [id:dp6201418873594458] 
\draw  [fill={rgb, 255:red, 255; green, 255; blue, 255 }  ,fill opacity=1 ] (327.5,151) .. controls (327.5,148.24) and (329.74,146) .. (332.5,146) .. controls (335.26,146) and (337.5,148.24) .. (337.5,151) .. controls (337.5,153.76) and (335.26,156) .. (332.5,156) .. controls (329.74,156) and (327.5,153.76) .. (327.5,151) -- cycle ;
%Shape: Circle [id:dp9673245161692609] 
\draw  [fill={rgb, 255:red, 255; green, 255; blue, 255 }  ,fill opacity=1 ] (402.5,130) .. controls (402.5,127.24) and (404.74,125) .. (407.5,125) .. controls (410.26,125) and (412.5,127.24) .. (412.5,130) .. controls (412.5,132.76) and (410.26,135) .. (407.5,135) .. controls (404.74,135) and (402.5,132.76) .. (402.5,130) -- cycle ;
%Shape: Circle [id:dp9893667006188185] 
\draw  [fill={rgb, 255:red, 255; green, 255; blue, 255 }  ,fill opacity=1 ] (416,154.5) .. controls (416,151.74) and (418.24,149.5) .. (421,149.5) .. controls (423.76,149.5) and (426,151.74) .. (426,154.5) .. controls (426,157.26) and (423.76,159.5) .. (421,159.5) .. controls (418.24,159.5) and (416,157.26) .. (416,154.5) -- cycle ;
%Shape: Circle [id:dp13151405878312206] 
\draw  [fill={rgb, 255:red, 255; green, 255; blue, 255 }  ,fill opacity=1 ] (389,154.5) .. controls (389,151.74) and (391.24,149.5) .. (394,149.5) .. controls (396.76,149.5) and (399,151.74) .. (399,154.5) .. controls (399,157.26) and (396.76,159.5) .. (394,159.5) .. controls (391.24,159.5) and (389,157.26) .. (389,154.5) -- cycle ;
%Shape: Circle [id:dp38287703917174165] 
\draw  [fill={rgb, 255:red, 255; green, 255; blue, 255 }  ,fill opacity=1 ] (402.5,105.5) .. controls (402.5,102.74) and (404.74,100.5) .. (407.5,100.5) .. controls (410.26,100.5) and (412.5,102.74) .. (412.5,105.5) .. controls (412.5,108.26) and (410.26,110.5) .. (407.5,110.5) .. controls (404.74,110.5) and (402.5,108.26) .. (402.5,105.5) -- cycle ;
%Shape: Circle [id:dp4511827536333234] 
\draw  [fill={rgb, 255:red, 255; green, 255; blue, 255 }  ,fill opacity=1 ] (477.5,130) .. controls (477.5,127.24) and (479.74,125) .. (482.5,125) .. controls (485.26,125) and (487.5,127.24) .. (487.5,130) .. controls (487.5,132.76) and (485.26,135) .. (482.5,135) .. controls (479.74,135) and (477.5,132.76) .. (477.5,130) -- cycle ;
%Shape: Circle [id:dp2231722441884677] 
\draw  [fill={rgb, 255:red, 255; green, 255; blue, 255 }  ,fill opacity=1 ] (504.5,179) .. controls (504.5,176.24) and (506.74,174) .. (509.5,174) .. controls (512.26,174) and (514.5,176.24) .. (514.5,179) .. controls (514.5,181.76) and (512.26,184) .. (509.5,184) .. controls (506.74,184) and (504.5,181.76) .. (504.5,179) -- cycle ;
%Shape: Circle [id:dp06993401709751113] 
\draw  [fill={rgb, 255:red, 255; green, 255; blue, 255 }  ,fill opacity=1 ] (464,154.5) .. controls (464,151.74) and (466.24,149.5) .. (469,149.5) .. controls (471.76,149.5) and (474,151.74) .. (474,154.5) .. controls (474,157.26) and (471.76,159.5) .. (469,159.5) .. controls (466.24,159.5) and (464,157.26) .. (464,154.5) -- cycle ;
%Shape: Circle [id:dp3546501398790314] 
\draw  [fill={rgb, 255:red, 255; green, 255; blue, 255 }  ,fill opacity=1 ] (477.5,105.5) .. controls (477.5,102.74) and (479.74,100.5) .. (482.5,100.5) .. controls (485.26,100.5) and (487.5,102.74) .. (487.5,105.5) .. controls (487.5,108.26) and (485.26,110.5) .. (482.5,110.5) .. controls (479.74,110.5) and (477.5,108.26) .. (477.5,105.5) -- cycle ;
%Shape: Circle [id:dp8854501813727735] 
\draw  [fill={rgb, 255:red, 255; green, 255; blue, 255 }  ,fill opacity=1 ] (477.5,81) .. controls (477.5,78.24) and (479.74,76) .. (482.5,76) .. controls (485.26,76) and (487.5,78.24) .. (487.5,81) .. controls (487.5,83.76) and (485.26,86) .. (482.5,86) .. controls (479.74,86) and (477.5,83.76) .. (477.5,81) -- cycle ;
%Shape: Circle [id:dp9692514782982684] 
\draw  [fill={rgb, 255:red, 255; green, 255; blue, 255 }  ,fill opacity=1 ] (491,154.5) .. controls (491,151.74) and (493.24,149.5) .. (496,149.5) .. controls (498.76,149.5) and (501,151.74) .. (501,154.5) .. controls (501,157.26) and (498.76,159.5) .. (496,159.5) .. controls (493.24,159.5) and (491,157.26) .. (491,154.5) -- cycle ;

% Text Node
\draw (183,185) node [anchor=north west][inner sep=0.75pt]   [align=left] {(1)};
% Text Node
\draw (301,185) node [anchor=north west][inner sep=0.75pt]   [align=left] {(2)};
% Text Node
\draw (397,185) node [anchor=north west][inner sep=0.75pt]   [align=left] {(3)};
% Text Node
\draw (473,185) node [anchor=north west][inner sep=0.75pt]   [align=left] {(4)};

\end{tikzpicture}

    \caption{(1) The $H$-graph (which is also a $H_0$-graph). (2) A $H_1$-graph with branches of size~3,2,2,1. (3) The claw. (4) A subdivided claw with branches of size~2,2,1.}
    \label{fig:H_graph}
\end{figure}
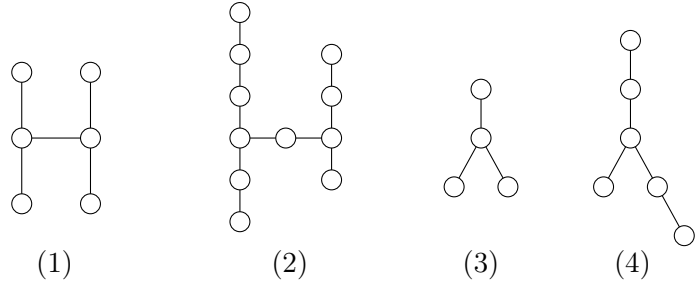

In the rest of the proof, we will assume that $n_1 \geqslant k$. 
Let us first define a few special graphs represented in Figure~\ref{fig:H_graph}.
The \emph{$H$-graph} is the graph on $6$ vertices obtained by adding an edge between the middle vertices of two paths of length~$3$.
For $s \in \mathbb{N}$, a \emph{$H_s$-graph} is a subdivision of the $H$-graph, with the constraint that the middle edge is subdivided exactly $s$~times (and all the other edges are subdivided an arbitrary number of times, possibly zero).
The \emph{claw} is the graph $K_{1,3}$. A \emph{subdivided claw} is any subdivision of the claw graph. The \emph{branches} of a $H_s$-graph (resp. of a subdivided claw) are the connected components obtained when removing the central path of length~$s+2$ (resp. the central vertex).

We need the following result.

\begin{restatable}{proposition}{PropBadSetHalfVertices}
    \label{prop:bad_set_at_least_n/2}
    Let $T$ be a proper $n$-vertex tree.
    Then, we have $\bad{T} +2 \geqslant \left\lfloor \frac{n}{2} \right\rfloor$. Moreover, if $T$ is not a subdivided claw with all branches of even size, we have $\bad{T} +2 \geqslant \left\lceil \frac{n}{2} \right\rceil$.
\end{restatable}

The proof of Proposition~\ref{prop:bad_set_at_least_n/2} requires itself several lemmas.

\begin{lemma}
    \label{lem:branching_node_deg_4}
    If there exists an external branching node that has degree at least~$4$, then $\bad{T}+2 \geqslant \left\lceil \frac{n}{2} \right\rceil$.
\end{lemma}

\begin{proof}
    Let $u$ be an external branching node in~$T$ that has degree at least~4. Then, $T \setminus~\{u\}$ has at least four connected components, and at most one is not a path. Let $P^{(1)}, P^{(2)}, P^{(3)}$ be three connected components of $T \setminus \{u\}$ that are paths, and for every $i \in \{1,2,3\}$, let $m_i:=|P^{(i)}|$. Let $Y$ be the set of vertices in $V \setminus \bigcup_{i=1}^3 P^{(i)}$, and $m_0 := |Y|$. We have $n = \sum_{i=0}^3 m_i$.
    Finally, let $X$ be the set consisting of all the vertices of $Y$, plus one vertex out of two in $P^{(i)}$ for every $i \in \{1,2,3\}$, starting and ending by a vertex not in~$X$. Then, by a simple counting argument, $X$ is a bat set, because the number of components in $\mathcal{C}_X$ is equal to the number of vertices in the union of their boundaries plus two. See Figure~\ref{fig:bad_set_branching_node_deg4} for an example.

    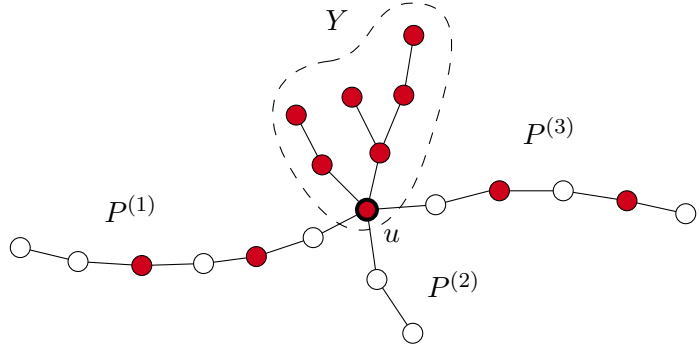
\begin{figure}[h!]
        \centering
        \begin{tikzpicture}[x=0.75pt,y=0.75pt,yscale=-1,xscale=1]
%uncomment if require: \path (0,536); %set diagram left start at 0, and has height of 536

%Straight Lines [id:da851038714323336] 
\draw    (289.5,91) -- (302.5,116) ;
%Straight Lines [id:da16393219942753956] 
\draw    (302.5,116) -- (325,138.5) ;
%Straight Lines [id:da4826670860212594] 
\draw    (485,140.5) -- (455.5,134) ;
%Straight Lines [id:da5470445038905127] 
\draw    (455.5,134) -- (423.5,129) ;
%Straight Lines [id:da05039107272369103] 
\draw    (180,164.5) -- (151,157.5) ;
%Straight Lines [id:da8503013335588411] 
\draw    (423.5,129) -- (391.5,129) ;
%Straight Lines [id:da6009179544890023] 
\draw    (330,173.5) -- (348,200.5) ;
%Straight Lines [id:da1996649595631972] 
\draw    (325,138.5) -- (330,173.5) ;
%Straight Lines [id:da9426518342970487] 
\draw    (343.5,81) -- (331.5,110) ;
%Straight Lines [id:da5425074944247688] 
\draw    (325,138.5) -- (298,152.5) ;
%Straight Lines [id:da8913814798034984] 
\draw    (269.5,162) -- (243,165.5) ;
%Straight Lines [id:da0481778897623355] 
\draw    (243,165.5) -- (212,166.5) ;
%Straight Lines [id:da21986512282467863] 
\draw    (212,166.5) -- (180,164.5) ;
%Straight Lines [id:da7372009613262848] 
\draw    (348.5,51) -- (343.5,81) ;
%Straight Lines [id:da5946088733752817] 
\draw    (317.5,82) -- (331.5,110) ;
%Straight Lines [id:da7353438497076654] 
\draw    (359.5,136) -- (325,138.5) ;
%Straight Lines [id:da3412271168120218] 
\draw    (331.5,110) -- (325,138.5) ;
%Straight Lines [id:da4020846975331607] 
\draw    (386.5,129) -- (359.5,136) ;
%Straight Lines [id:da587911896183378] 
\draw    (298,152.5) -- (269.5,162) ;
%Shape: Circle [id:dp13027301388570967] 
\draw  [fill={rgb, 255:red, 208; green, 2; blue, 27 }  ,fill opacity=1 ] (264.5,162) .. controls (264.5,159.24) and (266.74,157) .. (269.5,157) .. controls (272.26,157) and (274.5,159.24) .. (274.5,162) .. controls (274.5,164.76) and (272.26,167) .. (269.5,167) .. controls (266.74,167) and (264.5,164.76) .. (264.5,162) -- cycle ;
%Shape: Circle [id:dp09814865698431163] 
\draw  [fill={rgb, 255:red, 208; green, 2; blue, 27 }  ,fill opacity=1 ][line width=1.5]  (320,138.5) .. controls (320,135.74) and (322.24,133.5) .. (325,133.5) .. controls (327.76,133.5) and (330,135.74) .. (330,138.5) .. controls (330,141.26) and (327.76,143.5) .. (325,143.5) .. controls (322.24,143.5) and (320,141.26) .. (320,138.5) -- cycle ;
%Shape: Circle [id:dp32680019551275286] 
\draw  [fill={rgb, 255:red, 208; green, 2; blue, 27 }  ,fill opacity=1 ] (343.5,51) .. controls (343.5,48.24) and (345.74,46) .. (348.5,46) .. controls (351.26,46) and (353.5,48.24) .. (353.5,51) .. controls (353.5,53.76) and (351.26,56) .. (348.5,56) .. controls (345.74,56) and (343.5,53.76) .. (343.5,51) -- cycle ;
%Shape: Circle [id:dp5260350169407157] 
\draw  [fill={rgb, 255:red, 255; green, 255; blue, 255 }  ,fill opacity=1 ] (175,164.5) .. controls (175,161.74) and (177.24,159.5) .. (180,159.5) .. controls (182.76,159.5) and (185,161.74) .. (185,164.5) .. controls (185,167.26) and (182.76,169.5) .. (180,169.5) .. controls (177.24,169.5) and (175,167.26) .. (175,164.5) -- cycle ;
%Shape: Circle [id:dp4875756138769306] 
\draw  [fill={rgb, 255:red, 208; green, 2; blue, 27 }  ,fill opacity=1 ] (207,166.5) .. controls (207,163.74) and (209.24,161.5) .. (212,161.5) .. controls (214.76,161.5) and (217,163.74) .. (217,166.5) .. controls (217,169.26) and (214.76,171.5) .. (212,171.5) .. controls (209.24,171.5) and (207,169.26) .. (207,166.5) -- cycle ;
%Shape: Circle [id:dp5205936351763446] 
\draw  [fill={rgb, 255:red, 255; green, 255; blue, 255 }  ,fill opacity=1 ]  (238,165.5) .. controls (238,162.74) and (240.24,160.5) .. (243,160.5) .. controls (245.76,160.5) and (248,162.74) .. (248,165.5) .. controls (248,168.26) and (245.76,170.5) .. (243,170.5) .. controls (240.24,170.5) and (238,168.26) .. (238,165.5) -- cycle ;
%Shape: Circle [id:dp36982359412320487] 
\draw  [fill={rgb, 255:red, 208; green, 2; blue, 27 }  ,fill opacity=1 ] (338.5,81) .. controls (338.5,78.24) and (340.74,76) .. (343.5,76) .. controls (346.26,76) and (348.5,78.24) .. (348.5,81) .. controls (348.5,83.76) and (346.26,86) .. (343.5,86) .. controls (340.74,86) and (338.5,83.76) .. (338.5,81) -- cycle ;
%Shape: Circle [id:dp28342352667004034] 
\draw  [fill={rgb, 255:red, 208; green, 2; blue, 27 }  ,fill opacity=1 ] (312.5,82) .. controls (312.5,79.24) and (314.74,77) .. (317.5,77) .. controls (320.26,77) and (322.5,79.24) .. (322.5,82) .. controls (322.5,84.76) and (320.26,87) .. (317.5,87) .. controls (314.74,87) and (312.5,84.76) .. (312.5,82) -- cycle ;
%Shape: Circle [id:dp15102330475828252] 
\draw  [fill={rgb, 255:red, 208; green, 2; blue, 27 }  ,fill opacity=1 ] (386.5,129) .. controls (386.5,126.24) and (388.74,124) .. (391.5,124) .. controls (394.26,124) and (396.5,126.24) .. (396.5,129) .. controls (396.5,131.76) and (394.26,134) .. (391.5,134) .. controls (388.74,134) and (386.5,131.76) .. (386.5,129) -- cycle ;
%Shape: Circle [id:dp27658428026450665] 
\draw  [fill={rgb, 255:red, 208; green, 2; blue, 27 }  ,fill opacity=1 ] (326.5,110) .. controls (326.5,107.24) and (328.74,105) .. (331.5,105) .. controls (334.26,105) and (336.5,107.24) .. (336.5,110) .. controls (336.5,112.76) and (334.26,115) .. (331.5,115) .. controls (328.74,115) and (326.5,112.76) .. (326.5,110) -- cycle ;
%Shape: Circle [id:dp313871448575894] 
\draw  [fill={rgb, 255:red, 255; green, 255; blue, 255 }  ,fill opacity=1 ] (354.5,136) .. controls (354.5,133.24) and (356.74,131) .. (359.5,131) .. controls (362.26,131) and (364.5,133.24) .. (364.5,136) .. controls (364.5,138.76) and (362.26,141) .. (359.5,141) .. controls (356.74,141) and (354.5,138.76) .. (354.5,136) -- cycle ;
%Shape: Circle [id:dp9189953538315903] 
\draw  [fill={rgb, 255:red, 255; green, 255; blue, 255 }  ,fill opacity=1 ] (293,152.5) .. controls (293,149.74) and (295.24,147.5) .. (298,147.5) .. controls (300.76,147.5) and (303,149.74) .. (303,152.5) .. controls (303,155.26) and (300.76,157.5) .. (298,157.5) .. controls (295.24,157.5) and (293,155.26) .. (293,152.5) -- cycle ;
%Shape: Circle [id:dp16004628617394923] 
\draw  [fill={rgb, 255:red, 255; green, 255; blue, 255 }  ,fill opacity=1 ] (325,173.5) .. controls (325,170.74) and (327.24,168.5) .. (330,168.5) .. controls (332.76,168.5) and (335,170.74) .. (335,173.5) .. controls (335,176.26) and (332.76,178.5) .. (330,178.5) .. controls (327.24,178.5) and (325,176.26) .. (325,173.5) -- cycle ;
%Shape: Circle [id:dp024080458748005018] 
\draw  [fill={rgb, 255:red, 255; green, 255; blue, 255 }  ,fill opacity=1 ] (343,200.5) .. controls (343,197.74) and (345.24,195.5) .. (348,195.5) .. controls (350.76,195.5) and (353,197.74) .. (353,200.5) .. controls (353,203.26) and (350.76,205.5) .. (348,205.5) .. controls (345.24,205.5) and (343,203.26) .. (343,200.5) -- cycle ;
%Shape: Circle [id:dp4039083507931521] 
\draw  [fill={rgb, 255:red, 255; green, 255; blue, 255 }  ,fill opacity=1 ] (418.5,129) .. controls (418.5,126.24) and (420.74,124) .. (423.5,124) .. controls (426.26,124) and (428.5,126.24) .. (428.5,129) .. controls (428.5,131.76) and (426.26,134) .. (423.5,134) .. controls (420.74,134) and (418.5,131.76) .. (418.5,129) -- cycle ;
%Shape: Circle [id:dp8713270215866521] 
\draw  [fill={rgb, 255:red, 255; green, 255; blue, 255 }  ,fill opacity=1 ] (146,157.5) .. controls (146,154.74) and (148.24,152.5) .. (151,152.5) .. controls (153.76,152.5) and (156,154.74) .. (156,157.5) .. controls (156,160.26) and (153.76,162.5) .. (151,162.5) .. controls (148.24,162.5) and (146,160.26) .. (146,157.5) -- cycle ;
%Shape: Circle [id:dp40143366202498154] 
\draw  [fill={rgb, 255:red, 208; green, 2; blue, 27 }  ,fill opacity=1 ] (450.5,134) .. controls (450.5,131.24) and (452.74,129) .. (455.5,129) .. controls (458.26,129) and (460.5,131.24) .. (460.5,134) .. controls (460.5,136.76) and (458.26,139) .. (455.5,139) .. controls (452.74,139) and (450.5,136.76) .. (450.5,134) -- cycle ;
%Shape: Circle [id:dp4714757815202335] 
\draw  [fill={rgb, 255:red, 255; green, 255; blue, 255 }  ,fill opacity=1 ] (480,140.5) .. controls (480,137.74) and (482.24,135.5) .. (485,135.5) .. controls (487.76,135.5) and (490,137.74) .. (490,140.5) .. controls (490,143.26) and (487.76,145.5) .. (485,145.5) .. controls (482.24,145.5) and (480,143.26) .. (480,140.5) -- cycle ;
%Shape: Circle [id:dp40176221323208294] 
\draw  [fill={rgb, 255:red, 208; green, 2; blue, 27 }  ,fill opacity=1 ] (297.5,116) .. controls (297.5,113.24) and (299.74,111) .. (302.5,111) .. controls (305.26,111) and (307.5,113.24) .. (307.5,116) .. controls (307.5,118.76) and (305.26,121) .. (302.5,121) .. controls (299.74,121) and (297.5,118.76) .. (297.5,116) -- cycle ;
%Shape: Polygon Curved [id:ds4980664536804321] 
\draw  [dash pattern={on 4.5pt off 4.5pt}] (305.5,64) .. controls (330.5,61) and (324.5,32) .. (351.5,35) .. controls (378.5,38) and (365.5,76) .. (354.5,111) .. controls (343.5,146) and (320.5,167) .. (292.5,127) .. controls (264.5,87) and (280.5,67) .. (305.5,64) -- cycle ;
%Shape: Circle [id:dp41797765005221854] 
\draw  [fill={rgb, 255:red, 208; green, 2; blue, 27 }  ,fill opacity=1 ] (284.5,91) .. controls (284.5,88.24) and (286.74,86) .. (289.5,86) .. controls (292.26,86) and (294.5,88.24) .. (294.5,91) .. controls (294.5,93.76) and (292.26,96) .. (289.5,96) .. controls (286.74,96) and (284.5,93.76) .. (284.5,91) -- cycle ;

% Text Node
\draw (332,146.9) node [anchor=north west][inner sep=0.75pt]    {$u$};
% Text Node
\draw (303,36.9) node [anchor=north west][inner sep=0.75pt]    {$Y$};
% Text Node
\draw (192,129.9) node [anchor=north west][inner sep=0.75pt]    {$P^{( 1)}$};
% Text Node
\draw (402,90.9) node [anchor=north west][inner sep=0.75pt]    {$P^{( 3)}$};
% Text Node
\draw (354,167.9) node [anchor=north west][inner sep=0.75pt]    {$P^{( 2)}$};

\end{tikzpicture}

        \caption{In this example, $u$ is an external branching node which has degree at least~4. The nodes in~$X$ are colored in red. Here, we have $n = 20$, $m_0 = 7$, $m_1 = 6$, $m_2 = 2$ and $m_3 = 5$. There are seven components in $\mathcal{C}_X$, and five vertices in the union of their boundaries, so $X$ is a bad~set.}
        \label{fig:bad_set_branching_node_deg4}
    \end{figure}

    Since $X$ is a bad set, we have $\bad{T} \geqslant |X|$. Moreover, the number of vertices in $X$ is equal to the number of vertices in~$Y$, which is~$m_0$, plus the number of vertices selected in the paths $P^{(i)}$ for every $i \in \{1,2,3\}$, which is $\left\lfloor\frac{m_i-1}{2}\right\rfloor$. Thus, we get:
    
    \begin{align*}
        \bad{T}+2 &\geqslant m_0 + \left(\sum_{i=1}^{3}\left\lfloor\frac{m_i-1}{2}\right\rfloor\right) + 2 \\
        &\geqslant m_0 + \left(\sum_{i=1}^{3}\frac{m_i}{2}-1\right) + 2\\
        &\geqslant \frac{n}{2} + \frac{m_0}{2} - 1
    \end{align*}

    Finally, since $m_0 \geqslant 2$ (because $u$ has degree at least~$4$), and since $\bad{T}$ is an integer, we get $\bad{T} + 2 \geqslant \left\lceil\frac{n}{2}\right\rceil$.
\end{proof}

\begin{lemma}
    \label{lem:3_branching_nodes}
    If $T$ has at least three branching nodes, or if $T$ has exactly two branching nodes which are at distance at least three, then $\bad{T}+2 \geqslant \left\lceil \frac{n}{2} \right\rceil$.
\end{lemma}

\begin{proof}
    Since T has at least two branching nodes, it has at least two external branching nodes (because the minimal subtree of~$T$ containing all branching nodes has at least two leaves).
    So let $u$, $v$ be two distinct external branching nodes.
    Let us consider four connected components of $T \setminus \{u,v\}$ that are paths, denoted by $P^{(1)}, P^{(2)}, P^{(3)}, P^{(4)}$. For each $i \in \{1, 2, 3, 4\}$, let us denote by $m_i$ the number of vertices in $P^{(i)}$. Let $Y$ be the set of vertices in $T \setminus \bigcup_{i=1}^{4} P^{(i)}$, and $m_0:=|Y|$. We have $n = \sum_{i=0}^{4} m_i$. Let $X$ be the set consisting of all the vertices of~$Y$, plus one vertex out of two in $P^{(i)}$ for every $i \in \{1, \ldots, 4\}$, starting and ending by a vertex not in~$X$. Then, by a simple counting argument, $X$ is a bad set, because the number of components in $\mathcal{C}_X$ is equal to the number of vertices in the union of their boundaries plus two. See Figure~\ref{fig:bad_set_2_branching_nodes} for an example.
    
    \begin{figure}[h!]
        \centering
        \begin{tikzpicture}[x=0.75pt,y=0.75pt,yscale=-1,xscale=1]
%uncomment if require: \path (0,536); %set diagram left start at 0, and has height of 536

%Straight Lines [id:da9702165553412809] 
\draw    (164,198.5) -- (130,195.5) ;
%Straight Lines [id:da8712876634855502] 
\draw    (480,149.5) -- (450.5,143) ;
%Straight Lines [id:da18850975550576543] 
\draw    (450.5,143) -- (423.5,134) ;
%Straight Lines [id:da8922777269533128] 
\draw    (195,194.5) -- (164,198.5) ;
%Straight Lines [id:da2856980863151598] 
\draw    (423.5,134) -- (391.5,129) ;
%Straight Lines [id:da28146744936138446] 
\draw    (330,173.5) -- (348,200.5) ;
%Straight Lines [id:da08235187315648096] 
\draw    (325,138.5) -- (330,173.5) ;
%Straight Lines [id:da31849648758689086] 
\draw    (269.5,162) -- (259,194.5) ;
%Straight Lines [id:da4544164313129484] 
\draw    (325,138.5) -- (298,152.5) ;
%Straight Lines [id:da14434498250607064] 
\draw    (269.5,162) -- (243,165.5) ;
%Straight Lines [id:da504623965247179] 
\draw    (243,165.5) -- (219,182.5) ;
%Straight Lines [id:da4679483422308077] 
\draw    (219,182.5) -- (195,194.5) ;
%Straight Lines [id:da697778168293573] 
\draw    (259,194.5) -- (242,217.5) ;
%Straight Lines [id:da2804907015199576] 
\draw    (259,194.5) -- (273,218.5) ;
%Straight Lines [id:da4050797062586843] 
\draw    (354.5,129) -- (325,138.5) ;
%Straight Lines [id:da08121297364481528] 
\draw    (210,144.5) -- (243,165.5) ;
%Straight Lines [id:da1070185839612725] 
\draw    (326.5,110) -- (325,138.5) ;
%Straight Lines [id:da7542915800712738] 
\draw    (386.5,129) -- (354.5,129) ;
%Straight Lines [id:da7823214474298615] 
\draw    (298,152.5) -- (269.5,162) ;
%Shape: Circle [id:dp5647601291546733] 
\draw  [fill={rgb, 255:red, 208; green, 2; blue, 27 }  ,fill opacity=1 ] (264.5,162) .. controls (264.5,159.24) and (266.74,157) .. (269.5,157) .. controls (272.26,157) and (274.5,159.24) .. (274.5,162) .. controls (274.5,164.76) and (272.26,167) .. (269.5,167) .. controls (266.74,167) and (264.5,164.76) .. (264.5,162) -- cycle ;
%Shape: Circle [id:dp8779460188493665] 
\draw  [fill={rgb, 255:red, 208; green, 2; blue, 27 }  ,fill opacity=1 ][line width=1.5]  (320,138.5) .. controls (320,135.74) and (322.24,133.5) .. (325,133.5) .. controls (327.76,133.5) and (330,135.74) .. (330,138.5) .. controls (330,141.26) and (327.76,143.5) .. (325,143.5) .. controls (322.24,143.5) and (320,141.26) .. (320,138.5) -- cycle ;
%Shape: Circle [id:dp7051885630524245] 
\draw  [fill={rgb, 255:red, 208; green, 2; blue, 27 }  ,fill opacity=1 ] (237,217.5) .. controls (237,214.74) and (239.24,212.5) .. (242,212.5) .. controls (244.76,212.5) and (247,214.74) .. (247,217.5) .. controls (247,220.26) and (244.76,222.5) .. (242,222.5) .. controls (239.24,222.5) and (237,220.26) .. (237,217.5) -- cycle ;
%Shape: Circle [id:dp3813179747465063] 
\draw  [fill={rgb, 255:red, 208; green, 2; blue, 27 }  ,fill opacity=1 ] (190,194.5) .. controls (190,191.74) and (192.24,189.5) .. (195,189.5) .. controls (197.76,189.5) and (200,191.74) .. (200,194.5) .. controls (200,197.26) and (197.76,199.5) .. (195,199.5) .. controls (192.24,199.5) and (190,197.26) .. (190,194.5) -- cycle ;
%Shape: Circle [id:dp600947876000403] 
\draw  [fill={rgb, 255:red, 255; green, 255; blue, 255 }  ,fill opacity=1 ] (214,182.5) .. controls (214,179.74) and (216.24,177.5) .. (219,177.5) .. controls (221.76,177.5) and (224,179.74) .. (224,182.5) .. controls (224,185.26) and (221.76,187.5) .. (219,187.5) .. controls (216.24,187.5) and (214,185.26) .. (214,182.5) -- cycle ;
%Shape: Circle [id:dp6080833655458162] 
\draw  [fill={rgb, 255:red, 208; green, 2; blue, 27 }  ,fill opacity=1 ][line width=1.5]  (238,165.5) .. controls (238,162.74) and (240.24,160.5) .. (243,160.5) .. controls (245.76,160.5) and (248,162.74) .. (248,165.5) .. controls (248,168.26) and (245.76,170.5) .. (243,170.5) .. controls (240.24,170.5) and (238,168.26) .. (238,165.5) -- cycle ;
%Shape: Circle [id:dp8414124630433598] 
\draw  [fill={rgb, 255:red, 208; green, 2; blue, 27 }  ,fill opacity=1 ] (254,194.5) .. controls (254,191.74) and (256.24,189.5) .. (259,189.5) .. controls (261.76,189.5) and (264,191.74) .. (264,194.5) .. controls (264,197.26) and (261.76,199.5) .. (259,199.5) .. controls (256.24,199.5) and (254,197.26) .. (254,194.5) -- cycle ;
%Shape: Circle [id:dp48112729727624604] 
\draw  [fill={rgb, 255:red, 208; green, 2; blue, 27 }  ,fill opacity=1 ] (268,218.5) .. controls (268,215.74) and (270.24,213.5) .. (273,213.5) .. controls (275.76,213.5) and (278,215.74) .. (278,218.5) .. controls (278,221.26) and (275.76,223.5) .. (273,223.5) .. controls (270.24,223.5) and (268,221.26) .. (268,218.5) -- cycle ;
%Shape: Circle [id:dp4755267185216989] 
\draw  [fill={rgb, 255:red, 255; green, 255; blue, 255 }  ,fill opacity=1 ] (205,144.5) .. controls (205,141.74) and (207.24,139.5) .. (210,139.5) .. controls (212.76,139.5) and (215,141.74) .. (215,144.5) .. controls (215,147.26) and (212.76,149.5) .. (210,149.5) .. controls (207.24,149.5) and (205,147.26) .. (205,144.5) -- cycle ;
%Shape: Circle [id:dp6250548784245415] 
\draw  [fill={rgb, 255:red, 208; green, 2; blue, 27 }  ,fill opacity=1 ] (386.5,129) .. controls (386.5,126.24) and (388.74,124) .. (391.5,124) .. controls (394.26,124) and (396.5,126.24) .. (396.5,129) .. controls (396.5,131.76) and (394.26,134) .. (391.5,134) .. controls (388.74,134) and (386.5,131.76) .. (386.5,129) -- cycle ;
%Shape: Circle [id:dp563946532444999] 
\draw  [fill={rgb, 255:red, 208; green, 2; blue, 27 }  ,fill opacity=1 ] (321.5,110) .. controls (321.5,107.24) and (323.74,105) .. (326.5,105) .. controls (329.26,105) and (331.5,107.24) .. (331.5,110) .. controls (331.5,112.76) and (329.26,115) .. (326.5,115) .. controls (323.74,115) and (321.5,112.76) .. (321.5,110) -- cycle ;
%Shape: Circle [id:dp09941219810338953] 
\draw  [fill={rgb, 255:red, 255; green, 255; blue, 255 }  ,fill opacity=1 ] (349.5,129) .. controls (349.5,126.24) and (351.74,124) .. (354.5,124) .. controls (357.26,124) and (359.5,126.24) .. (359.5,129) .. controls (359.5,131.76) and (357.26,134) .. (354.5,134) .. controls (351.74,134) and (349.5,131.76) .. (349.5,129) -- cycle ;
%Shape: Circle [id:dp4903061976229728] 
\draw  [fill={rgb, 255:red, 208; green, 2; blue, 27 }  ,fill opacity=1 ] (293,152.5) .. controls (293,149.74) and (295.24,147.5) .. (298,147.5) .. controls (300.76,147.5) and (303,149.74) .. (303,152.5) .. controls (303,155.26) and (300.76,157.5) .. (298,157.5) .. controls (295.24,157.5) and (293,155.26) .. (293,152.5) -- cycle ;
%Shape: Circle [id:dp9574691571282662] 
\draw  [fill={rgb, 255:red, 255; green, 255; blue, 255 }  ,fill opacity=1 ] (325,173.5) .. controls (325,170.74) and (327.24,168.5) .. (330,168.5) .. controls (332.76,168.5) and (335,170.74) .. (335,173.5) .. controls (335,176.26) and (332.76,178.5) .. (330,178.5) .. controls (327.24,178.5) and (325,176.26) .. (325,173.5) -- cycle ;
%Shape: Circle [id:dp8434245729912899] 
\draw  [fill={rgb, 255:red, 255; green, 255; blue, 255 }  ,fill opacity=1 ] (343,200.5) .. controls (343,197.74) and (345.24,195.5) .. (348,195.5) .. controls (350.76,195.5) and (353,197.74) .. (353,200.5) .. controls (353,203.26) and (350.76,205.5) .. (348,205.5) .. controls (345.24,205.5) and (343,203.26) .. (343,200.5) -- cycle ;
%Shape: Circle [id:dp38673981120812606] 
\draw  [fill={rgb, 255:red, 255; green, 255; blue, 255 }  ,fill opacity=1 ] (418.5,134) .. controls (418.5,131.24) and (420.74,129) .. (423.5,129) .. controls (426.26,129) and (428.5,131.24) .. (428.5,134) .. controls (428.5,136.76) and (426.26,139) .. (423.5,139) .. controls (420.74,139) and (418.5,136.76) .. (418.5,134) -- cycle ;
%Shape: Circle [id:dp4866696572778383] 
\draw  [fill={rgb, 255:red, 255; green, 255; blue, 255 }  ,fill opacity=1 ] (159,198.5) .. controls (159,195.74) and (161.24,193.5) .. (164,193.5) .. controls (166.76,193.5) and (169,195.74) .. (169,198.5) .. controls (169,201.26) and (166.76,203.5) .. (164,203.5) .. controls (161.24,203.5) and (159,201.26) .. (159,198.5) -- cycle ;
%Shape: Circle [id:dp846274151223042] 
\draw  [fill={rgb, 255:red, 208; green, 2; blue, 27 }  ,fill opacity=1 ] (445.5,143) .. controls (445.5,140.24) and (447.74,138) .. (450.5,138) .. controls (453.26,138) and (455.5,140.24) .. (455.5,143) .. controls (455.5,145.76) and (453.26,148) .. (450.5,148) .. controls (447.74,148) and (445.5,145.76) .. (445.5,143) -- cycle ;
%Shape: Circle [id:dp2570712328926511] 
\draw  [fill={rgb, 255:red, 255; green, 255; blue, 255 }  ,fill opacity=1 ] (475,149.5) .. controls (475,146.74) and (477.24,144.5) .. (480,144.5) .. controls (482.76,144.5) and (485,146.74) .. (485,149.5) .. controls (485,152.26) and (482.76,154.5) .. (480,154.5) .. controls (477.24,154.5) and (475,152.26) .. (475,149.5) -- cycle ;
%Shape: Circle [id:dp14581733617272408] 
\draw  [fill={rgb, 255:red, 255; green, 255; blue, 255 }  ,fill opacity=1 ] (125,195.5) .. controls (125,192.74) and (127.24,190.5) .. (130,190.5) .. controls (132.76,190.5) and (135,192.74) .. (135,195.5) .. controls (135,198.26) and (132.76,200.5) .. (130,200.5) .. controls (127.24,200.5) and (125,198.26) .. (125,195.5) -- cycle ;
%Shape: Polygon Curved [id:ds08968847609013941] 
\draw  [dash pattern={on 4.5pt off 4.5pt}] (300,136.5) .. controls (314,117.5) and (308,95.5) .. (326,93.5) .. controls (344,91.5) and (342,126.5) .. (335,140.5) .. controls (328,154.5) and (301,170.5) .. (284,174.5) .. controls (267,178.5) and (295,214.5) .. (289,227.5) .. controls (283,240.5) and (240,233.5) .. (230,225.5) .. controls (220,217.5) and (256,194.5) .. (249,183.5) .. controls (242,172.5) and (226,183.5) .. (229,166.5) .. controls (232,149.5) and (286,155.5) .. (300,136.5) -- cycle ;

% Text Node
\draw (244,144) node [anchor=north west][inner sep=0.75pt]    {$u$};
% Text Node
\draw (334,141.9) node [anchor=north west][inner sep=0.75pt]    {$v$};
% Text Node
\draw (297,234) node [anchor=north west][inner sep=0.75pt]    {$Y$};
% Text Node
\draw (117,160.9) node [anchor=north west][inner sep=0.75pt]    {$P^{( 1)}$};
% Text Node
\draw (193,106.9) node [anchor=north west][inner sep=0.75pt]    {$P^{( 2)}$};
% Text Node
\draw (417,96.9) node [anchor=north west][inner sep=0.75pt]    {$P^{( 3)}$};
% Text Node
\draw (359,167.9) node [anchor=north west][inner sep=0.75pt]    {$P^{( 4)}$};

\end{tikzpicture}

        \caption{In this example, $u$ and $v$ are two external branching nodes. The nodes in~$X$ are colored in red. Here, we have $n = 20$, $m_1 = 4$, $m_2 = 1$, $m_3 = 5$, $m_4 = 2$ and $m_5 = 8$. There are seven components in $\mathcal{C}_X$, and five vertices in the union of their boundaries, so $X$ is a bad~set.}
        \label{fig:bad_set_2_branching_nodes}
    \end{figure}
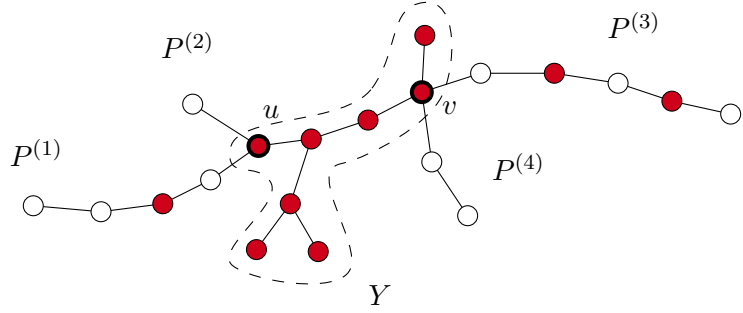

    Since $X$ is a bad set, we have $\bad{T} \geqslant |X|$. Moreover, the number of vertices in $X$ is equal to the number of vertices in~$Y$, which is~$m_0$, plus the number of vertices selected in the paths $P^{(i)}$ for every $i \in \{1,2,3,4\}$, which is $\left\lfloor\frac{m_i-1}{2}\right\rfloor$. Thus, we get:
    \begin{align*}
        \bad{T}+2 &\geqslant m_0 + \left(\sum_{i=1}^{4}\left\lfloor\frac{m_i-1}{2}\right\rfloor\right) + 2 \\
        &\geqslant m_0 + \left(\sum_{i=1}^{4}\frac{m_i}{2}-1\right) + 2\\
        &\geqslant \frac{n}{2} + \frac{m_0}{2} - 2
    \end{align*}

    If $m_0 \geqslant 4$, since $\bad{T}$ is an integer, we get $\bad{T}+2 \geqslant \left\lceil\frac{n}{2}\right\rceil$.
    This is the case if there are at least two other vertices than $u$ and~$v$ in~$Y$.
    In particular, this is the case if $u$ and~$v$ are at distance at least three, or if there is another branching node~$w$ (in this case, $Y$ also contains $w$ and at least one other vertex since $w$ is a branching node).
\end{proof}

\begin{lemma}
    \label{lem:H1}
    If $T$ is a $H_1$-graph, then $\bad{T}+2 \geqslant \left\lceil \frac{n}{2} \right\rceil$.
\end{lemma}

\begin{proof}
    Let us denote the four branches by $P^{(1)}, P^{(2)}, P^{(3)}, P^{(4)}$ and by $m_1, m_2, m_3, m_4$ their sizes. We have $n = 3 + \sum_{i=1}^4 m_i$.
    Let $X$ be the set of vertices containing the two branching nodes, all the vertices in $P^{(1)}$, and for every $i \in \{2,3,4\}$, one vertex out of two in $P^{(i)}$, starting and ending by a vertex not in~$X$. Then, by a simple counting argument, $X$ is a bad set, because the number of components in~$\mathcal{C}_X$ is equal to the number of vertices in the union of their boundaries plus two. See Figure~\ref{fig:H1} for an example.

    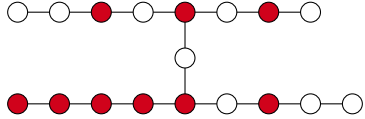
\begin{figure}[h!]
        \centering
        \begin{tikzpicture}[x=0.75pt,y=0.75pt,yscale=-1,xscale=1]

%Straight Lines [id:da1510382038271626] 
\draw    (466.5,157) -- (487.5,157) ;
%Straight Lines [id:da5599606569515323] 
\draw    (445.5,157) -- (466.5,157) ;
%Straight Lines [id:da13192369919648872] 
\draw    (319.5,157) -- (340.5,157) ;
%Straight Lines [id:da0867549944245668] 
\draw    (340.5,111) -- (361.5,111) ;
%Straight Lines [id:da30671996025673487] 
\draw    (319.5,111) -- (340.5,111) ;
%Straight Lines [id:da5842948004747573] 
\draw    (424.5,111) -- (445.5,111) ;
%Straight Lines [id:da7283896766154507] 
\draw    (445.5,111) -- (466.5,111) ;
%Straight Lines [id:da6008751066637773] 
\draw    (382.5,157) -- (403.5,157) ;
%Straight Lines [id:da34124311114560624] 
\draw    (361.5,157) -- (382.5,157) ;
%Straight Lines [id:da4689426224871478] 
\draw    (340.5,157) -- (361.5,157) ;
%Straight Lines [id:da44159274650311653] 
\draw    (403.5,157) -- (424.5,157) ;
%Straight Lines [id:da263906819065303] 
\draw    (424.5,157) -- (445.5,157) ;
%Straight Lines [id:da39493300795591424] 
\draw    (403.5,157) -- (403.5,134) ;
%Straight Lines [id:da09130044608623211] 
\draw    (403.5,134) -- (403.5,111) ;
%Straight Lines [id:da379498495153585] 
\draw    (382.5,111) -- (403.5,111) ;
%Straight Lines [id:da6366725359245404] 
\draw    (361.5,111) -- (382.5,111) ;
%Straight Lines [id:da4187383528593672] 
\draw    (403.5,111) -- (424.5,111) ;
%Shape: Circle [id:dp48582983027777726] 
\draw  [fill={rgb, 255:red, 208; green, 2; blue, 27 }  ,fill opacity=1 ] (361.5,162) .. controls (358.74,162) and (356.5,159.76) .. (356.5,157) .. controls (356.5,154.24) and (358.74,152) .. (361.5,152) .. controls (364.26,152) and (366.5,154.24) .. (366.5,157) .. controls (366.5,159.76) and (364.26,162) .. (361.5,162) -- cycle ;
%Shape: Circle [id:dp2229579680723598] 
\draw  [fill={rgb, 255:red, 208; green, 2; blue, 27 }  ,fill opacity=1 ] (340.5,162) .. controls (337.74,162) and (335.5,159.76) .. (335.5,157) .. controls (335.5,154.24) and (337.74,152) .. (340.5,152) .. controls (343.26,152) and (345.5,154.24) .. (345.5,157) .. controls (345.5,159.76) and (343.26,162) .. (340.5,162) -- cycle ;
%Shape: Circle [id:dp7320716439041451] 
\draw  [fill={rgb, 255:red, 208; green, 2; blue, 27 }  ,fill opacity=1 ] (445.5,162) .. controls (442.74,162) and (440.5,159.76) .. (440.5,157) .. controls (440.5,154.24) and (442.74,152) .. (445.5,152) .. controls (448.26,152) and (450.5,154.24) .. (450.5,157) .. controls (450.5,159.76) and (448.26,162) .. (445.5,162) -- cycle ;
%Shape: Circle [id:dp5604273094042708] 
\draw  [fill={rgb, 255:red, 255; green, 255; blue, 255 }  ,fill opacity=1 ] (424.5,162) .. controls (421.74,162) and (419.5,159.76) .. (419.5,157) .. controls (419.5,154.24) and (421.74,152) .. (424.5,152) .. controls (427.26,152) and (429.5,154.24) .. (429.5,157) .. controls (429.5,159.76) and (427.26,162) .. (424.5,162) -- cycle ;
%Shape: Circle [id:dp6568586810063457] 
\draw  [fill={rgb, 255:red, 208; green, 2; blue, 27 }  ,fill opacity=1 ] (382.5,162) .. controls (379.74,162) and (377.5,159.76) .. (377.5,157) .. controls (377.5,154.24) and (379.74,152) .. (382.5,152) .. controls (385.26,152) and (387.5,154.24) .. (387.5,157) .. controls (387.5,159.76) and (385.26,162) .. (382.5,162) -- cycle ;
%Shape: Circle [id:dp08375944795584778] 
\draw  [fill={rgb, 255:red, 208; green, 2; blue, 27 }  ,fill opacity=1 ] (403.5,162) .. controls (400.74,162) and (398.5,159.76) .. (398.5,157) .. controls (398.5,154.24) and (400.74,152) .. (403.5,152) .. controls (406.26,152) and (408.5,154.24) .. (408.5,157) .. controls (408.5,159.76) and (406.26,162) .. (403.5,162) -- cycle ;
%Shape: Circle [id:dp4084817011893256] 
\draw  [fill={rgb, 255:red, 255; green, 255; blue, 255 }  ,fill opacity=1 ] (403.5,139) .. controls (400.74,139) and (398.5,136.76) .. (398.5,134) .. controls (398.5,131.24) and (400.74,129) .. (403.5,129) .. controls (406.26,129) and (408.5,131.24) .. (408.5,134) .. controls (408.5,136.76) and (406.26,139) .. (403.5,139) -- cycle ;
%Shape: Circle [id:dp9733872613256555] 
\draw  [fill={rgb, 255:red, 208; green, 2; blue, 27 }  ,fill opacity=1 ] (403.5,116) .. controls (400.74,116) and (398.5,113.76) .. (398.5,111) .. controls (398.5,108.24) and (400.74,106) .. (403.5,106) .. controls (406.26,106) and (408.5,108.24) .. (408.5,111) .. controls (408.5,113.76) and (406.26,116) .. (403.5,116) -- cycle ;
%Shape: Circle [id:dp4469920284745468] 
\draw  [fill={rgb, 255:red, 255; green, 255; blue, 255 }  ,fill opacity=1 ] (382.5,116) .. controls (379.74,116) and (377.5,113.76) .. (377.5,111) .. controls (377.5,108.24) and (379.74,106) .. (382.5,106) .. controls (385.26,106) and (387.5,108.24) .. (387.5,111) .. controls (387.5,113.76) and (385.26,116) .. (382.5,116) -- cycle ;
%Shape: Circle [id:dp17377819134393346] 
\draw  [fill={rgb, 255:red, 208; green, 2; blue, 27 }  ,fill opacity=1 ] (361.5,116) .. controls (358.74,116) and (356.5,113.76) .. (356.5,111) .. controls (356.5,108.24) and (358.74,106) .. (361.5,106) .. controls (364.26,106) and (366.5,108.24) .. (366.5,111) .. controls (366.5,113.76) and (364.26,116) .. (361.5,116) -- cycle ;
%Shape: Circle [id:dp6466578772752772] 
\draw  [fill={rgb, 255:red, 255; green, 255; blue, 255 }  ,fill opacity=1 ] (424.5,116) .. controls (421.74,116) and (419.5,113.76) .. (419.5,111) .. controls (419.5,108.24) and (421.74,106) .. (424.5,106) .. controls (427.26,106) and (429.5,108.24) .. (429.5,111) .. controls (429.5,113.76) and (427.26,116) .. (424.5,116) -- cycle ;
%Shape: Circle [id:dp6028999798423749] 
\draw  [fill={rgb, 255:red, 208; green, 2; blue, 27 }  ,fill opacity=1 ] (319.5,162) .. controls (316.74,162) and (314.5,159.76) .. (314.5,157) .. controls (314.5,154.24) and (316.74,152) .. (319.5,152) .. controls (322.26,152) and (324.5,154.24) .. (324.5,157) .. controls (324.5,159.76) and (322.26,162) .. (319.5,162) -- cycle ;
%Shape: Circle [id:dp3851050995331229] 
\draw  [fill={rgb, 255:red, 255; green, 255; blue, 255 }  ,fill opacity=1 ] (340.5,116) .. controls (337.74,116) and (335.5,113.76) .. (335.5,111) .. controls (335.5,108.24) and (337.74,106) .. (340.5,106) .. controls (343.26,106) and (345.5,108.24) .. (345.5,111) .. controls (345.5,113.76) and (343.26,116) .. (340.5,116) -- cycle ;
%Shape: Circle [id:dp26268161435358617] 
\draw  [fill={rgb, 255:red, 255; green, 255; blue, 255 }  ,fill opacity=1 ] (319.5,116) .. controls (316.74,116) and (314.5,113.76) .. (314.5,111) .. controls (314.5,108.24) and (316.74,106) .. (319.5,106) .. controls (322.26,106) and (324.5,108.24) .. (324.5,111) .. controls (324.5,113.76) and (322.26,116) .. (319.5,116) -- cycle ;
%Shape: Circle [id:dp22389007500907787] 
\draw  [fill={rgb, 255:red, 255; green, 255; blue, 255 }  ,fill opacity=1 ] (466.5,116) .. controls (463.74,116) and (461.5,113.76) .. (461.5,111) .. controls (461.5,108.24) and (463.74,106) .. (466.5,106) .. controls (469.26,106) and (471.5,108.24) .. (471.5,111) .. controls (471.5,113.76) and (469.26,116) .. (466.5,116) -- cycle ;
%Shape: Circle [id:dp04173384497647581] 
\draw  [fill={rgb, 255:red, 208; green, 2; blue, 27 }  ,fill opacity=1 ] (445.5,116) .. controls (442.74,116) and (440.5,113.76) .. (440.5,111) .. controls (440.5,108.24) and (442.74,106) .. (445.5,106) .. controls (448.26,106) and (450.5,108.24) .. (450.5,111) .. controls (450.5,113.76) and (448.26,116) .. (445.5,116) -- cycle ;
%Shape: Circle [id:dp7778221988712029] 
\draw  [fill={rgb, 255:red, 255; green, 255; blue, 255 }  ,fill opacity=1 ] (466.5,162) .. controls (463.74,162) and (461.5,159.76) .. (461.5,157) .. controls (461.5,154.24) and (463.74,152) .. (466.5,152) .. controls (469.26,152) and (471.5,154.24) .. (471.5,157) .. controls (471.5,159.76) and (469.26,162) .. (466.5,162) -- cycle ;
%Shape: Circle [id:dp35517307565770884] 
\draw  [fill={rgb, 255:red, 255; green, 255; blue, 255 }  ,fill opacity=1 ] (487.5,162) .. controls (484.74,162) and (482.5,159.76) .. (482.5,157) .. controls (482.5,154.24) and (484.74,152) .. (487.5,152) .. controls (490.26,152) and (492.5,154.24) .. (492.5,157) .. controls (492.5,159.76) and (490.26,162) .. (487.5,162) -- cycle ;

\end{tikzpicture}

        \caption{Illustration of the proof of Lemma~\ref{lem:H1}. Here, we have $n=18$, and the size of the branches are~$4,4,4,3$. The nodes in~$X$ are colored in red, and $X$ is a bad set.}
        \label{fig:H1}
    \end{figure}

    Since $X$ is a bad set, we have $\bad{T} \geqslant |X|$. Moreover, the number of vertices in $X$ is equal to the number of vertices in~$P^{(1)}$, which is~$m_1$, plus the number of vertices selected in the paths $P^{(i)}$ for every $i \in \{2,3,4\}$, which is $\left\lfloor\frac{m_i-1}{2}\right\rfloor$, plus~$2$ for the two branching nodes. Thus, we get:
    \begin{align*}
        \bad{T}+2 &\geqslant m_1 + \left(\sum_{i=2}^{4}\left\lfloor\frac{m_i-1}{2}\right\rfloor\right) + 4 \\
        &\geqslant m_1 + \left(\sum_{i=2}^{4}\frac{m_i}{2}-1\right) + 4\\
        &\geqslant \frac{n}{2} + \frac{m_1-1}{2}
    \end{align*}

    Finally, since $m_1 \geqslant 1$ and since $\bad{T}$ is an integer, we get $\bad{T}+2 \geqslant \left\lceil\frac{n}{2}\right\rceil$.
\end{proof}

\begin{lemma}
    \label{lem:H0}
    If $T$ is a $H_0$-graph, then $\bad{T}+2 \geqslant \left\lceil \frac{n}{2} \right\rceil$.
\end{lemma}

\begin{proof}
    If $T$ is a $H_0$-graph, let~$P^{(1)}$ and~$P^{(2)}$ be the two paths obtained by removing the central edge, and let~$m_1$ and~$m_2$ denote their respective lengths. Let $u_1$ (resp. $u_2$) be the vertex in~$P^{(1)}$ (resp. in $P^{(2)}$) that is an endpoint of the central edge. We have~$n = m_1+m_2$. There are two cases, which are illustrated on Figure~\ref{fig:H0}.

    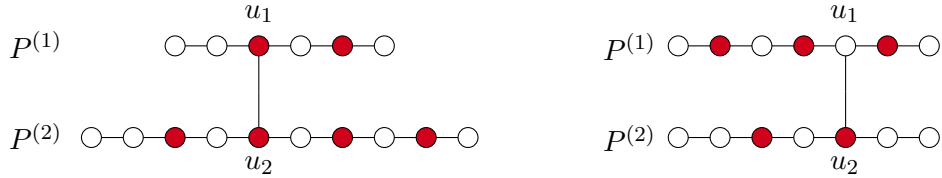
\begin{figure}[h!]
        \centering
        \begin{tikzpicture}[x=0.75pt,y=0.75pt,yscale=-1,xscale=1]

%Straight Lines [id:da9166760626305058] 
\draw    (198.5,94) -- (219.5,94) ;
%Straight Lines [id:da405199732220453] 
\draw    (177.5,94) -- (198.5,94) ;
%Straight Lines [id:da6724575890446367] 
\draw    (156.5,94) -- (177.5,94) ;
%Straight Lines [id:da5981779935485758] 
\draw    (30.5,94) -- (51.5,94) ;
%Straight Lines [id:da8107261454300966] 
\draw    (135.5,48) -- (156.5,48) ;
%Straight Lines [id:da017512730429570245] 
\draw    (156.5,48) -- (177.5,48) ;
%Straight Lines [id:da34358019007526275] 
\draw    (93.5,94) -- (114.5,94) ;
%Straight Lines [id:da26810840341971864] 
\draw    (72.5,94) -- (93.5,94) ;
%Straight Lines [id:da4558244562370106] 
\draw    (51.5,94) -- (72.5,94) ;
%Straight Lines [id:da48490578270332985] 
\draw    (114.5,94) -- (135.5,94) ;
%Straight Lines [id:da3397978249385405] 
\draw    (135.5,94) -- (156.5,94) ;
%Straight Lines [id:da3780116886317858] 
\draw    (114.5,94) -- (114.5,71) ;
%Straight Lines [id:da2938209536660572] 
\draw    (114.5,71) -- (114.5,48) ;
%Straight Lines [id:da5409540528607487] 
\draw    (93.5,48) -- (114.5,48) ;
%Straight Lines [id:da4427688962757882] 
\draw    (72.5,48) -- (93.5,48) ;
%Straight Lines [id:da4242441374108158] 
\draw    (114.5,48) -- (135.5,48) ;
%Shape: Circle [id:dp18651834111424936] 
\draw  [fill={rgb, 255:red, 208; green, 2; blue, 27 }  ,fill opacity=1 ] (72.5,99) .. controls (69.74,99) and (67.5,96.76) .. (67.5,94) .. controls (67.5,91.24) and (69.74,89) .. (72.5,89) .. controls (75.26,89) and (77.5,91.24) .. (77.5,94) .. controls (77.5,96.76) and (75.26,99) .. (72.5,99) -- cycle ;
%Shape: Circle [id:dp5021760268861571] 
\draw  [fill={rgb, 255:red, 255; green, 255; blue, 255 }  ,fill opacity=1 ] (51.5,99) .. controls (48.74,99) and (46.5,96.76) .. (46.5,94) .. controls (46.5,91.24) and (48.74,89) .. (51.5,89) .. controls (54.26,89) and (56.5,91.24) .. (56.5,94) .. controls (56.5,96.76) and (54.26,99) .. (51.5,99) -- cycle ;
%Shape: Circle [id:dp06861896374294851] 
\draw  [fill={rgb, 255:red, 208; green, 2; blue, 27 }  ,fill opacity=1 ] (156.5,99) .. controls (153.74,99) and (151.5,96.76) .. (151.5,94) .. controls (151.5,91.24) and (153.74,89) .. (156.5,89) .. controls (159.26,89) and (161.5,91.24) .. (161.5,94) .. controls (161.5,96.76) and (159.26,99) .. (156.5,99) -- cycle ;
%Shape: Circle [id:dp7152014833018251] 
\draw  [fill={rgb, 255:red, 255; green, 255; blue, 255 }  ,fill opacity=1 ] (135.5,99) .. controls (132.74,99) and (130.5,96.76) .. (130.5,94) .. controls (130.5,91.24) and (132.74,89) .. (135.5,89) .. controls (138.26,89) and (140.5,91.24) .. (140.5,94) .. controls (140.5,96.76) and (138.26,99) .. (135.5,99) -- cycle ;
%Shape: Circle [id:dp7470575730874089] 
\draw  [fill={rgb, 255:red, 255; green, 255; blue, 255 }  ,fill opacity=1 ] (93.5,99) .. controls (90.74,99) and (88.5,96.76) .. (88.5,94) .. controls (88.5,91.24) and (90.74,89) .. (93.5,89) .. controls (96.26,89) and (98.5,91.24) .. (98.5,94) .. controls (98.5,96.76) and (96.26,99) .. (93.5,99) -- cycle ;
%Shape: Circle [id:dp1442130853157928] 
\draw  [fill={rgb, 255:red, 208; green, 2; blue, 27 }  ,fill opacity=1 ] (114.5,99) .. controls (111.74,99) and (109.5,96.76) .. (109.5,94) .. controls (109.5,91.24) and (111.74,89) .. (114.5,89) .. controls (117.26,89) and (119.5,91.24) .. (119.5,94) .. controls (119.5,96.76) and (117.26,99) .. (114.5,99) -- cycle ;
%Shape: Circle [id:dp7883382766894546] 
\draw  [fill={rgb, 255:red, 208; green, 2; blue, 27 }  ,fill opacity=1 ] (114.5,53) .. controls (111.74,53) and (109.5,50.76) .. (109.5,48) .. controls (109.5,45.24) and (111.74,43) .. (114.5,43) .. controls (117.26,43) and (119.5,45.24) .. (119.5,48) .. controls (119.5,50.76) and (117.26,53) .. (114.5,53) -- cycle ;
%Shape: Circle [id:dp6551603697838927] 
\draw  [fill={rgb, 255:red, 255; green, 255; blue, 255 }  ,fill opacity=1 ] (93.5,53) .. controls (90.74,53) and (88.5,50.76) .. (88.5,48) .. controls (88.5,45.24) and (90.74,43) .. (93.5,43) .. controls (96.26,43) and (98.5,45.24) .. (98.5,48) .. controls (98.5,50.76) and (96.26,53) .. (93.5,53) -- cycle ;
%Shape: Circle [id:dp9076336313776601] 
\draw  [fill={rgb, 255:red, 255; green, 255; blue, 255 }  ,fill opacity=1 ] (72.5,53) .. controls (69.74,53) and (67.5,50.76) .. (67.5,48) .. controls (67.5,45.24) and (69.74,43) .. (72.5,43) .. controls (75.26,43) and (77.5,45.24) .. (77.5,48) .. controls (77.5,50.76) and (75.26,53) .. (72.5,53) -- cycle ;
%Shape: Circle [id:dp5500530314285211] 
\draw  [fill={rgb, 255:red, 255; green, 255; blue, 255 }  ,fill opacity=1 ] (135.5,53) .. controls (132.74,53) and (130.5,50.76) .. (130.5,48) .. controls (130.5,45.24) and (132.74,43) .. (135.5,43) .. controls (138.26,43) and (140.5,45.24) .. (140.5,48) .. controls (140.5,50.76) and (138.26,53) .. (135.5,53) -- cycle ;
%Shape: Circle [id:dp8736085417338938] 
\draw  [fill={rgb, 255:red, 255; green, 255; blue, 255 }  ,fill opacity=1 ] (30.5,99) .. controls (27.74,99) and (25.5,96.76) .. (25.5,94) .. controls (25.5,91.24) and (27.74,89) .. (30.5,89) .. controls (33.26,89) and (35.5,91.24) .. (35.5,94) .. controls (35.5,96.76) and (33.26,99) .. (30.5,99) -- cycle ;
%Shape: Circle [id:dp5644438377652462] 
\draw  [fill={rgb, 255:red, 255; green, 255; blue, 255 }  ,fill opacity=1 ] (177.5,53) .. controls (174.74,53) and (172.5,50.76) .. (172.5,48) .. controls (172.5,45.24) and (174.74,43) .. (177.5,43) .. controls (180.26,43) and (182.5,45.24) .. (182.5,48) .. controls (182.5,50.76) and (180.26,53) .. (177.5,53) -- cycle ;
%Shape: Circle [id:dp07179887065945623] 
\draw  [fill={rgb, 255:red, 208; green, 2; blue, 27 }  ,fill opacity=1 ] (156.5,53) .. controls (153.74,53) and (151.5,50.76) .. (151.5,48) .. controls (151.5,45.24) and (153.74,43) .. (156.5,43) .. controls (159.26,43) and (161.5,45.24) .. (161.5,48) .. controls (161.5,50.76) and (159.26,53) .. (156.5,53) -- cycle ;
%Shape: Circle [id:dp4150977916270726] 
\draw  [fill={rgb, 255:red, 255; green, 255; blue, 255 }  ,fill opacity=1 ] (177.5,99) .. controls (174.74,99) and (172.5,96.76) .. (172.5,94) .. controls (172.5,91.24) and (174.74,89) .. (177.5,89) .. controls (180.26,89) and (182.5,91.24) .. (182.5,94) .. controls (182.5,96.76) and (180.26,99) .. (177.5,99) -- cycle ;
%Shape: Circle [id:dp0408741870438144] 
\draw  [fill={rgb, 255:red, 208; green, 2; blue, 27 }  ,fill opacity=1 ] (198.5,99) .. controls (195.74,99) and (193.5,96.76) .. (193.5,94) .. controls (193.5,91.24) and (195.74,89) .. (198.5,89) .. controls (201.26,89) and (203.5,91.24) .. (203.5,94) .. controls (203.5,96.76) and (201.26,99) .. (198.5,99) -- cycle ;
%Shape: Circle [id:dp03394673000432713] 
\draw  [fill={rgb, 255:red, 255; green, 255; blue, 255 }  ,fill opacity=1 ] (219.5,99) .. controls (216.74,99) and (214.5,96.76) .. (214.5,94) .. controls (214.5,91.24) and (216.74,89) .. (219.5,89) .. controls (222.26,89) and (224.5,91.24) .. (224.5,94) .. controls (224.5,96.76) and (222.26,99) .. (219.5,99) -- cycle ;

%Straight Lines [id:da1593408591682356] 
\draw    (325,94) -- (346,94) ;
%Straight Lines [id:da06968641654117036] 
\draw    (346,94) -- (367,94) ;
%Straight Lines [id:da9080696094249888] 
\draw    (346,48) -- (367,48) ;
%Straight Lines [id:da002487642310386895] 
\draw    (325,48) -- (346,48) ;
%Straight Lines [id:da10721328547207631] 
\draw    (430,48) -- (451,48) ;
%Straight Lines [id:da008019555727570515] 
\draw    (388,94) -- (409,94) ;
%Straight Lines [id:da6265038794944834] 
\draw    (367,94) -- (388,94) ;
%Straight Lines [id:da4510882535025523] 
\draw    (409,94) -- (430,94) ;
%Straight Lines [id:da34224393230831407] 
\draw    (430,94) -- (451,94) ;
%Straight Lines [id:da42090447516410634] 
\draw    (409,94) -- (409,71) ;
%Straight Lines [id:da7558780747005893] 
\draw    (409,71) -- (409,48) ;
%Straight Lines [id:da9893606353432844] 
\draw    (388,48) -- (409,48) ;
%Straight Lines [id:da34924085678353856] 
\draw    (367,48) -- (388,48) ;
%Straight Lines [id:da15046640848599613] 
\draw    (409,48) -- (430,48) ;
%Shape: Circle [id:dp9494209751831042] 
\draw  [fill={rgb, 255:red, 208; green, 2; blue, 27 }  ,fill opacity=1 ] (367,99) .. controls (364.24,99) and (362,96.76) .. (362,94) .. controls (362,91.24) and (364.24,89) .. (367,89) .. controls (369.76,89) and (372,91.24) .. (372,94) .. controls (372,96.76) and (369.76,99) .. (367,99) -- cycle ;
%Shape: Circle [id:dp2716080156959121] 
\draw  [fill={rgb, 255:red, 255; green, 255; blue, 255 }  ,fill opacity=1 ] (346,99) .. controls (343.24,99) and (341,96.76) .. (341,94) .. controls (341,91.24) and (343.24,89) .. (346,89) .. controls (348.76,89) and (351,91.24) .. (351,94) .. controls (351,96.76) and (348.76,99) .. (346,99) -- cycle ;
%Shape: Circle [id:dp9463102386772745] 
\draw  [fill={rgb, 255:red, 255; green, 255; blue, 255 }  ,fill opacity=1 ] (430,99) .. controls (427.24,99) and (425,96.76) .. (425,94) .. controls (425,91.24) and (427.24,89) .. (430,89) .. controls (432.76,89) and (435,91.24) .. (435,94) .. controls (435,96.76) and (432.76,99) .. (430,99) -- cycle ;
%Shape: Circle [id:dp883377759316855] 
\draw  [fill={rgb, 255:red, 255; green, 255; blue, 255 }  ,fill opacity=1 ] (388,99) .. controls (385.24,99) and (383,96.76) .. (383,94) .. controls (383,91.24) and (385.24,89) .. (388,89) .. controls (390.76,89) and (393,91.24) .. (393,94) .. controls (393,96.76) and (390.76,99) .. (388,99) -- cycle ;
%Shape: Circle [id:dp5489451117540562] 
\draw  [fill={rgb, 255:red, 208; green, 2; blue, 27 }  ,fill opacity=1 ] (409,99) .. controls (406.24,99) and (404,96.76) .. (404,94) .. controls (404,91.24) and (406.24,89) .. (409,89) .. controls (411.76,89) and (414,91.24) .. (414,94) .. controls (414,96.76) and (411.76,99) .. (409,99) -- cycle ;
%Shape: Circle [id:dp33600466081317726] 
\draw  [fill={rgb, 255:red, 255; green, 255; blue, 255 }  ,fill opacity=1 ] (409,53) .. controls (406.24,53) and (404,50.76) .. (404,48) .. controls (404,45.24) and (406.24,43) .. (409,43) .. controls (411.76,43) and (414,45.24) .. (414,48) .. controls (414,50.76) and (411.76,53) .. (409,53) -- cycle ;
%Shape: Circle [id:dp0027182324017188675] 
\draw  [fill={rgb, 255:red, 208; green, 2; blue, 27 }  ,fill opacity=1 ] (388,53) .. controls (385.24,53) and (383,50.76) .. (383,48) .. controls (383,45.24) and (385.24,43) .. (388,43) .. controls (390.76,43) and (393,45.24) .. (393,48) .. controls (393,50.76) and (390.76,53) .. (388,53) -- cycle ;
%Shape: Circle [id:dp2665021381083943] 
\draw  [fill={rgb, 255:red, 255; green, 255; blue, 255 }  ,fill opacity=1 ] (367,53) .. controls (364.24,53) and (362,50.76) .. (362,48) .. controls (362,45.24) and (364.24,43) .. (367,43) .. controls (369.76,43) and (372,45.24) .. (372,48) .. controls (372,50.76) and (369.76,53) .. (367,53) -- cycle ;
%Shape: Circle [id:dp5727637827193055] 
\draw  [fill={rgb, 255:red, 208; green, 2; blue, 27 }  ,fill opacity=1 ] (430,53) .. controls (427.24,53) and (425,50.76) .. (425,48) .. controls (425,45.24) and (427.24,43) .. (430,43) .. controls (432.76,43) and (435,45.24) .. (435,48) .. controls (435,50.76) and (432.76,53) .. (430,53) -- cycle ;
%Shape: Circle [id:dp2743646721934506] 
\draw  [fill={rgb, 255:red, 208; green, 2; blue, 27 }  ,fill opacity=1 ] (346,53) .. controls (343.24,53) and (341,50.76) .. (341,48) .. controls (341,45.24) and (343.24,43) .. (346,43) .. controls (348.76,43) and (351,45.24) .. (351,48) .. controls (351,50.76) and (348.76,53) .. (346,53) -- cycle ;
%Shape: Circle [id:dp9176894859200173] 
\draw  [fill={rgb, 255:red, 255; green, 255; blue, 255 }  ,fill opacity=1 ] (325,53) .. controls (322.24,53) and (320,50.76) .. (320,48) .. controls (320,45.24) and (322.24,43) .. (325,43) .. controls (327.76,43) and (330,45.24) .. (330,48) .. controls (330,50.76) and (327.76,53) .. (325,53) -- cycle ;
%Shape: Circle [id:dp43546999774887707] 
\draw  [fill={rgb, 255:red, 255; green, 255; blue, 255 }  ,fill opacity=1 ] (451,53) .. controls (448.24,53) and (446,50.76) .. (446,48) .. controls (446,45.24) and (448.24,43) .. (451,43) .. controls (453.76,43) and (456,45.24) .. (456,48) .. controls (456,50.76) and (453.76,53) .. (451,53) -- cycle ;
%Shape: Circle [id:dp9017902678019508] 
\draw  [fill={rgb, 255:red, 255; green, 255; blue, 255 }  ,fill opacity=1 ] (451,99) .. controls (448.24,99) and (446,96.76) .. (446,94) .. controls (446,91.24) and (448.24,89) .. (451,89) .. controls (453.76,89) and (456,91.24) .. (456,94) .. controls (456,96.76) and (453.76,99) .. (451,99) -- cycle ;
%Shape: Circle [id:dp6039309796793231] 
\draw  [fill={rgb, 255:red, 255; green, 255; blue, 255 }  ,fill opacity=1 ] (325,99) .. controls (322.24,99) and (320,96.76) .. (320,94) .. controls (320,91.24) and (322.24,89) .. (325,89) .. controls (327.76,89) and (330,91.24) .. (330,94) .. controls (330,96.76) and (327.76,99) .. (325,99) -- cycle ;

% Text Node
\draw (-12,39.4) node [anchor=north west][inner sep=0.75pt]    {$P^{( 1)}$};
% Text Node
\draw (-12,85.4) node [anchor=north west][inner sep=0.75pt]    {$P^{( 2)}$};
% Text Node
\draw (286,40.4) node [anchor=north west][inner sep=0.75pt]    {$P^{( 1)}$};
% Text Node
\draw (286,86.4) node [anchor=north west][inner sep=0.75pt]    {$P^{( 2)}$};
% Text Node
\draw (106,26.4) node [anchor=north west][inner sep=0.75pt]    {$u_{1}$};
% Text Node
\draw (400,26.4) node [anchor=north west][inner sep=0.75pt]    {$u_{1}$};
% Text Node
\draw (106,102.4) node [anchor=north west][inner sep=0.75pt]    {$u_{2}$};
% Text Node
\draw (400,102.4) node [anchor=north west][inner sep=0.75pt]    {$u_{2}$};

\end{tikzpicture}

        \caption{Illustration of the proof of Lemma~\ref{lem:H0}. Left: the first case, $m_1$ and~$m_2$ are both even. Right: the second case, $m_1$ is odd. On each picture, the vertices in~$X$ are the red vertices.}
        \label{fig:H0}
    \end{figure}
    
    \begin{itemize}
        \item If $m_1$ and $m_2$ are both even, then we construct a bad set~$X$ as follows. For each~$i \in \{1,2\}$, we add in~$X$ one vertex of~$P^{(i)}$ out of two, with the constraint that both endpoints are not in~$X$, and $u_i$ is in~$X$ (see Figure~\ref{fig:H0}). Then, $X$ is a bad set and contains~$\frac{m_1}{2} + \frac{m_2}{2}-2$ vertices. We thus get~$\bad{T}+2 \geqslant \left\lceil \frac{n}{2} \right\rceil$.
        \item If $m_1$ is odd, we construct a bad set~$X$ as follows. We put one vertex out of two in~$P^{(1)}$, starting and ending by a vertex not in~$X$. The number of vertices in~$P^{(1)}$ we put is~$X$ is $\frac{m_1-1}{2}$. In~$P^{(2)}$, we also put one vertex out of two in~$X$, starting and ending by a vertex not in~$X$, with the constraint that~$u_2$ is in~$X$ (see Figure~\ref{fig:H0}). The number of vertices in~$P^{(2)}$ added in~$X$ is at least $\left\lfloor \frac{m_2}{2} \right\rfloor-1$. Thus, we obtain~$|X| \geqslant \frac{m_1}{2} + \frac{m_2}{2} - 2$. So we finally get $\bad{T}+2 \geqslant \left\lceil \frac{n}{2}\right\rceil$.\qedhere
    \end{itemize}

\end{proof}

\begin{lemma}
    \label{lem:subdivided_claw}
    Let~$T$ be a subdivided claw. We have $\bad{T}+2 \geqslant \left \lfloor \frac{n}{2}\right \rfloor$. Moreover, if at least one branch has odd size, then $\bad{T}+2 \geqslant \left \lceil \frac{n}{2}\right \rceil$.
\end{lemma}

\begin{proof}
Let us denote by $P^{(1)}, P^{(2)}, P^{(3)}$ the three branches of~$T$, and by~$m_1, m_2, m_3$ their respective lengths. We have~$n = 1 + \sum_{i=1}^3 m_i$. Let~$u$ be the branching node. Let us consider the following bad set~$X$. In~$X$, we put the vertex~$u$, and then for every~$i \in \{1,2,3\}$, we add one vertex of $P^{(i)}$ out of two in~$X$, starting and ending by a vertex not in~$X$. See Figure~\ref{fig:subdivided_claw} for an example. The set~$X$ is a bad set, so we have $\bad{T} \geqslant |X|$, and its size is the sum of the number of vertices selected in~$P^{(i)}$ for every~$i \in \{1,2,3\}$, which is $\left\lfloor\frac{m_i-1}{2}\right\rfloor$, plus one for~$u$. So we have:

\begin{align}
    \bad{T} + 2
    &\geqslant 3 + \sum_{i=1}^3 \left\lfloor\frac{m_i-1}{2}\right\rfloor \label{eq:inegality1} \\
    &\geqslant 3 + \sum_{i=1}^3\left(\frac{m_i}{2}-1\right) \label{eq:inegality2}\\
    &\geqslant \frac{n}{2} - \frac{1}{2} \label{eq:inegality3}
\end{align}

Since $\bad{T}$ is an integer, inequality~(\ref{eq:inegality3}) implies $\bad{T}+2 \geqslant \left\lfloor\frac{n}{2}\right\rfloor$. Moreover, if at least one branch has odd size, assume by symmetry that it is~$P^{(1)}$. The number of vertices in~$P^{(1)}$ that are in~$X$ is then $\frac{m_1-1}{2}$. So we can replace $\left(\frac{m_1}{2}-1\right)$ by~$\frac{m_1-1}{2}$ in inequality~(\ref{eq:inegality2}), and inequality~(\ref{eq:inegality3}) becomes~$\bad{T}+2 \geqslant \frac{n}{2}$. Since $\bad{T}$ is an integer, we finally get~\mbox{$\bad{T}+2 \geqslant \left\lceil\frac{n}{2}\right\rceil$} in that case. \qedhere

\begin{figure}[h!]
    \centering
    \begin{tikzpicture}[x=0.75pt,y=0.75pt,yscale=-1,xscale=1]

%Straight Lines [id:da24139526426190117] 
\draw    (216.5,215.5) -- (190.5,230) ;
%Straight Lines [id:da730942269893526] 
\draw    (398.5,230) -- (424.5,244.5) ;
%Straight Lines [id:da15404446098358793] 
\draw    (372.5,215.5) -- (398.5,230) ;
%Straight Lines [id:da3945960816543814] 
\draw    (268.5,186.5) -- (242.5,201) ;
%Straight Lines [id:da9943710151998266] 
\draw    (242.5,201) -- (216.5,215.5) ;
%Straight Lines [id:da8428396824284238] 
\draw    (320.5,186.5) -- (346.5,201) ;
%Straight Lines [id:da16200037657482946] 
\draw    (346.5,201) -- (372.5,215.5) ;
%Straight Lines [id:da6004045196702029] 
\draw    (294.5,123) -- (294.5,147.5) ;
%Straight Lines [id:da005538361667071334] 
\draw    (294.5,172) -- (268.5,186.5) ;
%Straight Lines [id:da3011704503555923] 
\draw    (294.5,172) -- (320.5,186.5) ;
%Straight Lines [id:da6421647061188668] 
\draw    (294.5,147.5) -- (294.5,172) ;
%Shape: Circle [id:dp5341915597622731] 
\draw  [fill={rgb, 255:red, 208; green, 2; blue, 27 }  ,fill opacity=1 ] (289.5,172) .. controls (289.5,169.24) and (291.74,167) .. (294.5,167) .. controls (297.26,167) and (299.5,169.24) .. (299.5,172) .. controls (299.5,174.76) and (297.26,177) .. (294.5,177) .. controls (291.74,177) and (289.5,174.76) .. (289.5,172) -- cycle ;
%Shape: Circle [id:dp7288285223008791] 
\draw  [fill={rgb, 255:red, 208; green, 2; blue, 27 }  ,fill opacity=1 ] (341.5,201) .. controls (341.5,198.24) and (343.74,196) .. (346.5,196) .. controls (349.26,196) and (351.5,198.24) .. (351.5,201) .. controls (351.5,203.76) and (349.26,206) .. (346.5,206) .. controls (343.74,206) and (341.5,203.76) .. (341.5,201) -- cycle ;
%Shape: Circle [id:dp6023800848189296] 
\draw  [fill={rgb, 255:red, 255; green, 255; blue, 255 }  ,fill opacity=1 ] (263.5,186.5) .. controls (263.5,183.74) and (265.74,181.5) .. (268.5,181.5) .. controls (271.26,181.5) and (273.5,183.74) .. (273.5,186.5) .. controls (273.5,189.26) and (271.26,191.5) .. (268.5,191.5) .. controls (265.74,191.5) and (263.5,189.26) .. (263.5,186.5) -- cycle ;
%Shape: Circle [id:dp9061692833049527] 
\draw  [fill={rgb, 255:red, 255; green, 255; blue, 255 }  ,fill opacity=1 ] (289.5,147.5) .. controls (289.5,144.74) and (291.74,142.5) .. (294.5,142.5) .. controls (297.26,142.5) and (299.5,144.74) .. (299.5,147.5) .. controls (299.5,150.26) and (297.26,152.5) .. (294.5,152.5) .. controls (291.74,152.5) and (289.5,150.26) .. (289.5,147.5) -- cycle ;
%Shape: Circle [id:dp8945639251977581] 
\draw  [fill={rgb, 255:red, 255; green, 255; blue, 255 }  ,fill opacity=1 ] (289.5,123) .. controls (289.5,120.24) and (291.74,118) .. (294.5,118) .. controls (297.26,118) and (299.5,120.24) .. (299.5,123) .. controls (299.5,125.76) and (297.26,128) .. (294.5,128) .. controls (291.74,128) and (289.5,125.76) .. (289.5,123) -- cycle ;
%Shape: Circle [id:dp0893144556479436] 
\draw  [fill={rgb, 255:red, 255; green, 255; blue, 255 }  ,fill opacity=1 ] (315.5,186.5) .. controls (315.5,183.74) and (317.74,181.5) .. (320.5,181.5) .. controls (323.26,181.5) and (325.5,183.74) .. (325.5,186.5) .. controls (325.5,189.26) and (323.26,191.5) .. (320.5,191.5) .. controls (317.74,191.5) and (315.5,189.26) .. (315.5,186.5) -- cycle ;
%Shape: Circle [id:dp3793904030001152] 
\draw  [fill={rgb, 255:red, 255; green, 255; blue, 255 }  ,fill opacity=1 ] (367.5,215.5) .. controls (367.5,212.74) and (369.74,210.5) .. (372.5,210.5) .. controls (375.26,210.5) and (377.5,212.74) .. (377.5,215.5) .. controls (377.5,218.26) and (375.26,220.5) .. (372.5,220.5) .. controls (369.74,220.5) and (367.5,218.26) .. (367.5,215.5) -- cycle ;
%Shape: Circle [id:dp821786526995522] 
\draw  [fill={rgb, 255:red, 208; green, 2; blue, 27 }  ,fill opacity=1 ] (237.5,201) .. controls (237.5,198.24) and (239.74,196) .. (242.5,196) .. controls (245.26,196) and (247.5,198.24) .. (247.5,201) .. controls (247.5,203.76) and (245.26,206) .. (242.5,206) .. controls (239.74,206) and (237.5,203.76) .. (237.5,201) -- cycle ;
%Shape: Circle [id:dp9850607409089521] 
\draw  [fill={rgb, 255:red, 255; green, 255; blue, 255 }  ,fill opacity=1 ] (211.5,215.5) .. controls (211.5,212.74) and (213.74,210.5) .. (216.5,210.5) .. controls (219.26,210.5) and (221.5,212.74) .. (221.5,215.5) .. controls (221.5,218.26) and (219.26,220.5) .. (216.5,220.5) .. controls (213.74,220.5) and (211.5,218.26) .. (211.5,215.5) -- cycle ;
%Shape: Circle [id:dp02172742463742916] 
\draw  [fill={rgb, 255:red, 208; green, 2; blue, 27 }  ,fill opacity=1 ] (393.5,230) .. controls (393.5,227.24) and (395.74,225) .. (398.5,225) .. controls (401.26,225) and (403.5,227.24) .. (403.5,230) .. controls (403.5,232.76) and (401.26,235) .. (398.5,235) .. controls (395.74,235) and (393.5,232.76) .. (393.5,230) -- cycle ;
%Shape: Circle [id:dp8946104756441109] 
\draw  [fill={rgb, 255:red, 255; green, 255; blue, 255 }  ,fill opacity=1 ] (419.5,244.5) .. controls (419.5,241.74) and (421.74,239.5) .. (424.5,239.5) .. controls (427.26,239.5) and (429.5,241.74) .. (429.5,244.5) .. controls (429.5,247.26) and (427.26,249.5) .. (424.5,249.5) .. controls (421.74,249.5) and (419.5,247.26) .. (419.5,244.5) -- cycle ;
%Shape: Circle [id:dp9082623116162253] 
\draw  [fill={rgb, 255:red, 255; green, 255; blue, 255 }  ,fill opacity=1 ] (185.5,230) .. controls (185.5,227.24) and (187.74,225) .. (190.5,225) .. controls (193.26,225) and (195.5,227.24) .. (195.5,230) .. controls (195.5,232.76) and (193.26,235) .. (190.5,235) .. controls (187.74,235) and (185.5,232.76) .. (185.5,230) -- cycle ;

\end{tikzpicture}

    \caption{Illustration of the proof of Lemma~\ref{lem:subdivided_claw}. The vertices in~$X$ are the red vertices.}
    \label{fig:subdivided_claw}
\end{figure}
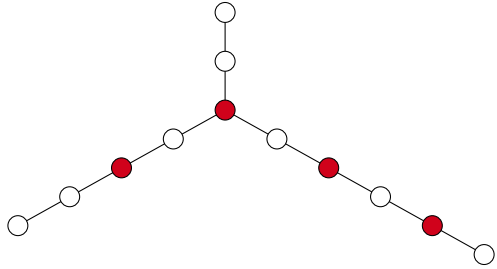
\end{proof}

\begin{proof}[Proof of Proposition~\ref{prop:bad_set_at_least_n/2}.]
Let~$T$ be a proper $n$-vertex tree. By Lemma~\ref{lem:3_branching_nodes}, we know that if there are at least three branching nodes, or if there are two branching nodes at distance at least~$3$, then the result holds. So we can assume that there are either one or two branching nodes, and if there are two, that they are at distance at most~$2$. In particular, all branching nodes are external. By Lemma~\ref{lem:branching_node_deg_4}, we know that if one of the branching nodes has degree at least~$4$, then the result holds. So we can assume that all branching nodes have degree exactly~$3$. The remaining cases are thus $H_0$-graphs, $H_1$-graphs, and subdivided claws. Finally, we can simply use Lemmas~\ref{lem:H1}, \ref{lem:H0} and~\ref{lem:subdivided_claw} to conclude the proof.
\end{proof}

We can now finish the proof of Theorem~\ref{thm:binary}.
In Sections~\ref{sec:overview} to~\ref{sec:splitting}, we assume that~$T$ is not a subdivided claw with all branches of even size. We treat this case separately afterwards, in Appendix~\ref{app:subdivided_claw}.

\subsubsection{Overview of our Algorithm}
\label{sec:overview}

Since~$T$ is not a subdivided claw with all branches of even size, by Proposition~\ref{prop:bad_set_at_least_n/2}, we have $\bad{T}+2 \geqslant \left\lceil\frac{n}{2}\right\rceil$, so $n_1 \leqslant n \leqslant 2n_1$.
Our goal is to create a partition of~$T$ into~$k$ subtrees, such that each subtree either contains a unique vertex in~$V_1$, or contains exactly two vertices in~$V_1$ and one of them is a leaf. Every such partition is an~\EFO allocation.
To do so, we apply an algorithm having three main steps, which are the following ones:
\begin{enumerate}
    \item (Partitioning step) We create a partition of~$T$ into subtrees of size~$1$ or~$2$, such that the total number of subtrees is at most~$k$.
    \item (Merging step) We merge subtrees that contain only vertices in~$V_0$ to adjacent subtrees, until each subtree contains at least one vertex in~$V_1$.
    \item (Splitting step) We split some subtrees if needed, to obtain, at the end of the process, exactly~$k$ subtrees satisfying the requirements above.
\end{enumerate}

Let us detail each of these steps separately.

\subsubsection{Partitioning Step}
\label{sec:partitioning}

Let us first focus on step~$1$, the partitioning step. We apply the following recursive procedure, by distinguishing cases as follows, until~$T$ is empty. At the same time, our procedure creates a set~$X$ of vertices. See Figure~\ref{fig:partition} for an example.
\begin{enumerate}[label=(\alph*)]
\item If~$T$ has a single vertex~$u$, we add~$\{u\}$ to the partition and we delete~$u$ from~$T$.
\item If $T$ contains a leaf~$u$ such that the neighbor~$v$ of $u$ is not a branching node, we add the set~$\{u,v\}$ to the partition, we add~$v$ to~$X$ and we delete these two vertices from~$T$.
\item Else, let~$v$ be an external branching node. By definition, in~$T \setminus \{v\}$, at most one connected component is not a path. Moreover, every such path is a single vertex, otherwise we are in the previous case. Let us denote the degree of~$v$ by~$d \geqslant 3$. Then, $v$ has at least~$d-1$ adjacent leaves.  Let~$u$ be some leaf adjacent to~$v$, chosen arbitrarily. We add the set~$\{u,v\}$ to the partition. For every other leaf~$w$ adjacent to~$v$, we add the set~$\{w\}$ to the partition. Finally, we add~$v$ to~$X$, and we delete $v$ and all its adjacent leaves from~$T$.
\end{enumerate}

\begin{figure}[h!]
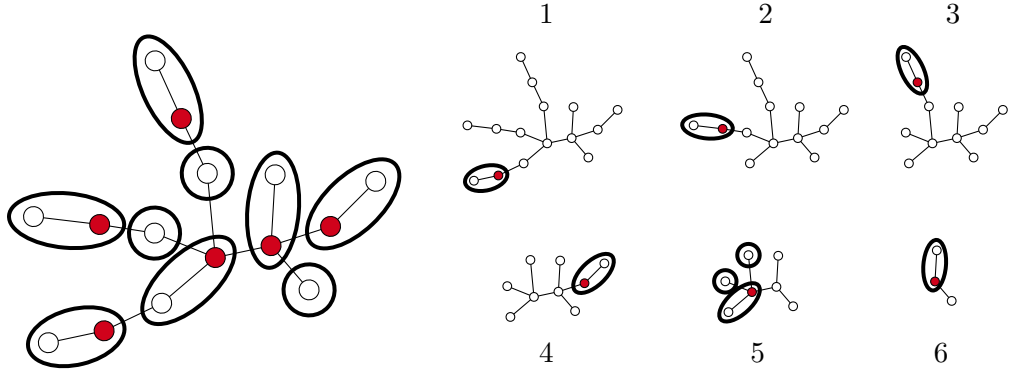

    \centering



    \caption{Left: a tree~$T$ and the partition into sets of size~$1$ or~$2$ obtained with our procedure. Right: execution of our partitioning procedure on that tree. The vertices in~$X$ are colored in red.}
    \label{fig:partition}
\end{figure}

Let us prove that the number of subtrees in that partition is at most~$k$. Let $t_1$ (resp. $t_2$) be the number of subtrees of size~$1$ (resp. of size~$2$) in the partition.
We have $n = t_1 + 2t_2$, and we want to show that $k \geqslant t_1 + t_2$.
If $t_1 \in \{0,1\}$, we have $t_1 + t_2 =  \left\lceil\frac{n}{2}\right\rceil$. Since $k \geqslant \bad{T}+2$, Proposition~\ref{prop:bad_set_at_least_n/2} immediately implies that $k \geqslant t_1+t_2$.
Let us now assume that $t_1 \geqslant 2$.

\begin{claim}
    \label{claim:T-X indep}
    The graph $T \setminus X$ is an independent set (where~$X$ is the set constructed in our procedure).
\end{claim}

\begin{proof}
    We prove it by induction on~$|T|$. It is true if~$T$ is empty. Else, there are three cases, depending on which case of our procedure is applied. For each of the three cases, let us use the same notations as in our procedure.
    
    If we are in case~(a), the result holds because~$T$ is a single vertex. If we are in case~(b), let $T':=T \setminus\{u,v\}$. The vertex~$v$ is added to~$X$, and~$u$ does not have any neighbor in~$T'$. Thus, the result holds by applying the induction hypothesis for~$T'$. Finally, if we are in case~(c), let us denote by~$L_v$ the set of leaves adjacent to~$v$, and let~$T' := T \setminus (\{v\} \cup L_v)$. The vertex~$v$ is added to~$X$, and none of the vertices in~$L_v$ have any neighbor in~$T'$. Thus, the result holds by applying the induction hypothesis for~$T'$.
\end{proof}

By construction of~$X$, we have $|X| = t_2$. Let~$X'$ be a set containing $t_1-2$ arbitrary vertices in~$T \setminus X$. Let us prove that $X \cup X'$ is a bad set. We have $|X \cup X'| = t_1 + t_2 - 2$, and $|T \setminus (X \cup X')| = t_2+2$. Since $T \setminus X$ is an independent set by Claim~\ref{claim:T-X indep}, it is also the case of $T \setminus (X \cup X')$. Thus, the number of connected components in $\mathcal{C}_{X \cup X'}$ is the number of vertices in $T \setminus (X \cup X')$, that is, $t_2+2$. Moreover, by Claim~\ref{claim:T-X indep}, the vertices in $T \setminus X$ have all their neighbors in $X$. Thus, the union of the border of the vertices in $T \setminus (X \cup X')$ is included in~$X$, so its size is at most~$t_2$. Hence, by a simple cardinality argument, $X \cup X'$ is a bad set. Finally, we get:
$$k \geqslant \bad{T}+2 \geqslant |X \cup X'|+2 \geqslant t_1 + t_2,$$
which is precisely the inequality we wanted to prove.

\subsubsection{Merging Step}
\label{sec:merging}

Let us now turn to step~2, the merging step. Let us consider the partition obtained at the end of step~$1$. If there is any subtree of this partition which is included in~$V_0$, we merge it with an adjacent subtree (chosen arbitrarily among the adjacent subtrees). We repeat this until all subtrees contain at least one vertex from~$V_1$. Note that the maximum number of vertices of~$V_1$ in a subtree of the partition is invariant during the merging step, because when merging two subtrees, at least one is included in~$V_0$. Thus, at the end of the merging step, all subtrees contain either one or two vertices in~$V_1$. See Figure~\ref{fig:merging} for an illustration.

\begin{figure}[h!]
    \centering
    \begin{tikzpicture}[x=0.75pt,y=0.75pt,yscale=-0.9,xscale=0.9]
%uncomment if require: \path (0,443); %set diagram left start at 0, and has height of 443

%Straight Lines [id:da8456000181735318] 
\draw    (157.7,149.45) -- (155.65,185) ;
%Straight Lines [id:da9045520520253254] 
\draw    (127.68,190.98) -- (155.65,185) ;
%Straight Lines [id:da10957620326546602] 
\draw    (123.65,148.72) -- (127.68,190.98) ;
%Straight Lines [id:da6675039598868313] 
\draw    (155.65,185) -- (174.73,207.2) ;
%Straight Lines [id:da23921096903266836] 
\draw    (155.65,185) -- (185.51,175.36) ;
%Straight Lines [id:da6345240956177087] 
\draw    (208.2,152.7) -- (185.51,175.36) ;
%Straight Lines [id:da3840427709571075] 
\draw    (100.76,213.9) -- (127.68,190.98) ;
%Straight Lines [id:da6352455643639959] 
\draw    (97.16,178.7) -- (127.68,190.98) ;
%Straight Lines [id:da9903963521020633] 
\draw    (69.62,174.38) -- (97.16,178.7) ;
%Straight Lines [id:da33733608335515086] 
\draw    (43.89,233.71) -- (71.86,227.73) ;
%Straight Lines [id:da6456068866162207] 
\draw    (36.43,170.41) -- (69.62,174.38) ;
%Shape: Circle [id:dp34732667702843223] 
\draw  [fill={rgb, 255:red, 255; green, 255; blue, 255 }  ,fill opacity=1 ] (33.12,174.16) .. controls (31.05,172.33) and (30.86,169.17) .. (32.68,167.1) .. controls (34.51,165.03) and (37.67,164.84) .. (39.74,166.66) .. controls (41.81,168.49) and (42.01,171.65) .. (40.18,173.72) .. controls (38.35,175.79) and (35.19,175.99) .. (33.12,174.16) -- cycle ;
%Straight Lines [id:da6661205390942416] 
\draw    (71.86,227.73) -- (100.76,213.9) ;
%Shape: Circle [id:dp6014240296668381] 
\draw  [fill={rgb, 255:red, 0; green, 0; blue, 0 }  ,fill opacity=1 ] (97.45,217.65) .. controls (95.38,215.82) and (95.18,212.66) .. (97.01,210.59) .. controls (98.84,208.52) and (102,208.32) .. (104.07,210.15) .. controls (106.14,211.98) and (106.33,215.14) .. (104.5,217.21) .. controls (102.68,219.28) and (99.52,219.47) .. (97.45,217.65) -- cycle ;
%Shape: Circle [id:dp816698340985456] 
\draw  [fill={rgb, 255:red, 255; green, 255; blue, 255 }  ,fill opacity=1 ] (68.55,231.48) .. controls (66.48,229.65) and (66.28,226.49) .. (68.11,224.42) .. controls (69.94,222.35) and (73.1,222.16) .. (75.17,223.98) .. controls (77.24,225.81) and (77.43,228.97) .. (75.6,231.04) .. controls (73.78,233.11) and (70.62,233.31) .. (68.55,231.48) -- cycle ;
%Shape: Circle [id:dp5177904997367696] 
\draw  [fill={rgb, 255:red, 0; green, 0; blue, 0 }  ,fill opacity=1 ] (124.37,194.73) .. controls (122.3,192.9) and (122.1,189.74) .. (123.93,187.67) .. controls (125.76,185.6) and (128.92,185.4) .. (130.99,187.23) .. controls (133.06,189.06) and (133.25,192.22) .. (131.43,194.29) .. controls (129.6,196.36) and (126.44,196.56) .. (124.37,194.73) -- cycle ;
%Shape: Circle [id:dp4271046719164028] 
\draw  [fill={rgb, 255:red, 255; green, 255; blue, 255 }  ,fill opacity=1 ] (93.85,182.45) .. controls (91.78,180.62) and (91.59,177.46) .. (93.42,175.39) .. controls (95.24,173.32) and (98.4,173.12) .. (100.47,174.95) .. controls (102.54,176.78) and (102.74,179.94) .. (100.91,182.01) .. controls (99.08,184.08) and (95.92,184.27) .. (93.85,182.45) -- cycle ;
%Shape: Circle [id:dp3741869670689263] 
\draw  [fill={rgb, 255:red, 255; green, 255; blue, 255 }  ,fill opacity=1 ] (40.58,237.46) .. controls (38.51,235.63) and (38.31,232.47) .. (40.14,230.4) .. controls (41.97,228.33) and (45.13,228.14) .. (47.2,229.96) .. controls (49.27,231.79) and (49.46,234.95) .. (47.64,237.02) .. controls (45.81,239.09) and (42.65,239.29) .. (40.58,237.46) -- cycle ;
%Shape: Circle [id:dp5303194858629644] 
\draw  [fill={rgb, 255:red, 255; green, 255; blue, 255 }  ,fill opacity=1 ] (66.31,178.13) .. controls (64.24,176.3) and (64.04,173.14) .. (65.87,171.07) .. controls (67.7,169) and (70.86,168.81) .. (72.93,170.63) .. controls (75,172.46) and (75.2,175.62) .. (73.37,177.69) .. controls (71.54,179.76) and (68.38,179.96) .. (66.31,178.13) -- cycle ;
%Straight Lines [id:da4285466170356156] 
\draw    (97.48,92.25) -- (110.61,121.19) ;
%Straight Lines [id:da1207697791829122] 
\draw    (123.65,148.72) -- (110.61,121.19) ;
%Shape: Circle [id:dp3637642439637653] 
\draw  [fill={rgb, 255:red, 0; green, 0; blue, 0 }  ,fill opacity=1 ] (120.34,152.47) .. controls (118.27,150.64) and (118.07,147.48) .. (119.9,145.41) .. controls (121.73,143.34) and (124.89,143.15) .. (126.96,144.97) .. controls (129.03,146.8) and (129.22,149.96) .. (127.39,152.03) .. controls (125.57,154.1) and (122.41,154.3) .. (120.34,152.47) -- cycle ;
%Shape: Circle [id:dp6896165646886178] 
\draw  [fill={rgb, 255:red, 0; green, 0; blue, 0 }  ,fill opacity=1 ] (152.34,188.75) .. controls (150.27,186.92) and (150.07,183.76) .. (151.9,181.69) .. controls (153.73,179.62) and (156.89,179.42) .. (158.96,181.25) .. controls (161.03,183.08) and (161.22,186.24) .. (159.39,188.31) .. controls (157.57,190.38) and (154.41,190.57) .. (152.34,188.75) -- cycle ;
%Shape: Circle [id:dp9378278477854591] 
\draw  [fill={rgb, 255:red, 0; green, 0; blue, 0 }  ,fill opacity=1 ] (94.17,95.99) .. controls (92.1,94.17) and (91.91,91.01) .. (93.74,88.94) .. controls (95.56,86.87) and (98.72,86.67) .. (100.79,88.5) .. controls (102.86,90.33) and (103.06,93.49) .. (101.23,95.56) .. controls (99.4,97.63) and (96.24,97.82) .. (94.17,95.99) -- cycle ;
%Shape: Circle [id:dp3488059780093318] 
\draw  [fill={rgb, 255:red, 0; green, 0; blue, 0 }  ,fill opacity=1 ] (107.3,124.94) .. controls (105.23,123.11) and (105.03,119.95) .. (106.86,117.88) .. controls (108.69,115.81) and (111.85,115.61) .. (113.92,117.44) .. controls (115.99,119.27) and (116.18,122.43) .. (114.36,124.5) .. controls (112.53,126.57) and (109.37,126.77) .. (107.3,124.94) -- cycle ;
%Shape: Circle [id:dp8468577520851289] 
\draw  [fill={rgb, 255:red, 255; green, 255; blue, 255 }  ,fill opacity=1 ] (171.42,210.94) .. controls (169.35,209.12) and (169.16,205.96) .. (170.98,203.89) .. controls (172.81,201.82) and (175.97,201.62) .. (178.04,203.45) .. controls (180.11,205.28) and (180.31,208.44) .. (178.48,210.51) .. controls (176.65,212.58) and (173.49,212.77) .. (171.42,210.94) -- cycle ;
%Shape: Circle [id:dp07032550683368066] 
\draw  [fill={rgb, 255:red, 255; green, 255; blue, 255 }  ,fill opacity=1 ] (182.2,179.11) .. controls (180.13,177.28) and (179.94,174.12) .. (181.77,172.05) .. controls (183.59,169.98) and (186.75,169.78) .. (188.82,171.61) .. controls (190.89,173.44) and (191.09,176.6) .. (189.26,178.67) .. controls (187.43,180.74) and (184.27,180.93) .. (182.2,179.11) -- cycle ;
%Shape: Circle [id:dp6875419394188359] 
\draw  [fill={rgb, 255:red, 0; green, 0; blue, 0 }  ,fill opacity=1 ] (204.89,156.45) .. controls (202.82,154.62) and (202.62,151.46) .. (204.45,149.39) .. controls (206.28,147.32) and (209.44,147.13) .. (211.51,148.95) .. controls (213.58,150.78) and (213.77,153.94) .. (211.95,156.01) .. controls (210.12,158.08) and (206.96,158.28) .. (204.89,156.45) -- cycle ;
%Shape: Circle [id:dp7486926307138599] 
\draw  [fill={rgb, 255:red, 0; green, 0; blue, 0 }  ,fill opacity=1 ] (154.39,153.2) .. controls (152.32,151.37) and (152.12,148.21) .. (153.95,146.14) .. controls (155.78,144.07) and (158.94,143.87) .. (161.01,145.7) .. controls (163.08,147.53) and (163.27,150.69) .. (161.44,152.76) .. controls (159.62,154.83) and (156.46,155.02) .. (154.39,153.2) -- cycle ;
%Shape: Ellipse [id:dp7420259089149404] 
\draw  [line width=1.5]  (33.88,236.58) .. controls (32.05,229.09) and (41.31,220.39) .. (54.56,217.15) .. controls (67.81,213.92) and (80.04,217.37) .. (81.87,224.86) .. controls (83.7,232.35) and (74.44,241.05) .. (61.19,244.29) .. controls (47.94,247.53) and (35.71,244.08) .. (33.88,236.58) -- cycle ;
%Shape: Ellipse [id:dp007666068173858842] 
\draw  [line width=1.5]  (24.16,169.85) .. controls (24.82,162.34) and (38.28,157.4) .. (54.22,158.81) .. controls (70.17,160.21) and (82.56,167.44) .. (81.9,174.94) .. controls (81.23,182.45) and (67.77,187.39) .. (51.83,185.99) .. controls (35.88,184.58) and (23.49,177.36) .. (24.16,169.85) -- cycle ;
%Shape: Ellipse [id:dp791196217167453] 
\draw  [line width=1.5]  (91.56,80.56) .. controls (97.84,77.57) and (108.52,86.85) .. (115.41,101.29) .. controls (122.31,115.74) and (122.8,129.88) .. (116.53,132.87) .. controls (110.25,135.87) and (99.57,126.59) .. (92.68,112.14) .. controls (85.79,97.7) and (85.29,83.56) .. (91.56,80.56) -- cycle ;
%Shape: Ellipse [id:dp1782685695004791] 
\draw  [line width=1.5]  (175.77,183.92) .. controls (170.46,178.28) and (175.58,164.8) .. (187.23,153.82) .. controls (198.87,142.84) and (212.62,138.5) .. (217.94,144.14) .. controls (223.26,149.78) and (218.13,163.26) .. (206.49,174.24) .. controls (194.84,185.22) and (181.09,189.56) .. (175.77,183.92) -- cycle ;
%Shape: Ellipse [id:dp30069966114155] 
\draw  [line width=1.5]  (154.28,196.11) .. controls (147.17,195.52) and (142.49,182.11) .. (143.81,166.16) .. controls (145.13,150.21) and (151.96,137.75) .. (159.07,138.34) .. controls (166.17,138.93) and (170.86,152.34) .. (169.53,168.29) .. controls (168.21,184.24) and (161.38,196.7) .. (154.28,196.11) -- cycle ;
%Shape: Ellipse [id:dp020668453098644846] 
\draw  [line width=1.5]  (92.84,222.01) .. controls (87.86,216.57) and (93.4,203.4) .. (105.2,192.59) .. controls (117.01,181.78) and (130.61,177.43) .. (135.59,182.87) .. controls (140.57,188.31) and (135.04,201.48) .. (123.23,212.29) .. controls (111.43,223.1) and (97.82,227.45) .. (92.84,222.01) -- cycle ;
%Shape: Circle [id:dp17237684034192202] 
\draw  [line width=1.5]  (88.68,188.3) .. controls (83.38,183.62) and (82.88,175.52) .. (87.56,170.22) .. controls (92.25,164.91) and (100.34,164.41) .. (105.65,169.1) .. controls (110.95,173.78) and (111.45,181.88) .. (106.76,187.18) .. controls (102.08,192.48) and (93.98,192.99) .. (88.68,188.3) -- cycle ;
%Shape: Circle [id:dp6965343797325375] 
\draw  [line width=1.5]  (115.16,158.32) .. controls (109.86,153.64) and (109.36,145.54) .. (114.04,140.24) .. controls (118.73,134.94) and (126.83,134.43) .. (132.13,139.12) .. controls (137.43,143.8) and (137.93,151.9) .. (133.25,157.2) .. controls (128.56,162.51) and (120.47,163.01) .. (115.16,158.32) -- cycle ;
%Shape: Circle [id:dp20461809360122496] 
\draw  [line width=1.5]  (166.25,216.8) .. controls (160.94,212.11) and (160.44,204.02) .. (165.13,198.71) .. controls (169.81,193.41) and (177.91,192.91) .. (183.21,197.6) .. controls (188.52,202.28) and (189.02,210.38) .. (184.33,215.68) .. controls (179.65,220.98) and (171.55,221.48) .. (166.25,216.8) -- cycle ;

%Straight Lines [id:da833608285989884] 
\draw [line width=2.25]    (262,164) -- (309.5,164) ;
\draw [shift={(314.5,164)}, rotate = 180] [fill={rgb, 255:red, 0; green, 0; blue, 0 }  ][line width=0.08]  [draw opacity=0] (16.07,-7.72) -- (0,0) -- (16.07,7.72) -- (10.67,0) -- cycle    ;
%Straight Lines [id:da557066392475115] 
\draw    (485.7,149.45) -- (483.65,185) ;
%Straight Lines [id:da1183347194831329] 
\draw    (455.68,190.98) -- (483.65,185) ;
%Straight Lines [id:da7043627509632235] 
\draw    (451.65,148.72) -- (455.68,190.98) ;
%Straight Lines [id:da41250976004875206] 
\draw    (483.65,185) -- (502.73,207.2) ;
%Straight Lines [id:da5364140860506336] 
\draw    (483.65,185) -- (513.51,175.36) ;
%Straight Lines [id:da28560447909745157] 
\draw    (536.2,152.7) -- (513.51,175.36) ;
%Straight Lines [id:da23546831182963213] 
\draw    (428.76,213.9) -- (455.68,190.98) ;
%Straight Lines [id:da8951382359942104] 
\draw    (425.16,178.7) -- (455.68,190.98) ;
%Straight Lines [id:da5119493162870873] 
\draw    (397.62,174.38) -- (425.16,178.7) ;
%Straight Lines [id:da883570909185129] 
\draw    (371.89,233.71) -- (399.86,227.73) ;
%Straight Lines [id:da5186786387881275] 
\draw    (364.43,170.41) -- (397.62,174.38) ;
%Shape: Circle [id:dp6261996926105942] 
\draw  [fill={rgb, 255:red, 255; green, 255; blue, 255 }  ,fill opacity=1 ] (361.12,174.16) .. controls (359.05,172.33) and (358.86,169.17) .. (360.68,167.1) .. controls (362.51,165.03) and (365.67,164.84) .. (367.74,166.66) .. controls (369.81,168.49) and (370.01,171.65) .. (368.18,173.72) .. controls (366.35,175.79) and (363.19,175.99) .. (361.12,174.16) -- cycle ;
%Straight Lines [id:da8143868779620337] 
\draw    (399.86,227.73) -- (428.76,213.9) ;
%Shape: Circle [id:dp5193883731160708] 
\draw  [fill={rgb, 255:red, 0; green, 0; blue, 0 }  ,fill opacity=1 ] (425.45,217.65) .. controls (423.38,215.82) and (423.18,212.66) .. (425.01,210.59) .. controls (426.84,208.52) and (430,208.32) .. (432.07,210.15) .. controls (434.14,211.98) and (434.33,215.14) .. (432.5,217.21) .. controls (430.68,219.28) and (427.52,219.47) .. (425.45,217.65) -- cycle ;
%Shape: Circle [id:dp35677022527252755] 
\draw  [fill={rgb, 255:red, 255; green, 255; blue, 255 }  ,fill opacity=1 ] (396.55,231.48) .. controls (394.48,229.65) and (394.28,226.49) .. (396.11,224.42) .. controls (397.94,222.35) and (401.1,222.16) .. (403.17,223.98) .. controls (405.24,225.81) and (405.43,228.97) .. (403.6,231.04) .. controls (401.78,233.11) and (398.62,233.31) .. (396.55,231.48) -- cycle ;
%Shape: Circle [id:dp15312956089120688] 
\draw  [fill={rgb, 255:red, 0; green, 0; blue, 0 }  ,fill opacity=1 ] (452.37,194.73) .. controls (450.3,192.9) and (450.1,189.74) .. (451.93,187.67) .. controls (453.76,185.6) and (456.92,185.4) .. (458.99,187.23) .. controls (461.06,189.06) and (461.25,192.22) .. (459.43,194.29) .. controls (457.6,196.36) and (454.44,196.56) .. (452.37,194.73) -- cycle ;
%Shape: Circle [id:dp35028009471714094] 
\draw  [fill={rgb, 255:red, 255; green, 255; blue, 255 }  ,fill opacity=1 ] (421.85,182.45) .. controls (419.78,180.62) and (419.59,177.46) .. (421.42,175.39) .. controls (423.24,173.32) and (426.4,173.12) .. (428.47,174.95) .. controls (430.54,176.78) and (430.74,179.94) .. (428.91,182.01) .. controls (427.08,184.08) and (423.92,184.27) .. (421.85,182.45) -- cycle ;
%Shape: Circle [id:dp5578233951530956] 
\draw  [fill={rgb, 255:red, 255; green, 255; blue, 255 }  ,fill opacity=1 ] (368.58,237.46) .. controls (366.51,235.63) and (366.31,232.47) .. (368.14,230.4) .. controls (369.97,228.33) and (373.13,228.14) .. (375.2,229.96) .. controls (377.27,231.79) and (377.46,234.95) .. (375.64,237.02) .. controls (373.81,239.09) and (370.65,239.29) .. (368.58,237.46) -- cycle ;
%Shape: Circle [id:dp7970071599226868] 
\draw  [fill={rgb, 255:red, 255; green, 255; blue, 255 }  ,fill opacity=1 ] (394.31,178.13) .. controls (392.24,176.3) and (392.04,173.14) .. (393.87,171.07) .. controls (395.7,169) and (398.86,168.81) .. (400.93,170.63) .. controls (403,172.46) and (403.2,175.62) .. (401.37,177.69) .. controls (399.54,179.76) and (396.38,179.96) .. (394.31,178.13) -- cycle ;
%Straight Lines [id:da9754695779310514] 
\draw    (425.48,92.25) -- (438.61,121.19) ;
%Straight Lines [id:da5443994522139035] 
\draw    (451.65,148.72) -- (438.61,121.19) ;
%Shape: Circle [id:dp4735786712770905] 
\draw  [fill={rgb, 255:red, 0; green, 0; blue, 0 }  ,fill opacity=1 ] (448.34,152.47) .. controls (446.27,150.64) and (446.07,147.48) .. (447.9,145.41) .. controls (449.73,143.34) and (452.89,143.15) .. (454.96,144.97) .. controls (457.03,146.8) and (457.22,149.96) .. (455.39,152.03) .. controls (453.57,154.1) and (450.41,154.3) .. (448.34,152.47) -- cycle ;
%Shape: Circle [id:dp2845860897432211] 
\draw  [fill={rgb, 255:red, 0; green, 0; blue, 0 }  ,fill opacity=1 ] (480.34,188.75) .. controls (478.27,186.92) and (478.07,183.76) .. (479.9,181.69) .. controls (481.73,179.62) and (484.89,179.42) .. (486.96,181.25) .. controls (489.03,183.08) and (489.22,186.24) .. (487.39,188.31) .. controls (485.57,190.38) and (482.41,190.57) .. (480.34,188.75) -- cycle ;
%Shape: Circle [id:dp8572480821134184] 
\draw  [fill={rgb, 255:red, 0; green, 0; blue, 0 }  ,fill opacity=1 ] (422.17,95.99) .. controls (420.1,94.17) and (419.91,91.01) .. (421.74,88.94) .. controls (423.56,86.87) and (426.72,86.67) .. (428.79,88.5) .. controls (430.86,90.33) and (431.06,93.49) .. (429.23,95.56) .. controls (427.4,97.63) and (424.24,97.82) .. (422.17,95.99) -- cycle ;
%Shape: Circle [id:dp658551788029342] 
\draw  [fill={rgb, 255:red, 0; green, 0; blue, 0 }  ,fill opacity=1 ] (435.3,124.94) .. controls (433.23,123.11) and (433.03,119.95) .. (434.86,117.88) .. controls (436.69,115.81) and (439.85,115.61) .. (441.92,117.44) .. controls (443.99,119.27) and (444.18,122.43) .. (442.36,124.5) .. controls (440.53,126.57) and (437.37,126.77) .. (435.3,124.94) -- cycle ;
%Shape: Circle [id:dp17493773265019785] 
\draw  [fill={rgb, 255:red, 255; green, 255; blue, 255 }  ,fill opacity=1 ] (499.42,210.94) .. controls (497.35,209.12) and (497.16,205.96) .. (498.98,203.89) .. controls (500.81,201.82) and (503.97,201.62) .. (506.04,203.45) .. controls (508.11,205.28) and (508.31,208.44) .. (506.48,210.51) .. controls (504.65,212.58) and (501.49,212.77) .. (499.42,210.94) -- cycle ;
%Shape: Circle [id:dp5833185390284086] 
\draw  [fill={rgb, 255:red, 255; green, 255; blue, 255 }  ,fill opacity=1 ] (510.2,179.11) .. controls (508.13,177.28) and (507.94,174.12) .. (509.77,172.05) .. controls (511.59,169.98) and (514.75,169.78) .. (516.82,171.61) .. controls (518.89,173.44) and (519.09,176.6) .. (517.26,178.67) .. controls (515.43,180.74) and (512.27,180.93) .. (510.2,179.11) -- cycle ;
%Shape: Circle [id:dp9703433304811658] 
\draw  [fill={rgb, 255:red, 0; green, 0; blue, 0 }  ,fill opacity=1 ] (532.89,156.45) .. controls (530.82,154.62) and (530.62,151.46) .. (532.45,149.39) .. controls (534.28,147.32) and (537.44,147.13) .. (539.51,148.95) .. controls (541.58,150.78) and (541.77,153.94) .. (539.95,156.01) .. controls (538.12,158.08) and (534.96,158.28) .. (532.89,156.45) -- cycle ;
%Shape: Circle [id:dp6959723411212233] 
\draw  [fill={rgb, 255:red, 0; green, 0; blue, 0 }  ,fill opacity=1 ] (482.39,153.2) .. controls (480.32,151.37) and (480.12,148.21) .. (481.95,146.14) .. controls (483.78,144.07) and (486.94,143.87) .. (489.01,145.7) .. controls (491.08,147.53) and (491.27,150.69) .. (489.44,152.76) .. controls (487.62,154.83) and (484.46,155.02) .. (482.39,153.2) -- cycle ;
%Shape: Ellipse [id:dp7680533540870695] 
\draw  [line width=1.5]  (503.77,183.92) .. controls (498.46,178.28) and (503.58,164.8) .. (515.23,153.82) .. controls (526.87,142.84) and (540.62,138.5) .. (545.94,144.14) .. controls (551.26,149.78) and (546.13,163.26) .. (534.49,174.24) .. controls (522.84,185.22) and (509.09,189.56) .. (503.77,183.92) -- cycle ;
%Shape: Polygon Curved [id:ds23555236676404678] 
\draw  [line width=1.5]  (358,237) .. controls (346,218) and (441,205) .. (438,196) .. controls (435,187) and (345,195) .. (351,168) .. controls (357,141) and (410,161) .. (417,163) .. controls (424,165) and (460,174) .. (469.66,187.99) .. controls (479.32,201.98) and (448,218) .. (426,231) .. controls (404,244) and (370,256) .. (358,237) -- cycle ;
%Shape: Polygon Curved [id:ds9202259000427551] 
\draw  [line width=1.5]  (486,136) .. controls (497,136) and (494.2,163.36) .. (495.6,180.18) .. controls (497,197) and (523,199) .. (515,216) .. controls (507,233) and (479,205) .. (475,189) .. controls (471,173) and (475,136) .. (486,136) -- cycle ;
%Shape: Ellipse [id:dp09333787102195568] 
\draw  [line width=1.5]  (419.56,80.56) .. controls (425.84,77.57) and (436.52,86.85) .. (443.41,101.29) .. controls (450.31,115.74) and (450.8,129.88) .. (444.53,132.87) .. controls (438.25,135.87) and (427.57,126.59) .. (420.68,112.14) .. controls (413.79,97.7) and (413.29,83.56) .. (419.56,80.56) -- cycle ;
%Shape: Circle [id:dp6490699103415422] 
\draw  [line width=1.5]  (443.16,158.32) .. controls (437.86,153.64) and (437.36,145.54) .. (442.04,140.24) .. controls (446.73,134.94) and (454.83,134.43) .. (460.13,139.12) .. controls (465.43,143.8) and (465.93,151.9) .. (461.25,157.2) .. controls (456.56,162.51) and (448.47,163.01) .. (443.16,158.32) -- cycle ;

\end{tikzpicture}

    \caption{Illustration of the merging step on some tree. Vertices in~$V_0$ are colored white, vertices in~$V_1$ are colored black.}
    \label{fig:merging}
\end{figure}
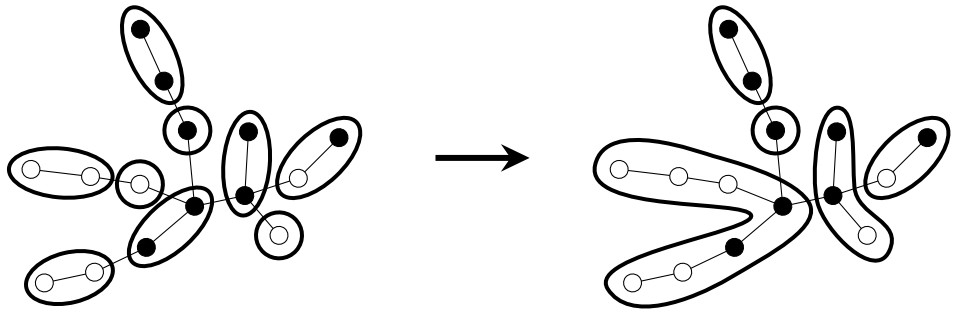

\subsubsection{Splitting Step}
\label{sec:splitting}

Let us finally explain how is performed step~3, the splitting step.
Recall that, at the end of the splitting step, our goal is to have exactly~$k$ subtrees in the partition, and we want each subtree to contain either one vertex in~$V_1$, or two vertices in~$V_1$ with at least one being a leaf.
Such a subtree will be called a \emph{valid} subtree.
Otherwise it is called an \emph{invalid} subtree.
Let us also denote by~$t$ (resp. by~$t'$) the number of subtrees in the partition at the beginning (resp. at the end) of the merging step.
We have $t' \leqslant t$. Let us denote by $m:=t-t'$ the number of merges in the merging step.
The splitting step consists itself in two phases: first, we will split invalid subtrees, in order to have only valid subtrees remaining. We will prove that, after doing this first phase, the number of subtrees in our partition is still at most~$k$. Then, the second phase will consist in splitting some valid subtrees to increase their number, to finally obtain exactly~$k$ valid subtrees at the end, which will be an \EFO allocation.

Let us detail the first phase. Let $T'$ be an invalid subtree of~$T$ in the partition at the end of the merging step. Since all subtrees intersecting~$V_1$ at the end of the partitioning step are valid, this implies that $T'$ was obtained by a merge during the merging step. Thus, the number of invalid subtrees is at most~$m$. We split each invalid subtree~$T'$ into two subtrees~$T'_1,T'_2$ such that $|T'_1 \cap V_1| = 1$ and $|T'_2 \cap V_1| = 1$.
Since we have at most~$m$ invalid subtrees after the merging step, we can bound the number~$t''$ of subtrees in the partition at the end of the first phase: we have $t'' \leqslant t' + m \leqslant t \leqslant k$.

Let us detail the second phase. At the end of the first phase, we have only valid subtrees remaining in our partition, and there are at most~$k$ of them. Let us denote by~$t_1''$ (resp. by~$t_2''$) the number of subtrees~$T'$ after the first phase, such that $|T' \cap V_1| = 1$ (resp. $|T' \cap V_1| = 2$). We have $t'' = t_1''+t_2''$. Moreover, the total number of vertices in $V_1$ is $n_1 = t_1'' + 2t_2'' = t'' + t_2''$, and by assumption we have $n_1 \geqslant k$. Thus, $t_2'' \geqslant k-t''$. Let us consider $k-t''$ distinct subtrees~$T'$ such that $|T' \cap V_1| = 2$ at the end of the first phase. We split each of these $k-t''$ subtrees into two subtrees~$T'_1,T'_2$ such that $|T'_1 \cap V_1| = 1$ and $|T'_2 \cap V_1| = 1$. All the subtrees obtained by this splits are still valids, and there are $t'' + (k-t'') = k$ of them at the end of this second phase. Thus, at the end of the second phase, we obtain an \EFO allocation. See Figure~\ref{fig:splitting} for an illustration of the two phases of the splitting step.

\begin{figure}[h!]
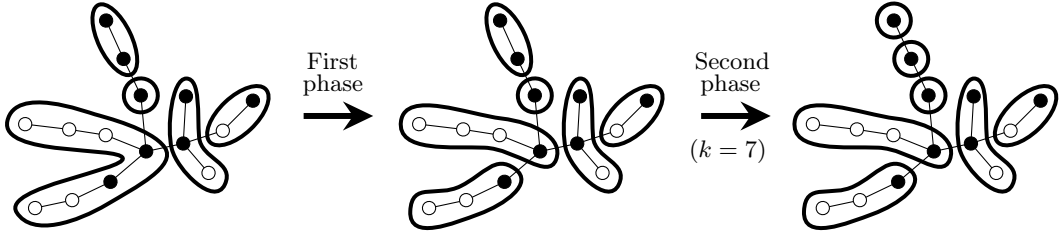

    \centering


    \caption{Illustration of the splitting step. Vertices in~$V_0$ are colored white, vertices in~$V_1$ are colored black. At the beginning, there is one invalid subtree, so we split it during the first phase to have only valid subtrees remaining. In this example, we have~$k=7$, and there are only~$6$ subtrees at the end of the first phase. Thus, we split one subtree in the second phase.}
    \label{fig:splitting}
\end{figure}

\subsection{Computing the Maximum Size of a Bad Set in Polynomial Time}\label{sec:polytime}

In this section, we note that the maximum size of a bad set can be computed in polynomial time on trees, using a dynamic algorithm that computes the maximum size of so-called $r$-locally bad sets, whose maximum size we show to be equivalent to bad sets.

\begin{restatable}{theorem}{ThmPoly}
\label{thm:poly}
The maximum size of a bad set can be computed in polynomial time on trees. In particular, the threshold number of trees for binary common additive valuation can be computed in polynomial time.
\end{restatable}

To prove Theorem~\ref{thm:poly}, we will prove that there exists an alternative definition of maximum bad sets. Before stating the lemma, let us recall the following classical statement of Hall:

\begin{lemma}[Hall's theorem]
Let $G(V_1, V_2, E) $be a bipartite graph with $|V_1| \le |V_2|$. Then $G$ has a matching saturating\footnote{A mathcing $M$ \emph{saturates} $V_1$ if every vertex of $V_1$ is an endpoint of an edge of $M$.} every vertex of $V_1$ if and only if, for every subset $S$ of $V_1$, we have that $|S| \le |N(S)|$.
\end{lemma}

\begin{lemma}
A set $X$ is a maximum bad set if and only if it is a set $X$ of maximum size such that the number of connected components of $T \setminus X$ is exactly two more than the number of vertices in $N(T \setminus X) \cap X$.
\end{lemma}
\begin{proof}
Assume first that we have a set $X$ such that the number of connected components of $T \setminus X$ is exactly two more than the number of vertices in $N(T \setminus X) \cap X$. Then, even if we remove one connected component of $T \setminus X$, the cardinality assumption ensures that it is impossible to match these connectec components to the vertices of $N(T \setminus X) \cap X$, which ensures that $X$ is a bad set.

Assume now that $X$ is a bad set of maximum size. 
Let us denote by $\mathcal{C}$ the set of connected components of $T \setminus X$. Since it is impossible to match all the components of $\mathcal{C}$ to a subset of $X$, by Hall's theorem, there exists a set of components $\mathcal{C}' \subseteq \mathcal{C}$ such that $|\mathcal{C}'| > |N(\cup_{C \in \mathcal{C}'} C \cap X|$. Let us denote by $Y$ the set $N(\cup_{C \in \mathcal{C}'} C) \cap X$. Let us choose such a set of minimum size. Note that we can assume that the difference is exactly one since otherwise, we can remove elements of $\mathcal{C}'$ to reach that bound. By Hall's theorem, the minimality of $\mathcal{C}'$ ensures that, if we remove any element $C \in \mathcal{C}'$, there is a perfect matching between the components of $\mathcal{C}'\setminus C$ and $Y$.
 Let us consider the bipartite graph between the components of $\mathcal{C}\setminus \mathcal{C}'$ and the vertices incident to them in $X \setminus Y$. If there is a perfect matching between them, then there is a matching between $\mathcal{C}$ and $X$ that matches all the components but $C$, a contradiction with the definition of bad set. So, by Hall's theorem, there exists a subset $\mathcal{C}''$ of $\mathcal{C} \setminus \mathcal{C''}$ such that its neighborhood in $X \setminus Y$ is smaller than its size. Again we can take such a set of minimum size.
 
 Note that the neighborhood in $X$ of $\mathcal{C}' \cup \mathcal{C}''$ has size  $|\mathcal{C}' \cup \mathcal{C}''|-2$. Thus the complement of $\mathcal{C}' \cup \mathcal{C}''$ provides a bad set whose size is at least the size of $X$ that satisfies the condition of the lemma, which completes the proof. 
\end{proof}

Using this lemma, we are now ready to prove Theorem~\ref{thm:poly}.

\begin{proof}[Proof of Theorem~\ref{thm:poly}.]
The proof is based on a dynamic programming algorithm that will be performed bottom-up from the leaves. We root the tree $T$ on an arbitrary root $r$. For every $u$, we denote by $T_u$ the subtree of $T$ rooted at $u$ (that is containing the subtree containing $u$ and all its descendant in the tree rooted in $r$). By abuse of notation, we will denote by $T_u$ both the subtree rooted at $u$ and the set of vertices in the subtree $T_u$

Let $u$ be a vertex and $X \subseteq V(T_u)$. For every integer $r$, we say that $X$ is \emph{$r$-locally bad for $u$}\footnote{When $u$ is clear from context, we will simply write \emph{$r$-locally bad}.} if the number of connected components in $T_u \setminus X$ is at least $r$ plus the number of vertices in $Y$ with neighbors in $T_X \setminus Y$. In other words $|N(T_u \setminus Y) \cap Y|$ is at most the number of connected components of $T_u \setminus Y$ minus $r$. We will call these vertices the \emph{border vertices of $Y$}.

For each node $u$ of the tree, the dynamic programming will compute the following values:
\begin{itemize}
\item For every $r \in \{1,2\}$, the maximum size $A_u(r)$ (if it exists) of a set $X$ such that the root $u$ is in the set $X$ and $X$ is $r$-locally bad and all the neighbors of $u$ are in $X$.
\item For every $r \in \{0,1,2\}$, the maximum size $B_u(r)$ (if it exists) of a set $X$ such that the root $u$ is in the set $X$ and $X$ is $r$-locally bad and there is a neighbor of $u$ not in $X$. 
\item For every $r \in \{1,2,3\}$, the maximum size $C_u(r)$ (if it exists) of a set $X$ such that the root $u$ is \textbf{not} in the set $X$ and $X$ is $r$-locally bad. 
\end{itemize}

The rest of the proof consists in proving that we can compute these five values using a bottom-up dynamic programming algorithm.  Actually, the truth is slightly more complicated.  One can remark that if $B_r(u)$ has some value $v_u$, then it is easy to find a bad set of size $v_u+|T \setminus T_u|$ since the complement of a $2$-locally bad set $X$ for $u$ such that $u \in X$ is a bad set. In particular, it implies that if, for $i \in \{1,2,3\}$, $C_u(i) < A_u(2)$, then computing $C_u(i)$ is useless. In the rest of the proof, we will prove that all the values are either correct or useless.

\smallskip

\noindent\textbf{Case 1.} $u$ is a leaf. \\
In that case, if we decide to add $u$ to $X$, no set of $T_u$ can be (weakly) locally bad, so the values of the first bullet are not defined. If we decide to select the empty set $X$, then we get a weakly locally bad set in the second item which has size $0$. (Note that when sets do not exist, we give their entry value $+\infty$).
\smallskip

\noindent\textbf{Case 2.} $u$ is not a leaf and $u \in X$. \\
Let us denote by $v_1,\ldots,v_\ell$ the children of $u$ in the rooted tree (possibly with $\ell=1$). Let us explain how we can compute the values using the previous ones.

Let us prove the following claims that will be used in the proof in the case where $u$ is in the set $X$:

\begin{claim}\label{clm:progdyn}
Let $u$ be a vertex and $T_u$ be the subtree rooted on $u$ and $r \le 3$. Then every set $X$ which is a maximum bad set  which is $r$-bad component for $u$ and contains $u$. Then
\begin{itemize}
    \item Assume that $X$ contains $T_u$. If we denote by $v_1,\ldots,v_\ell$ the children of $u$ in $T_u$ and by $b(v_i)$ the integer such that $X$ is $b(v_i)$-locally bad for $v_i$. Then we have that $X$ is $(\sum b(v_i))$-locally bad if $v_1,\ldots,v_r$ are in $X$ and is $(\sum b(v_i)) -1$-locally bad otherwise.
\item For every $v_i$, we have $b(v_i)$ is non-negative if $v_i \in X$ and is positive if $v_i \notin X$.
\item $X$ contains all the vertices of all the connected components of $T_u \setminus \{u\}$ but at most $r$. 
\end{itemize}
\end{claim}
\begin{proof}
The proof of the first point is very simple and simply follows from the definition.

Let $X$ be a $1,2$ or $3$-locally bad set of maximum size in $T_u$. Note that if there is a connected component $C$ of $T_u \setminus u$ rooted on $v$ is not $1+$-locally bad, then the number of border vertices in $T_v$ at most the  number of connected components in $T_v \setminus X$. If we add add all the vertices of $C$ to $X$ and denote by $X'$ the new set, the reduction of the number of connected components is at least as large as the reduction on the number of border vertices. So $X'$ is still $1$, $2$ or $3$-locally bad and is larger, a contradiction. 

So we can assume that all the connected components of $T_u \setminus u$ which do not contain all the vertices of $X$ are $k$-locally bad for some $k \ge 1$. If there are more than three such components, we consider any of them and add all the vertices of that component to $X$. The resulting set is still $k'$-locally bad for $k' \ge 3$ and the conclusion follows.
\end{proof}

This claim will allow us to complete the proof of this case. 
\medskip

\noindent \textit{Computation of $A_u(r)$ for $r \in \{1,2\}$.} \\
Claim~\ref{clm:progdyn} ensures that at most three connected components of do not fully contain vertices of $X$. 

Assume that there exists a component $C$ of $T_u \setminus u$ such that $X$ contains its root vertex $v$ and $X$ is $r$-locally bad. Then all the other vertices of $T_u$ are in $X$. To test that option when we compute $A_u(r)$ we simply have to look at the maximum size given, for every component $C$ of $T_u \setminus u$ attached on $v$, of $A_v(r) + |T_u \setminus T_v|$ and of $B_v(r) + |T_u \setminus T_v|$.

When we are looking for $A_u(1)$, since no vertex adjacent to $u$ is not in $X$, it should be true that one of the components attached on $u$ are $1$-locally bad and contain the root vertex. So the previous paragraph permits to compute that value.

When we are looking for $A_u(2)$, we have another additional case to consider. It might be possible that there exists two components $C_1,C_2$ of $T_u \setminus u$ rooted on $v_1$ and $v_2$ which both contain their root vertices and are $1$-locally bad. In that case, all the other vertices of $T_u$ are in $X$ and we can compute $A_u(2)$ using a formula similar to the one of the previous paragraph.
\medskip

\noindent \textit{Computation of the $B_u(r)$ for $r \in \{0,1,2\}$.} \\
Let $X$ be an $r$-locally bad set of maximum size of $T_u$ that contains $u$ and such that one of the neighbors of $u$ is not in $X$.
Let $v_1,v_r$ be the children of $u$ in $T_u$. Let us denote by $b(v_i)$ the integer such that $X$ is $b(v_i)$-locally bad for $v_i$. We have that $r = (\sum b(v_i)) -1$ since the vertex $u$ is a new vertex adjacent to a component of $T_u \setminus X$ which was counted in any of the $T_{v_i}$. 

In particular, the number of components attached on $u$ whose root is not in $X$ is at most $r+1$. Indeed, by Claim~\ref{clm:progdyn}, all these components are $r'$-locally bad for $r' >0$. And since we are only interested in being $r$-locally bad on $u$, the formula of the previous paragraph ensures that at most $r'+1$ children can be in $X$ and do not contain all their vertices in $X$. 

Consider now the vertices $v_i$ whose root is in $X$ and such that $X$ does not contain $T_{v_i}$. We can have at most $r$ such components that are $r'$-locally bad with $r'>0$ by the previous argument. To conclude, let us prove that there is no such components that are $0$-locally bad. Indeed, if there is such a component, we simply add in $X$ all the vertices in $T_{v_i}$. The size of $X$ increases and $X$ is still $r$-locally bad in $T_u$.

So in order to find the value of $B_r(u)$, we simply have to sum over a set of at most $4$ components, which completes the proof. 
\medskip

\noindent\textbf{Case 3.} $u$ is not a leaf and $u \notin X$. \\
Let us prove that we can compute the $C_u(r)$ for $r \in \{1,2,3\}$. 
Let $X$ be a set of maximum size in $T_u$ which is $r$-locally bad and such that $u \notin X$. Recall that, at the beginning of the proof, we have said that we will compute either the values or we will prove that the value is useless since it is smaller than $A_u(2)$. We did not use that fact so far and will use it to compute the $T_u$'s.

Let $r \in \{1,2,3\}$. Let $X$ be a set of maximum size $C_u(r)$ such that the root $u$ is not in the set $X$ and $X$ is $r$-locally bad.
Let $v_1,\ldots,v_r$ be the children of $u$. We will distinguish three types of components attached on $u$. The ones, of Type I, such that $v_i \in X$ and $v_i$ is adjacent to a vertex in $T_v \setminus X$. The ones, said to be of Type II, such that $v_i \notin X$. And finally the ones of Type III such that which corresponds to the $v_i$ such that $v_i \in X$ and all the children of $v_i$ in $T_{v_i}$ also are in $X$.

Assume first that there are at least $3$ components of Type III. Let us consider the set $X'=X \cup \{ u \}$. We have that the number of components of $T_u \setminus X'$ reduces by at most one compared to $T_u \setminus X$. And the number of vertices in $X'$ with a neighbor in $T \setminus X'$ reduces by at least two compared with $T \setminus X$. Indeed, the roots of the three components of Type III satisfied that $u$ was their only neighbor not in $X$. Thus these vertices are not anymore in the neighborhood of $T_u \setminus X'$. Since $X$ was $r$-locally bad, we have that $X'$, whose size is larger, is $2$-locally bad and has its root in $X'$. So, in that case, the value is useless.

So we can assume that there are at most $3$ components of Type III. Assume now at least $2$ components of Type II. Then if we set $X'= X \cup \{ u \}$, the border increases by at most one and the number of connected components of $T \setminus X'$ increases by at least $2$.  So $X'$ is larger than $X$, is $2$-locally bad and has its root in $X'$. So, again in that case, the value is useless. 

So we can finally assume that all the components but at most $5$ are of type I. The values corresponding to components of Type I are computed by the values $B_{v_i}(r')$. One can note that, if $r'$ is positive at least $2$ components, then we can again modify the set $X$ into a larger set $X'$ which is $2$-locally bad and has its root in $X'$. 
So at most one component satisfies that $r'>1$. So there are at most five components that are special, which we can guess, and, for all the others, we should have computed the value $B_{v_i}(0)$. This gives a polynomial number of cases to test which permits to compute the value $C_u(r)$ in polynomial time and completes the proof.
\end{proof}

\section{Spectrum Threshold on Trees with Many Leaves}\label{sec:manyleaves}

In this section, we discuss the maximum size of bad sets depending on the number of leaves of the tree. We showed in Proposition~\ref{prop:bad_set_at_least_n/2} that the size of the largest bad set plus two is almost always larger than half of the vertices.
However, when one tries to design bad examples, that is graphs which do not admit large bad sets, it becomes much harder when the number of leaves in the tree increases. The goal of this section is to formalize this intuition and show that the maximum size of bad sets increases when the number of leaves increases. In other words, the more the tree is branching, the easier it is to construct bad sets.

\begin{theorem}\label{thm:badset_leaves}
    There exists a constant $C$ such that, for every tree $T$ with $\ell$ leaves, the threshold of the spectrum is at least $(1-C \cdot \frac 1\ell)n$ even for common binary additive valuation functions. 
\end{theorem}

To prove it, we will prove that the following holds which is equivalent to Theorem~\ref{thm:badset_leaves} (up to an additive constant) by Theorem~\ref{thm:binary}.

\begin{theorem}\label{thm:badset_leaves2}
    There exists a constant $C$ such that, for every tree $T$ with $k$ leaves, there exists a bad set with at least $(1-\frac Ck)n$ vertices. 
\end{theorem}

Note that we did not try to optimize the size of the constant $C$ and instead focus on the simplicity of the proof. The idea of the proof consists in first exhibiting three structures whose existence ensures the existence of large bad sets and then proving every tree with a large enough number of leaves contains one of these structures.

\paragraph*{Relevant structures.}

\begin{figure}[h!]
    \centering
    \includegraphics[scale=0.65]{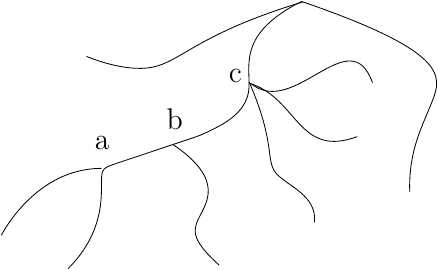}
    \caption{The vertex $b$ is $1$-terminal, $a$ is exactly $2$-terminal (and then $1$ and $2$-terminal) and $c$ is exactly $3$-terminal.}
    \label{fig:terminal}
\end{figure}

Let $T$ be a tree. We say that a vertex $v$ is a \emph{$\ell$-terminal vertex} if $T \setminus v$ has at least $\ell$ connected components $P_1,\ldots,P_\ell$ such that, for every $i$, $P_i$ is a path and $v$ is connected to an endpoint of $P_i$. Differently, it means that at least $\ell$  leaves of the tree are attached to $v$ by a path. The paths $P_1,\ldots,P_\ell$ are called the \emph{probes} of $v$. The vertex $v$ is an \emph{exactly $\ell$-terminal vertex} if exactly $\ell$ such connected components exist.
The \emph{$\ell$-weight}\footnote{When $\ell$ is clear from context, we will simply write \emph{weight}.} of an $\ell$-terminal vertex is the sum of the sizes of the paths $P_1,\ldots,P_\ell$. When more than $\ell$ paths are attached to $v$, the weight is the sum of the sizes of $\ell$ paths of minimum length leading to leaves attached to $v$.

The following observation, which simply follows from the fact that probes only contain degree at most $2$ vertices and then do not contain branching nodes, will be extensively used in an implicit way:

\begin{remark}
Let $\ell \ge 1$ and $v_1,v_2$ be two $\ell$-terminal vertices. The probes of $v_1$ and $v_2$ are pairwise vertex disjoint.
\end{remark}

Given two $\ell$-terminal vertices $v_1,v_2$, the \emph{$\ell$-weight} of the pair $(v_1,v_2)$ is the sum of the $\ell$-weights of $v_1$ and $v_2$. Using the tools we introduced in previous sections, it is easy to prove that if we have a $3$-terminal vertex or a pair of $2$-terminal vertices of small weight then the theorem holds. For instance, if we have a $3$-terminal vertex, the set consisting of all the vertices of $T$ but the vertices of $\cup_{i=1}^3 P_i$ forms a bad set. We can improve that bound by proving that the following holds:

\begin{lemma}\label{lem:small3branching}
    If $T$ contains a $3$-terminal vertex of weight $t$ then $\tau(G) \ge n - t/2 - 3$.
\end{lemma}
\begin{proof}
Let us construct a bad set $A$ as follows. Let $x$ be a $3$-terminal vertex of weight $t$ and $P_1,P_2,P_3$ be three probes of $x$ such that the sum of the sizes of the $P_i$ is equal to $t$.
All the vertices of $V \setminus (P_1 \cup P_2 \cup P_3)$ are in $A$. Moreover, for each path $P_i$ -whose vertices are denoted by $v_1^i,\ldots,v_r^i$ and such that $v_1^i$ is the neighbor of $x$, the vertex $v_j^i$ if $j$ is in $A$ if $j$ is even and $j \neq r$. The following holds: 
\begin{enumerate}
    \item $A$ is a bad set and,
    \item $|A| \ge n-t/2-6$.
\end{enumerate}

\noindent \textit{Proof of 1.} Let $C$ be the set of connected components of vertices not in $A$. In order to prove that $A$ is a bad set, it suffices to prove that $|C| \ge |N(C)\cap A| +2$. By definition of $A$, for every $i$, the path $P_i$ starts and end with a vertex not in $A$ and $A$ does not contain two consecutive vertices of $P_i$. Hence the number of connected components of vertices not in $A$ is larger than the number of vertices in $A$ in $P_i$. Since the statement holds for every $i \in \{1,2,3\}$ and that $N(P_i) \setminus P_i = \{ v \}$, we have that $|C| \ge |N(C)\cap A| +2$, which completes the proof that the set of $1$ vertices is a bad set.
\smallskip

\noindent
\textit{Proof of 2.} The number of vertices in the probes $P_1,P_2,P_3$ is at most $t$. Moreover, by definition of the valuation function, the number of vertices in $A$ in these paths is at least $\sum_{i=1}^3 (|P_i|/2 -1)$. So the size of $A$ is at least $n - \sum_i |P_i|/2 - 3 \ge n - t/2- 3$, which completes the proof.
\end{proof}

\begin{lemma}\label{lem:small2branching}
If $G$ contains a pair $v_1,v_2$ of two $2$-terminal vertices of weight $t$ then $\tau(G) \ge n - t/2 - 4$.
\end{lemma}
\begin{proof}
The proof is similar to the one of Lemma~\ref{lem:small3branching}. We define similarly the set $A$. The difference of $1$ comes from the fact that we have $4$ paths instead of $3$.
\end{proof}

However, one can remark that, even if we have a large number of nodes, we cannot ensure that the graph contains a $3$-terminal vertex of small weight or a pair of two $2$-terminal vertices of small weight. A typical example is as follows: consider a path consisting of $s$ vertices $x_1,\ldots,x_s$. Now we attach paths on these vertices as follows. We first attach a path on $x_1$ and a path on $x_s$ of length $n/2$. Then, for every $i \in \{1,\ldots,s\}$, we attach on $x_i$ a leaf. One can remark that there is no $3$-terminal vertex and that the only $2$-terminal vertices are $x_1$ and $x_s$ which have weight (approximately) $n/2$. But, in the middle of the structure, we have a long path of branching nodes of degree $3$ which are $1$-terminal. In this graph, we can find a third structure that also ensures that there exist large bad sets.

\begin{figure}
    \centering
    \includegraphics[scale=.73]{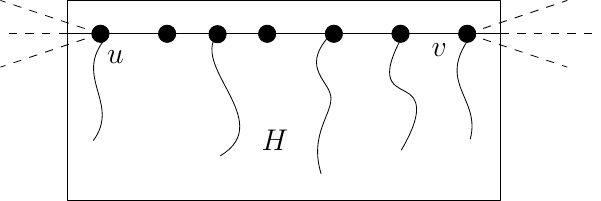}
    \caption{The subgraph $H$ (which is in the squared box) is a bad $5$-caterpillar. There are $5$ branching nodes between $u$ and $v$ (counting them) and there is exactly pending path on each of them, except on $u$ and $v$ on which more paths might exist.}
    \label{fig:badcaterpillar}
\end{figure}

Let $T$ be a tree. A connected subset $H$ of $T$ is a \emph{$\ell$-bad caterpillar} if $H$ contains two branching nodes $u,v$ such that:
\begin{itemize}
\item $u,v$ are the only vertices of $H$ connected to a vertex of $V(G) \setminus V(H)$ and,
\item All the branching nodes of $H$ are on the path between $u$ and $v$ and there are exactly $\ell$ of them (including $u,v$) and,
\item All the branching nodes but $u,v$ are exactly $1$-terminal  and $u,v$ are $1$-terminal and,
\item At least one of the paths of $T \setminus u$ and $T \setminus v$ are in $H$. 
\end{itemize}
(See Figure~\ref{fig:badcaterpillar} for an illustration). Said differently, a $\ell$-bad caterpillar $H$ contains two vertices $u,v$ such that the connected component of $G\setminus \{u,v\}$ adjacent to $u$ and $v$ is a path with $\ell$ branching nodes on which exactly one path is attached and moreover $H$ contains a path attached on~$u$ and on~$v$. 
The \emph{weight} of a $\ell$-bad caterpillar is its number of vertices.

We will prove that if $T$ contains a $5$-bad caterpillar of small weight then there is a large bad set. If we do not want to optimize the constants, one can build a bad set as follows. Let $H$ be a $5$-bad caterpillar and $u,v$ be the two branching nodes with neighbors not $H$. Let us denote by $P$ the path from $u$ to $v$ and $x_2,x_3,x_4$ the three branching nodes on the path $P$ distinct from $u,v$. We claim that the set containing all the vertices of $V(G) \setminus V(H)$ plus $u$, $x_3$ and $v$ form a bad set $X$. Indeed, there are five components in $V \setminus X$: the path attached on $u$, the path attached on $v$, the path attached on $x_3$, the component of $x_2$ and the component of $x_4$. And the border of these five components only contain three vertices (namely $u,v,w_2$). 
As in Lemma~\ref{lem:small3branching} and~\ref{lem:small2branching}, let us prove that we can increase the size of the bad set. 

\begin{lemma}\label{lem:small_badcaterpillar}
If $G$ contains a $5$-bad caterpillar of weight $t$ then $\tau \ge n-t/2-10$. 
\end{lemma}

%We argue as in Lemmas~\ref{lem:small3branching}--\ref{lem:small2branching}: take $A\supseteq V(T)\setminus V(H)$ and choose an alternating  subset of vertices on each pending path and on the spine segments of $H$ so that the same components--boundary counting certifies that $A$ is bad. This keeps all but at most $t/2+O(1)$ vertices of $H$, hence $|A|\ge n-t/2-O(1)$ and in particular $\tau(T)\ge n-t/2-10$; full details are deferred to Appendix \ref{sec: proof-smallbad}.

\begin{proof}
Let $H$ be a $5$-bad caterpillar. Let $x_1:=u,x_2,\ldots,x_5:=v$ be its $5$ branching vertices of $H$.  For every $x_i$, let $P_i$ be the unique connected component of $H \setminus x_i$ that is a path. Let us denote by $v_1^i,\ldots,v_{r_i}^i$ the vertices of $P_i$ starting from the vertex adjacent to $x_i$.
Let $W_i$ be the path (possibly reduced to a single edge) from $x_i$ to $x_{i+1}$ and let $w_1^i,\ldots,w_{q_i}^i$ be the vertices of $W_i$ where $w_1^i$ is adjacent to $x_i$ and $w_{q_i}^i$ is connected to $x_{i+1}$.

 Now we construct a bad set $A$ as follows. All the vertices which are not in $H$ are added in $A$. Now, on the paths $P_1,P_3,P_5$ (resp. $W_1,W_3$), we add  $v_j^i$ (resp. $w_j^i$) in $A$ if and only if $j$ is even and $j$ is not the last vertex of $P_i$ (resp. $W_i$).
Moreover, we add $x_1,x_3,x_5$ in $A$ (and do not include $x_2,x_4$ in $A$).
Finally, on the paths $P_2,P_4,W_2,W_4$  we do not include in $A$ all the vertices whose parity is the parity of the last vertex of the path and add the others in $A$. We claim that:

\begin{enumerate}
    \item $A$ is a bad set and,
    \item $|A| \ge n-t/2-10$.
\end{enumerate}

\noindent Proof of 1.
Let us denote by $C$ the connected components of vertices not in $A$. As in the proof of Lemma~\ref{lem:small3branching}, one can remark that, for each path in $P_1,P_3,P_5$, the number of connected components of $C$ is in these paths larger than the number of vertices of $A$ in these paths. Similarly, one can note that in the paths  $P_2,P_4,W_2,W_4$, the number of connected components of $C$ is at least as large as the number of vertices in $A$. Moreover, if there is a tie, the first vertex of the path is in $A$. So in total, $|C|$ is larger than the number of vertices of $A$ in $H$ plus $2$. Since the vertices of $H$ with neighbors in $T \setminus H$ are in $A$, $|N(C) \cap A | \ge |C|+2$, and then $C$ is a bad set.
\smallskip

\noindent 
Proof of 2.
For every path $P$ which is either a path $P_i$ or $Q_i$, the number of vertices added in $A$ is at least $\lceil |P|/2 \rceil -1$. Moreover, three of the branching nodes are in $A$ and two are not. So in total, the size of $A$ is at least $n-2-(t-5)/2-8 \ge n-t/2-10$.
\end{proof}

\paragraph*{Finding a relevant structure}
To complete the proof of Theorem~\ref{thm:badset_leaves2}, we simply have to show that every tree $T$ with enough leaves contains either a $3$-terminal node of weight $O(n/k)$ or, a pair of $2$-terminal nodes of weight $O(n/k)$ or a bad $5$-caterpillar of weight $O(n/k)$. Note that, free to get a worse constant, we can always assume that $k$ is large enough. 
We will distinguish two cases depending on the number of branching nodes compared to the number of leaves.
\medskip

\noindent \textbf{Case 1. Small number of branching nodes.} Assume that the number of branching nodes is at most $k/2$. For every leaf $\ell_i$, let $P_i$ be the path linking $\ell_i$ to its closest branching node. Note that, for every $i \ne j$, we have that $P_i$ and $P_j$ are vertex disjoint.

Let us consider the bipartite graph with vertex set $B \cup L$ where $L$ is the set of leaves and $B$ is the set of branching nodes of $T$. We create an edge from $\ell_i$ to $b_j$ of weight $w_i$ if $P_i$ ends on $b_j$ and has $w_i$ vertices in total. 
The sum of the weights of the edges in this bipartite graph is at most $n$. Note that at most $k/4$ edges can have weight at least $4n/k$ since otherwise the sum of the weigths would be more than~$n$. Let us remove these edges. It now remains $3k/4$ edges of weight at most $4n/k$. Since $B$ has size at most $k/2$, the average degree of the remaining vertices of $B$ is at least $\frac 32$. Thus, since $k$ is large enough, either there is one vertex of degree $3$ in $B$ (which gives a $3$-terminal node of weight at most $12n/k$) or two degree $2$ vertices (which provides a pair of $2$-terminal nodes of weight at most $16n/k$). 
Lemmas~\ref{lem:small3branching} and~\ref{lem:small2branching}, we have that $\tau(G) \ge n(1-8/k)-c$ where $c$ is a constant, which completes the proof in this case.
\medskip

\noindent \textbf{Case 2. Large number of branching nodes.} So, from now on, we will assume that the number of branching nodes is at least $k/2$. Let us denote by $X$ the set of $2$-terminal branching nodes of $T$.
For every $x \in X$, let us denote by $P_x$ the set of probes of $x$. And let $\mathcal{P}$ be the union of the $P_x$ for $x$ in $X$.

Let $\alpha:=\frac{1}{15}$.
If $X$ contains an $\alpha$-fraction of the branching nodes, then, since the paths of $\mathcal{P}$ are vertex disjoint, there exist two vertices $x,y$ in $X$ such that the union of probes $P_x$ and $P_y$ contains at most $2 \alpha n/ k$ vertices. Indeed, since $|X| \geqslant k/2\alpha$, the average size of union of probes is at most $2\alpha n/k$. So there is pair of $2$-terminal branching nodes of weight $4\alpha n /k$ and then a bad set of size $n(1-2 \alpha / k)-4$ by Lemma~\ref{lem:small2branching}.

So we can assume that $X$ contains less than an $\alpha$-fraction of the branching nodes. Let us denote by $B$ the set of branching nodes of $T$. Let us construct the auxiliary tree $T'$ as follows: the vertex set of $T'$ is $B$. Two vertices $x,x'$ of $T'$ are adjacent if and only if there exists a path in $T$ whose internal vertices are of degree $2$, between $x$ and $x'$. 

By assumption $|X| \le \alpha |T'|$ and $X$ is a vertex subset of $T'$. Note moreover that all the leaves of $T'$ are in $X$. Indeed, a leaf of $T'$ is a vertex $u$ such that $u$ has degree at least $3$ in $T$ and one in $T'$. So at least two paths leading to leaves are attached on $u$ in $T$ and then $u$ is in $X$\footnote{But note that some vertices of $X$ might not be leaves of $T$.}.
Thus, the number of leaves of $T'$ is at most $\alpha |T'|$. So the number of branching nodes $N$ of $T'$ also has size at most $\alpha |T'|$. Now let $R=N \cup X$. Note that $|R| \le 2\alpha |T'|$. Note that all the connected components of  $T' \setminus N$, denoted by $\mathcal{C}$, are paths of degree $2$ vertices in $T'$. 

Since all the vertices of $T'$ are branching nodes, all the vertices of $T' \setminus R$ have degree at least $3$ in $T$. Let $x$ be such a vertex. Since the degree of $x$ in $T$ is larger than in $T'$, it means that one connected component of $T \setminus x$ is a path. In other words, $x$ is a $1$-terminal vertex. Moreover, $x$ is not a $2$-terminal vertex since otherwise it would be in $X$. Thus $x$ has degree $3$ in $T$ and is an exactly-$1$-terminal vertex.

Let $P$ be a connected component of $\mathcal{C}$, which is a path. Let us denote by $v_1,\ldots,v_r$ the vertices of $P$. For every $i \ge 0$  such that $i \le r/5$, the \emph{$i$-slice $S_i$ of $P$} is subpath $v_{5i+1},\ldots,v_{5i+5}$. The vertices $v_{5i+1}$ and $v_{5i+5}$ are called the \emph{extremities} of the slice. Note that all the vertices of each connected component of $\mathcal{C}$ belong to a slice but at most $4$ (which form a leftover we will ignore).
Let us denote by $\mathcal{S}$ the union of the slices. Since $R$ contains at most $2\alpha |T'|$ vertices with at most $\alpha |T'|$ branching nodes, the number of components in $\mathcal{C}$ is at most $2\alpha |T'|$. Since all the vertices of a connected component of $\mathcal{C}$ but at most $4$ per component are in slices, at most $10 \alpha |T'|$ vertices are not in a slice of $\mathcal{S}$. Since all the slices contain $5$ vertices of $T'$, we have that the following holds:

\begin{remark}
    The number of slices in $\mathcal{S}$ is at least $((1-10\alpha)/5) |T'|$.
\end{remark}

Let $S \in \mathcal{S}$. The \emph{extension $T_S$} of $S$ is the subtree of $T$ restricted to the paths between the extremities of the slice plus the (unique) probe attached on each of the five branching vertices in $S$. Let us prove some very simple fact on extensions of the slices that easily follow from their definition:

\begin{lemma}
    For every $S,S'$ in $\mathcal{S}$ such that $S \ne S'$ then $T_S$ and $T_{S'}$ are vertex disjoint.
\end{lemma}
\begin{proof}
By definition of $\mathcal{S}$, all the slices contain five consecutive vertices in a path of $T'$. Moreover, all the slices contain pairwise disjoint vertices of $T'$. Thus the path between the extremities of a slice $S$ only contains $5$ branching vertices of $T$ which are the nodes of $T'$ in the slice. Since the probe does not contain any vertices of degree at least $3$, the addition of the probe cannot create an intersection between the extensions, which completes the proof.
\end{proof}

\begin{lemma}
    For every $S$, $T_S$ is a bad $5$-caterpillar.
\end{lemma} 
\begin{proof}
    Let $S$ be a slice. By definition of slices, all the vertices in $S$ are not in $R$. So in particular, all the vertices in $S$ have degree three in $T$ and exactly one of the components attached to them is a path. Indeed, if none of them were paths, they would be branching nodes of $T$' and then be $N$ which is included in $R$. If at least two of them were paths, then they would be in $X$ and then in $R$.

    Moreover, by definition of slices, they are consecutive vertices  $v_1,\ldots,v_5$ of a path in $T'$. Since $T$ is a tree and that edges in $T'$ are paths of degree $2$ vertices in $T$, if the path from $v_i$ to $v_j$  with $i < j \le 5$ contains a branching node of $T$, it is a node $v_t$ with $i<t<j$. 

    Finally, if we add the unique path attached to each of the $v_i$'s that ends on a leaf that indeed provides a bad $5$-caterpillar.
\end{proof}

So there exist $((1-10\alpha)/5) |T'|\geqslant ((1-10\alpha)/10) k$ pairwise vertex disjoint bad $5$-caterpillars in $T$ (the last inequality comes from the fact that the number of branching nodes is at least half of the number of leaves). If we set $\alpha=1/15$, it ensures that there exist $k/30$ pairwise vertex disjoint bad $5$-caterpillars in $T$. So one of them has weight at most $30n/k$. And Lemma~\ref{lem:small_badcaterpillar} ensures that $\tau(G) \ge (1-15/k)n-12$, which completes the proof.

\section{Algorithmic Aspects}
\label{sec:complexity}

A \emph{star} is a tree with one branching node, called the \emph{center}, on which a set of leaves is attached. Given a graph $G$ and a vertex $v$ of $G$, the \emph{$k$-blow-up of $G$ on $v$} is the graph obtained from $G$ by replacing $v$ by a clique of size $k$ such that the new vertices are adjacent to $u$ if and only if $uv$ is an edge of $G$. A \emph{blow-up} of $G$ on $v$ is a $k$-blow-up for some integer $k$. 
A graph has \emph{bandwidth $k$} if there exists an ordering $v_1,\ldots,v_n$ of the vertices of the graph such that, for every $i$, all the vertices adjacent to $v_i$ are in the set $\{v_{i-k},\ldots,v_{i+k} \}$.

An instance of the \textsc{Outer Connected Fair Division of Graphs} problem is given by a graph $G$, a set of agents $N$ and their respective valuation function $v_i$ for each $i \in N$, and consists in deciding whether there exists an \EFO allocation of the vertices of the graph.

\begin{theorem} 
\label{thm:np_result} 
    The \textsc{Outer Connected Fair Division of Graphs} problem is NP-complete even restricted to:
    \begin{itemize}
        \item The blow-up of subdivided stars (on their center) for common binary additive valuation functions.
        %\item The 2-blow-up of subdivided stars (on their center) for common additive valuation functions.
        \item Graphs of bounded bandwidth for common additive valuation functions.
    \end{itemize}
\end{theorem}
\begin{proof}
    Note first that the problem is in NP, since, to verify that an allocation is \EFO\!, it suffices to check for each pair of agents $i, j$, that up to one outer item, agent $i$ does not envy agent $j$, i.e. that there is some item allocated to agent $j$ which does not disconnect their bundle but sufficiently reduces its value. This can be done in polynomial time by exhaustive search.
    
    We reduce from the \textsc{Partition} problem, which, given a multiset $S$ of positive integers, consists in deciding whether there exists a partition $(S_1, S_2)$ such that the sum of numbers in $S_1$ equals the sum of numbers in $S_2$. 
    We first check, in linear time, whether there is any single number that is larger than the sum of the rest of the numbers. If we find such a number, it is a NO-instance. 
    Assuming that no such number exists, we construct an instance of the \textsc{Outer Connected Fair Division of Graphs} problem as follows. For each number $v$ in $S$, we create a chain of $v$ vertices of value 1. We add two \emph{center vertices} of value 0 that are adjacent to each other and to one endpoint of each chain. Finally, we add a pendant vertex of value 0 to each leaf (i.e. to the other endpoint of each chain). The resulting 2-blow-up of a subdivided star is illustrated in Figure \ref{fig:partition_star_blown_up_edge}, in the case of $S = \{3,1,5,4\}$.
    \begin{figure}[h!]
        \centering

\tikzset{every picture/.style={line width=0.75pt}} %set default line width to 0.75pt        

\begin{tikzpicture}[x=0.75pt,y=0.75pt,yscale=-0.7,xscale=0.7]

%Shape: Circle [id:dp9233277090605028] 
\draw [red]  (179,16.75) .. controls (179,13.02) and (182.02,10) .. (185.75,10) .. controls (189.48,10) and (192.5,13.02) .. (192.5,16.75) .. controls (192.5,20.48) and (189.48,23.5) .. (185.75,23.5) .. controls (182.02,23.5) and (179,20.48) .. (179,16.75) -- cycle ;
%Shape: Circle [id:dp4515144946097994] 
\draw [red]  (249.5,16.75) .. controls (249.5,13.02) and (252.52,10) .. (256.25,10) .. controls (259.98,10) and (263,13.02) .. (263,16.75) .. controls (263,20.48) and (259.98,23.5) .. (256.25,23.5) .. controls (252.52,23.5) and (249.5,20.48) .. (249.5,16.75) -- cycle ;
%Straight Lines [id:da5093828610648342] 
\draw [red]   (192.5,16.75) -- (249.5,16.75) ;
%Shape: Circle [id:dp8096188417365748] 
\draw  [fill={rgb, 255:red, 0; green, 0; blue, 0 }  ,fill opacity=1 ] (149,90.75) .. controls (149,87.02) and (152.02,84) .. (155.75,84) .. controls (159.48,84) and (162.5,87.02) .. (162.5,90.75) .. controls (162.5,94.48) and (159.48,97.5) .. (155.75,97.5) .. controls (152.02,97.5) and (149,94.48) .. (149,90.75) -- cycle ;
%Shape: Circle [id:dp15448185696504924] 
\draw  [fill={rgb, 255:red, 0; green, 0; blue, 0 }  ,fill opacity=1 ] (192.5,90.75) .. controls (192.5,87.02) and (195.52,84) .. (199.25,84) .. controls (202.98,84) and (206,87.02) .. (206,90.75) .. controls (206,94.48) and (202.98,97.5) .. (199.25,97.5) .. controls (195.52,97.5) and (192.5,94.48) .. (192.5,90.75) -- cycle ;
%Shape: Circle [id:dp15158042997386978] 
\draw  [fill={rgb, 255:red, 0; green, 0; blue, 0 }  ,fill opacity=1 ] (235,90.75) .. controls (235,87.02) and (238.02,84) .. (241.75,84) .. controls (245.48,84) and (248.5,87.02) .. (248.5,90.75) .. controls (248.5,94.48) and (245.48,97.5) .. (241.75,97.5) .. controls (238.02,97.5) and (235,94.48) .. (235,90.75) -- cycle ;
%Shape: Circle [id:dp0667078527862004] 
\draw  [fill={rgb, 255:red, 0; green, 0; blue, 0 }  ,fill opacity=1 ] (276.5,90.75) .. controls (276.5,87.02) and (279.52,84) .. (283.25,84) .. controls (286.98,84) and (290,87.02) .. (290,90.75) .. controls (290,94.48) and (286.98,97.5) .. (283.25,97.5) .. controls (279.52,97.5) and (276.5,94.48) .. (276.5,90.75) -- cycle ;
%Shape: Circle [id:dp7722521483672288] 
\draw  [fill={rgb, 255:red, 0; green, 0; blue, 0 }  ,fill opacity=1 ] (149,120.25) .. controls (149,116.52) and (152.02,113.5) .. (155.75,113.5) .. controls (159.48,113.5) and (162.5,116.52) .. (162.5,120.25) .. controls (162.5,123.98) and (159.48,127) .. (155.75,127) .. controls (152.02,127) and (149,123.98) .. (149,120.25) -- cycle ;
%Shape: Circle [id:dp661688365050443] 
\draw  [fill={rgb, 255:red, 0; green, 0; blue, 0 }  ,fill opacity=1 ] (235,120.25) .. controls (235,116.52) and (238.02,113.5) .. (241.75,113.5) .. controls (245.48,113.5) and (248.5,116.52) .. (248.5,120.25) .. controls (248.5,123.98) and (245.48,127) .. (241.75,127) .. controls (238.02,127) and (235,123.98) .. (235,120.25) -- cycle ;
%Shape: Circle [id:dp6056726859730958] 
\draw  [fill={rgb, 255:red, 0; green, 0; blue, 0 }  ,fill opacity=1 ] (276.5,120.25) .. controls (276.5,116.52) and (279.52,113.5) .. (283.25,113.5) .. controls (286.98,113.5) and (290,116.52) .. (290,120.25) .. controls (290,123.98) and (286.98,127) .. (283.25,127) .. controls (279.52,127) and (276.5,123.98) .. (276.5,120.25) -- cycle ;
%Shape: Circle [id:dp670281197882122] 
\draw  [fill={rgb, 255:red, 0; green, 0; blue, 0 }  ,fill opacity=1 ] (235,148.75) .. controls (235,145.02) and (238.02,142) .. (241.75,142) .. controls (245.48,142) and (248.5,145.02) .. (248.5,148.75) .. controls (248.5,152.48) and (245.48,155.5) .. (241.75,155.5) .. controls (238.02,155.5) and (235,152.48) .. (235,148.75) -- cycle ;
%Shape: Circle [id:dp3854341124629628] 
\draw  [fill={rgb, 255:red, 0; green, 0; blue, 0 }  ,fill opacity=1 ] (235,177.25) .. controls (235,173.52) and (238.02,170.5) .. (241.75,170.5) .. controls (245.48,170.5) and (248.5,173.52) .. (248.5,177.25) .. controls (248.5,180.98) and (245.48,184) .. (241.75,184) .. controls (238.02,184) and (235,180.98) .. (235,177.25) -- cycle ;
%Shape: Circle [id:dp4984720825840313] 
\draw  [fill={rgb, 255:red, 0; green, 0; blue, 0 }  ,fill opacity=1 ] (149,149.25) .. controls (149,145.52) and (152.02,142.5) .. (155.75,142.5) .. controls (159.48,142.5) and (162.5,145.52) .. (162.5,149.25) .. controls (162.5,152.98) and (159.48,156) .. (155.75,156) .. controls (152.02,156) and (149,152.98) .. (149,149.25) -- cycle ;
%Shape: Circle [id:dp4367936165809805] 
\draw  [fill={rgb, 255:red, 0; green, 0; blue, 0 }  ,fill opacity=1 ] (235,205.75) .. controls (235,202.02) and (238.02,199) .. (241.75,199) .. controls (245.48,199) and (248.5,202.02) .. (248.5,205.75) .. controls (248.5,209.48) and (245.48,212.5) .. (241.75,212.5) .. controls (238.02,212.5) and (235,209.48) .. (235,205.75) -- cycle ;
%Shape: Circle [id:dp5612381149010585] 
\draw  [fill={rgb, 255:red, 0; green, 0; blue, 0 }  ,fill opacity=1 ] (276.5,177.5) .. controls (276.5,173.77) and (279.52,170.75) .. (283.25,170.75) .. controls (286.98,170.75) and (290,173.77) .. (290,177.5) .. controls (290,181.23) and (286.98,184.25) .. (283.25,184.25) .. controls (279.52,184.25) and (276.5,181.23) .. (276.5,177.5) -- cycle ;
%Shape: Circle [id:dp43649185095321896] 
\draw  [fill={rgb, 255:red, 0; green, 0; blue, 0 }  ,fill opacity=1 ] (276.5,149) .. controls (276.5,145.27) and (279.52,142.25) .. (283.25,142.25) .. controls (286.98,142.25) and (290,145.27) .. (290,149) .. controls (290,152.73) and (286.98,155.75) .. (283.25,155.75) .. controls (279.52,155.75) and (276.5,152.73) .. (276.5,149) -- cycle ;
%Straight Lines [id:da7724006302380608] 
\draw    (155.75,97.5) -- (155.75,171) ;
%Straight Lines [id:da842233353046402] 
\draw    (241.75,97.5) -- (241.75,227.75) ;
%Shape: Circle [id:dp7773970674062611] 
\draw   (149,177.75) .. controls (149,174.02) and (152.02,171) .. (155.75,171) .. controls (159.48,171) and (162.5,174.02) .. (162.5,177.75) .. controls (162.5,181.48) and (159.48,184.5) .. (155.75,184.5) .. controls (152.02,184.5) and (149,181.48) .. (149,177.75) -- cycle ;
%Shape: Circle [id:dp3908086921194093] 
\draw   (192.5,119.25) .. controls (192.5,115.52) and (195.52,112.5) .. (199.25,112.5) .. controls (202.98,112.5) and (206,115.52) .. (206,119.25) .. controls (206,122.98) and (202.98,126) .. (199.25,126) .. controls (195.52,126) and (192.5,122.98) .. (192.5,119.25) -- cycle ;
%Shape: Circle [id:dp23747171801594869] 
\draw   (235,234.5) .. controls (235,230.77) and (238.02,227.75) .. (241.75,227.75) .. controls (245.48,227.75) and (248.5,230.77) .. (248.5,234.5) .. controls (248.5,238.23) and (245.48,241.25) .. (241.75,241.25) .. controls (238.02,241.25) and (235,238.23) .. (235,234.5) -- cycle ;
%Shape: Circle [id:dp7773618051815165] 
\draw   (276.5,206) .. controls (276.5,202.27) and (279.52,199.25) .. (283.25,199.25) .. controls (286.98,199.25) and (290,202.27) .. (290,206) .. controls (290,209.73) and (286.98,212.75) .. (283.25,212.75) .. controls (279.52,212.75) and (276.5,209.73) .. (276.5,206) -- cycle ;
%Straight Lines [id:da05560501775926163] 
\draw    (199.25,97.5) -- (199.25,112.5) ;
%Straight Lines [id:da578765961952844] 
\draw    (283.25,97.5) -- (283.25,199.25) ;
%Straight Lines [id:da5039638734429395] 
\draw    (185.75,23.5) -- (155.75,84) ;
%Straight Lines [id:da7196277438659886] 
\draw    (185.75,23.5) -- (199.25,84) ;
%Straight Lines [id:da6137218725231017] 
\draw    (185.75,23.5) -- (241.75,84) ;
%Straight Lines [id:da032385734517827336] 
\draw    (185.75,23.5) -- (283.25,84) ;
%Straight Lines [id:da37458305948884363] 
\draw    (155.75,84) -- (256.25,23.5) ;
%Straight Lines [id:da2782836008212979] 
\draw    (256.25,23.5) -- (199.25,84) ;
%Straight Lines [id:da2536458934721927] 
\draw    (256.25,23.5) -- (241.75,84) ;
%Straight Lines [id:da41496831319583616] 
\draw    (256.25,23.5) -- (283.25,84) ;
\end{tikzpicture}

        \caption{Subdivided star with center vertex blown-up by two vertices (in red), with vertices of value 0 and 1 filled in white and black respectively. This corresponds to the instance of the \textsc{Partition} problem with $S = \{3,1,5,4\}$.}
        \label{fig:partition_star_blown_up_edge}
    \end{figure}
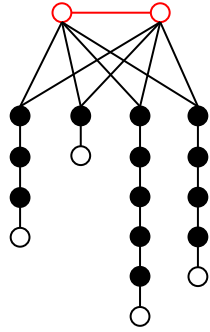
    
    To divide this graph among two agents with the given common binary additive valuation, we cannot cut any chain. Indeed, since bundles must be connected, this would imply that an agent receives only a strict subset of a single chain, but by our assumption that no number is larger than the rest, they would envy the other agent, even up to one vertex of value 1. Thus, each agent must receive one center vertex and multiple complete chains. Any allocation corresponds to a partition of the set of chains, or equivalently a partition of $S$, into two bundles. Since both bundles only have leaves of value 0, any \EFO allocation is in fact envy-free. And since the value of a bundle is precisely the sum of corresponding numbers in $S$, we get the desired equivalence.

    Since the \textsc{Partition} problem is only weakly NP-hard, we give a reduction from the strongly NP-hard problem \textsc{3-Partition} in Appendix~\ref{app:NP_reduction}. Note that this construction is only polynomial in the values, but not in the size of the input (i.e. the cardinality of $|S|$). To get a polynomial-size reduction, we can replace the chains of $v$ vertices of value 1, by a single vertex of value $v$, which requires extending the class of valuation functions to common additive (not necessarily binary).

    Finally, we give a reduction from \textsc{Partition} to an instance of \textsc{Outer Connected Fair Division of Graphs} with bounded bandwidth. Note that stars do not have bounded bandwidth. Instead we create the following graph. We start with $m$ vertices $v_1, \dots, v_m$ of the corresponding values of items in $S$. We add two disjoint \emph{railings}, that are paths on $m$ vertices of value zero. We attach each $v_i$ to the $i$\textsuperscript{th} vertex of each of the two railings and add a leaf of value zero attached to it. The construction is illustrated in Figure \ref{fig:3-partition_bounded_bandwidth}. This graph has bandwidth 4 by consistently numbering each $v_i$ and its three neighbours in increasing order of $i$. Any division of this graph among $2$ agents corresponds to a desired partition of $S$ into two equal parts, since both connected bundles must only have leaves of value zero as non-cut vertices. Conversely, using the two railings, we can cut the graph into two connected bundles corresponding to any partition $(S_1, S_2)$ of $S$.
    
    \begin{figure}[h!]
        \centering

\tikzset{every picture/.style={line width=0.75pt}} %set default line width to 0.75pt        

\begin{tikzpicture}[x=0.75pt,y=0.75pt,yscale=-0.7,xscale=0.7]

%Shape: Circle [id:dp1885461765467269] 
\draw   (157,72.75) .. controls (157,69.02) and (160.02,66) .. (163.75,66) .. controls (167.48,66) and (170.5,69.02) .. (170.5,72.75) .. controls (170.5,76.48) and (167.48,79.5) .. (163.75,79.5) .. controls (160.02,79.5) and (157,76.48) .. (157,72.75) -- cycle ;
%Shape: Circle [id:dp6141361352914103] 
\draw   (217.5,72.75) .. controls (217.5,69.02) and (220.52,66) .. (224.25,66) .. controls (227.98,66) and (231,69.02) .. (231,72.75) .. controls (231,76.48) and (227.98,79.5) .. (224.25,79.5) .. controls (220.52,79.5) and (217.5,76.48) .. (217.5,72.75) -- cycle ;
%Shape: Circle [id:dp8569486161329997] 
\draw   (278,72.75) .. controls (278,69.02) and (281.02,66) .. (284.75,66) .. controls (288.48,66) and (291.5,69.02) .. (291.5,72.75) .. controls (291.5,76.48) and (288.48,79.5) .. (284.75,79.5) .. controls (281.02,79.5) and (278,76.48) .. (278,72.75) -- cycle ;
%Shape: Circle [id:dp7865917146065812] 
\draw   (337.5,72.75) .. controls (337.5,69.02) and (340.52,66) .. (344.25,66) .. controls (347.98,66) and (351,69.02) .. (351,72.75) .. controls (351,76.48) and (347.98,79.5) .. (344.25,79.5) .. controls (340.52,79.5) and (337.5,76.48) .. (337.5,72.75) -- cycle ;
%Shape: Circle [id:dp4375062137395508] 
\draw   (398,72.75) .. controls (398,69.02) and (401.02,66) .. (404.75,66) .. controls (408.48,66) and (411.5,69.02) .. (411.5,72.75) .. controls (411.5,76.48) and (408.48,79.5) .. (404.75,79.5) .. controls (401.02,79.5) and (398,76.48) .. (398,72.75) -- cycle ;
%Shape: Circle [id:dp8039666358107113] 
\draw  [fill={rgb, 255:red, 0; green, 0; blue, 0 }  ,fill opacity=1 ] (147,145.75) .. controls (147,142.02) and (150.02,139) .. (153.75,139) .. controls (157.48,139) and (160.5,142.02) .. (160.5,145.75) .. controls (160.5,149.48) and (157.48,152.5) .. (153.75,152.5) .. controls (150.02,152.5) and (147,149.48) .. (147,145.75) -- cycle ;
%Shape: Circle [id:dp4916021062948961] 
\draw  [fill={rgb, 255:red, 0; green, 0; blue, 0 }  ,fill opacity=1 ] (207.5,145.75) .. controls (207.5,142.02) and (210.52,139) .. (214.25,139) .. controls (217.98,139) and (221,142.02) .. (221,145.75) .. controls (221,149.48) and (217.98,152.5) .. (214.25,152.5) .. controls (210.52,152.5) and (207.5,149.48) .. (207.5,145.75) -- cycle ;
%Shape: Circle [id:dp5065306656418656] 
\draw  [fill={rgb, 255:red, 0; green, 0; blue, 0 }  ,fill opacity=1 ] (268,145.75) .. controls (268,142.02) and (271.02,139) .. (274.75,139) .. controls (278.48,139) and (281.5,142.02) .. (281.5,145.75) .. controls (281.5,149.48) and (278.48,152.5) .. (274.75,152.5) .. controls (271.02,152.5) and (268,149.48) .. (268,145.75) -- cycle ;
%Shape: Circle [id:dp5929636851257375] 
\draw  [fill={rgb, 255:red, 0; green, 0; blue, 0 }  ,fill opacity=1 ] (327.5,145.75) .. controls (327.5,142.02) and (330.52,139) .. (334.25,139) .. controls (337.98,139) and (341,142.02) .. (341,145.75) .. controls (341,149.48) and (337.98,152.5) .. (334.25,152.5) .. controls (330.52,152.5) and (327.5,149.48) .. (327.5,145.75) -- cycle ;
%Shape: Circle [id:dp19899993273931937] 
\draw  [fill={rgb, 255:red, 0; green, 0; blue, 0 }  ,fill opacity=1 ] (388,145.75) .. controls (388,142.02) and (391.02,139) .. (394.75,139) .. controls (398.48,139) and (401.5,142.02) .. (401.5,145.75) .. controls (401.5,149.48) and (398.48,152.5) .. (394.75,152.5) .. controls (391.02,152.5) and (388,149.48) .. (388,145.75) -- cycle ;
%Shape: Circle [id:dp5966177700378156] 
\draw   (147,183.75) .. controls (147,180.02) and (150.02,177) .. (153.75,177) .. controls (157.48,177) and (160.5,180.02) .. (160.5,183.75) .. controls (160.5,187.48) and (157.48,190.5) .. (153.75,190.5) .. controls (150.02,190.5) and (147,187.48) .. (147,183.75) -- cycle ;
%Shape: Circle [id:dp3230833721286305] 
\draw   (207.5,183.75) .. controls (207.5,180.02) and (210.52,177) .. (214.25,177) .. controls (217.98,177) and (221,180.02) .. (221,183.75) .. controls (221,187.48) and (217.98,190.5) .. (214.25,190.5) .. controls (210.52,190.5) and (207.5,187.48) .. (207.5,183.75) -- cycle ;
%Shape: Circle [id:dp8322065605048705] 
\draw   (268,183.75) .. controls (268,180.02) and (271.02,177) .. (274.75,177) .. controls (278.48,177) and (281.5,180.02) .. (281.5,183.75) .. controls (281.5,187.48) and (278.48,190.5) .. (274.75,190.5) .. controls (271.02,190.5) and (268,187.48) .. (268,183.75) -- cycle ;
%Shape: Circle [id:dp7850715437456532] 
\draw   (327.5,183.75) .. controls (327.5,180.02) and (330.52,177) .. (334.25,177) .. controls (337.98,177) and (341,180.02) .. (341,183.75) .. controls (341,187.48) and (337.98,190.5) .. (334.25,190.5) .. controls (330.52,190.5) and (327.5,187.48) .. (327.5,183.75) -- cycle ;
%Shape: Circle [id:dp7735491822972161] 
\draw   (388,183.75) .. controls (388,180.02) and (391.02,177) .. (394.75,177) .. controls (398.48,177) and (401.5,180.02) .. (401.5,183.75) .. controls (401.5,187.48) and (398.48,190.5) .. (394.75,190.5) .. controls (391.02,190.5) and (388,187.48) .. (388,183.75) -- cycle ;
%Straight Lines [id:da29643171416031433] 
\draw    (153.75,152.5) -- (153.75,177) ;
%Straight Lines [id:da04530226848921026] 
\draw    (214.25,152.5) -- (214.25,177) ;
%Straight Lines [id:da7582830217105171] 
\draw    (274.75,152.5) -- (274.75,177) ;
%Straight Lines [id:da21109716297853476] 
\draw    (334.25,152.5) -- (334.25,177) ;
%Straight Lines [id:da7801592752329706] 
\draw    (394.75,152.5) -- (394.75,177) ;
%Straight Lines [id:da6021800474907649] 
\draw [color={rgb, 255:red, 208; green, 2; blue, 27 }  ,draw opacity=1 ]   (170.5,72.75) -- (217.5,72.75) ;
%Straight Lines [id:da5122912884318764] 
\draw [color={rgb, 255:red, 208; green, 2; blue, 27 }  ,draw opacity=1 ]   (231,72.75) -- (278,72.75) ;
%Straight Lines [id:da16397382131368388] 
\draw [color={rgb, 255:red, 208; green, 2; blue, 27 }  ,draw opacity=1 ]   (291.5,72.75) -- (337.5,72.75) ;
%Straight Lines [id:da7253561191884759] 
\draw [color={rgb, 255:red, 208; green, 2; blue, 27 }  ,draw opacity=1 ]   (351,72.75) -- (398,72.75) ;
%Shape: Circle [id:dp5520240057133794] 
\draw   (136,98.75) .. controls (136,95.02) and (139.02,92) .. (142.75,92) .. controls (146.48,92) and (149.5,95.02) .. (149.5,98.75) .. controls (149.5,102.48) and (146.48,105.5) .. (142.75,105.5) .. controls (139.02,105.5) and (136,102.48) .. (136,98.75) -- cycle ;
%Shape: Circle [id:dp12559275476831955] 
\draw   (196.5,98.75) .. controls (196.5,95.02) and (199.52,92) .. (203.25,92) .. controls (206.98,92) and (210,95.02) .. (210,98.75) .. controls (210,102.48) and (206.98,105.5) .. (203.25,105.5) .. controls (199.52,105.5) and (196.5,102.48) .. (196.5,98.75) -- cycle ;
%Shape: Circle [id:dp778384184266406] 
\draw   (257,98.75) .. controls (257,95.02) and (260.02,92) .. (263.75,92) .. controls (267.48,92) and (270.5,95.02) .. (270.5,98.75) .. controls (270.5,102.48) and (267.48,105.5) .. (263.75,105.5) .. controls (260.02,105.5) and (257,102.48) .. (257,98.75) -- cycle ;
%Shape: Circle [id:dp848204522411225] 
\draw   (316.5,98.75) .. controls (316.5,95.02) and (319.52,92) .. (323.25,92) .. controls (326.98,92) and (330,95.02) .. (330,98.75) .. controls (330,102.48) and (326.98,105.5) .. (323.25,105.5) .. controls (319.52,105.5) and (316.5,102.48) .. (316.5,98.75) -- cycle ;
%Shape: Circle [id:dp37820284851122965] 
\draw   (377,98.75) .. controls (377,95.02) and (380.02,92) .. (383.75,92) .. controls (387.48,92) and (390.5,95.02) .. (390.5,98.75) .. controls (390.5,102.48) and (387.48,105.5) .. (383.75,105.5) .. controls (380.02,105.5) and (377,102.48) .. (377,98.75) -- cycle ;
%Straight Lines [id:da016714750218947638] 
\draw [color={rgb, 255:red, 208; green, 2; blue, 27 }  ,draw opacity=1 ]   (149.5,98.75) -- (196.5,98.75) ;
%Straight Lines [id:da8887679278582782] 
\draw [color={rgb, 255:red, 208; green, 2; blue, 27 }  ,draw opacity=1 ]   (210,98.75) -- (257,98.75) ;
%Straight Lines [id:da7949394084329087] 
\draw [color={rgb, 255:red, 208; green, 2; blue, 27 }  ,draw opacity=1 ]   (270.5,98.75) -- (316.5,98.75) ;
%Straight Lines [id:da1159591368856292] 
\draw [color={rgb, 255:red, 208; green, 2; blue, 27 }  ,draw opacity=1 ]   (330,98.75) -- (377,98.75) ;
%Straight Lines [id:da7333960129123313] 
\draw    (142.75,105.5) -- (153.75,139) ;
%Straight Lines [id:da8689175703849558] 
\draw    (163.75,79.5) -- (153.75,139) ;
%Straight Lines [id:da44858797843420894] 
\draw    (203.25,105.5) -- (214.25,139) ;
%Straight Lines [id:da05602973258596122] 
\draw    (224.25,79.5) -- (214.25,139) ;
%Straight Lines [id:da034677535700208484] 
\draw    (263.75,105.5) -- (274.75,139) ;
%Straight Lines [id:da08634389700347112] 
\draw    (284.75,79.5) -- (274.75,139) ;
%Straight Lines [id:da2706364749866237] 
\draw    (323.25,105.5) -- (334.25,139) ;
%Straight Lines [id:da07463651867350696] 
\draw    (344.25,79.5) -- (334.25,139) ;
%Straight Lines [id:da9960622147859792] 
\draw    (383.75,105.5) -- (394.75,139) ;
%Straight Lines [id:da6278422979313477] 
\draw    (404.75,79.5) -- (394.75,139) ;

% Text Node
\draw (124,137) node [anchor=north west][inner sep=0.75pt]   [align=left] {$v_1$};
% Text Node
\draw (185,137) node [anchor=north west][inner sep=0.75pt]   [align=left] {$v_2$};
% Text Node
\draw (245,137) node [anchor=north west][inner sep=0.75pt]   [align=left] {$v_3$};
% Text Node
\draw (304,136) node [anchor=north west][inner sep=0.75pt]   [align=left] {$v_4$};
% Text Node
\draw (365,136) node [anchor=north west][inner sep=0.75pt]   [align=left] {$v_5$};

\end{tikzpicture}

        \caption{Graph of bandwith 4 constructed in the reduction from a \textsc{Partition} instance with $S = \{v_1, \dots, v_5\}$, where vertices filled in white have value 0, and vertices filled in black have the corresponding value $v_i$. The two railings are shown in red.}
        \label{fig:3-partition_bounded_bandwidth}
    \end{figure}
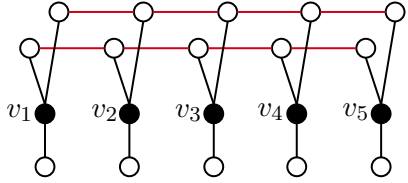 \vspace{-1cm}
\end{proof}

Note that this reduction ensures that the problem is NP-hard even on graphs of treewidth (and even pathwidth) $2$ and on graphs of bounded treedepth for general additive valuation functions. Also note that the graph of the second point is planar which shows that the problem is NP-hard even restricted to planar graphs.

\section{Conclusion and Future Work}

We resolved several open questions on connected \EFO allocations, which highlight remaining fundamental open problems. We showed that traceability is not a necessary condition for universal guarantees in the setting of common valuations. Extending our techniques to heterogeneous valuations is an important direction. In particular, we conjecture that our result on double suns continues to hold in this more general setting.

For common binary additive valuations on trees, we characterized the spectrum threshold using bad sets. Beyond binary valuations, even the case of common additive valuations with only two possible vertex values remains unresolved: the threshold may differ from the binary threshold, but a natural extension of bad sets appears to be a promising candidate characterization.

For general graphs, bad sets continue to provide a natural source of negative instances, suggesting a promising route toward understanding the spectrum. Fully characterizing the threshold beyond trees remains an important open question.

\bibliography{references}

@inproceedings{Abebe2017,
  author       = {Rediet Abebe and Jon M. Kleinberg and David C. Parkes},
  title        = {Fair Division via Social Comparison},
  booktitle    = {Proceedings of the 16th International Conference on Autonomous Agents and Multi-Agent Systems (AAMAS)},
  pages        = {281--289},
  year         = {2017},
  organization = {International Foundation for Autonomous Agents and Multiagent Systems},
  doi          = {10.5555/3091125.3091171}
}

@article{Amanatidis2023,
  author    = {Georgios Amanatidis and Haris Aziz and Georgios Birmpas and Aris Filos-Ratsikas and Bo Li and Herv{\'e} Moulin and Alexandros A. Voudouris and Xiaowei Wu},
  title     = {Fair division of indivisible goods: Recent progress and open questions},
  journal   = {Artificial Intelligence},
  volume    = {322},
  year      = {2023},
  pages     = {103965},
  doi       = {10.1016/j.artint.2023.103965}
}

@inproceedings{Aziz2019,
  author    = {Haris Aziz and Ioannis Caragiannis and Ayumi Igarashi and Toby Walsh},
  title     = {Fair Allocation of Indivisible Goods and Chores},
  booktitle = {Proceedings of the 28th International Joint Conference on Artificial Intelligence (IJCAI)},
  pages     = {53--59},
  year      = {2019},
  publisher = {IJCAI Organization}
}

@inproceedings{Bei2021,
  author    = {Bei, Xiaohui and Suksompong, Warut},
  title     = {Dividing a Graphical Cake},
  booktitle = {Proceedings of the 35th AAAI Conference on Artificial Intelligence (AAAI)},
  pages     = {5159--5166},
  year      = {2021}
}

@article{Bei2022,
  author    = {Xiaohui Bei and Ayumi Igarashi and Xinhang Lu and Warut Suksompong},
  title     = {The Price of Connectivity in Fair Division},
  journal   = {SIAM Journal on Discrete Mathematics},
  volume    = {36},
  number    = {2},
  year      = {2022},
  pages     = {}, 
  doi       = {10.1137/20M1388310}
}

@article{Bilo2022,
  author    = {Vittorio Bil{\`o} and Ioannis Caragiannis and Michele Flammini and Ayumi Igarashi and Gianpiero Monaco and Dominik Peters and Cosimo Vinci and William S.~Zwicker},
  title     = {Almost Envy{\textendash}Free Allocations with Connected Bundles},
  journal   = {Games and Economic Behavior},
  year      = {2022},
  volume    = {131},
  pages     = {197--221},
  doi       = {10.1016/j.geb.2021.11.006}
}

@incollection{Bouveret2015,
  author       = {Sylvain Bouveret and Yann Chevaleyre and Nicolas Maudet},
  title        = {Fair Allocation of Indivisible Goods},
  booktitle    = {Handbook of Computational Social Choice},
  editor       = {Felix Brandt and Vincent Conitzer and Ulle Endriss and J\'er\^ome Lang and Ariel D. Procaccia},
  chapter      = {12},
  pages        = {284--310},
  publisher    = {Cambridge University Press},
  year         = {2015},
  doi          = {10.1017/CBO9781107446984.013}
}

@inproceedings{Bouveret2017,
  author    = {Sylvain Bouveret and Katar{\'i}na Cechl{\'a}rov{\'a} and Edith Elkind and Ayumi Igarashi and Dominik Peters},
  title     = {Fair Division of a Graph},
  booktitle = {Proceedings of the Twenty-Sixth International Joint Conference on Artificial Intelligence (IJCAI)},
  year      = {2017},
  pages     = {135--141},
  doi       = {10.24963/ijcai.2017/20}
}

@article{bouveret2019,
  author    = {Sylvain Bouveret and Katar\'{\i}na Cechl\'{a}rov\'{a} and Julien Lesca},
  title     = {Chore Division on a Graph},
  journal   = {Autonomous Agents and Multi-Agent Systems},
  volume    = {33},
  number    = {5},
  pages     = {540--563},
  year      = {2019},
  doi       = {10.1007/s10458-019-09415-z}
}

@book{brams1996,
  author    = {Steven J. Brams and Alan D. Taylor},
  title     = {Fair Division: From Cake-Cutting to Dispute Resolution},
  publisher = {Cambridge University Press},
  year      = {1996},
  address   = {Cambridge, UK}
}

@article{Budish2011,
 author = {Eric Budish},
 journal = {Journal of Political Economy},
 number = {6},
 pages = {1061--1103},
 publisher = {The University of Chicago Press},
 title = {The Combinatorial Assignment Problem: Approximate Competitive Equilibrium from Equal Incomes},
 volume = {119},
 year = {2011}
}

@inproceedings{Chen2024, 
    author = {Chen, Jiehua and Zwicker, William S.}, 
    title = {Cutsets and {EF}1 Fair Division of Graphs}, 
    year = {2024}, 
    isbn = {9798400704864}, 
    publisher = {International Foundation for Autonomous Agents and Multiagent Systems}, 
    booktitle = {Proceedings of the 23rd International Conference on Autonomous Agents and Multiagent Systems}, 
    pages = {2192–2194},
    location = {Auckland, New Zealand}, 
    series = {AAMAS '24}
}

@article{Chevaleyre2017,
  author    = {Yann Chevaleyre and Ulle Endriss and Nicolas Maudet},
  title     = {Distributed Fair Allocation of Indivisible Goods},
  journal   = {Artificial Intelligence},
  volume    = {242},
  pages     = {1--22},
  year      = {2017},
  doi       = {10.1016/j.artint.2016.09.005}
}

@inproceedings{Deligkas2021,
  title     = {The Parameterized Complexity of Connected Fair Division},
  author    = {Deligkas, Argyrios and Eiben, Eduard and Ganian, Robert and Hamm, Thekla and Ordyniak, Sebastian},
  booktitle = {Proceedings of the Thirtieth International Joint Conference on Artificial Intelligence, {IJCAI-21}},
  publisher = {International Joint Conferences on Artificial Intelligence Organization},
  editor    = {Zhi-Hua Zhou},
  pages     = {139--145},
  year      = {2021},
  month     = {8},
  note      = {Main Track},
  doi       = {10.24963/ijcai.2021/20}
}

@inproceedings{Eiben2020,
    author    = {Eduard Eiben and Robert Ganian and Thekla Hamm and Sebastian Ordyniak},
    title     = {Parameterized Complexity of Envy‐Free Resource Allocation in Social Networks},
    booktitle = {Proceedings of the Thirty-Fourth AAAI Conference on Artificial Intelligence (AAAI 2020)},
    pages     = {7135--7142},
    publisher = {AAAI Press},
    year      = {2020}
}

@article{Foley1967,
    author  = {Duncan Foley},
    title   = {Resource allocation and the public sector},
    journal = {Yale Economic Essays},
    year    = {1967},
    volume  = {7},
    number  = {1}
}

@inproceedings{Greco2020,
    author    = {Gianluigi Greco and Francesco Scarcello},
    title     = {The Complexity of Computing Maximin Share Allocations on Graphs},
    booktitle = {Proceedings of the Thirty-Fourth AAAI Conference on Artificial Intelligence (AAAI 2020)},
    pages     = {2006--2013},
    publisher = {AAAI Press},
    year      = {2020}
}

@article{Hohne2021,
  author    = {Felix H\"ohne and Rob van Stee},
  title     = {Allocating Contiguous Blocks of Indivisible Chores Fairly},
  journal   = {Information and Computation},
  volume    = {281},
  pages     = {104739},
  year      = {2021},
  doi       = {10.1016/j.ic.2021.104739}
}

@inproceedings{Igarashi2023,
    author = {Igarashi, Ayumi},
    title = {How to cut a discrete cake fairly},
    year = {2023},
    isbn = {978-1-57735-880-0},
    publisher = {AAAI Press},
    booktitle = {Proceedings of the Thirty-Seventh AAAI Conference on Artificial Intelligence and Thirty-Fifth Conference on Innovative Applications of Artificial Intelligence and Thirteenth Symposium on Educational Advances in Artificial Intelligence},
    articleno = {636},
    numpages = {8},
    series = {AAAI'23/IAAI'23/EAAI'23}
}

@article{Igarashi2024,
  author    = {Ayumi Igarashi and William S. Zwicker},
  title     = {Fair Division of Graphs and of Tangled Cakes},
  journal   = {Mathematical Programming, Series B},
  volume    = {203},
  year      = {2024},
  pages     = {931--975},
  doi       = {10.1007/s10107-023-01969-9},
  month     = {April}
}

@article{Li2025,
  author    = {Bo Li and Ankang Sun and Mashbat Suzuki and Shiji Xing},
  title     = {On the Subsidy of Envy-Free Orientations in Graphs},
  journal   = {CoRR},
  volume    = {abs/2502.13671},
  year      = {2025},
  pages     = {},
  doi       = {10.48550/arXiv.2502.13671},
  url       = {https://arxiv.org/abs/2502.13671}
}

@inproceedings{Lipton2004,
    author = {Lipton, R. J. and Markakis, E. and Mossel, E. and Saberi, A.},
    title = {On approximately fair allocations of indivisible goods},
    year = {2004},
    isbn = {1581137710},
    publisher = {Association for Computing Machinery},
    address = {New York, NY, USA},
    booktitle = {Proceedings of the 5th ACM Conference on Electronic Commerce},
    pages = {125–131},
    numpages = {7},
}

@incollection{Markakis2017,
  author       = {Evangelos Markakis},
  title        = {Approximation Algorithms and Hardness Results for Fair Division with Indivisible Goods},
  booktitle    = {Trends in Computational Social Choice},
  editor       = {Ulle Endriss},
  chapter      = {12},
  pages        = {231--247},
  publisher    = {AI Access},
  year         = {2017}
}

@article{Moulin2019,
  author    = {Hervé Moulin},
  title     = {Fair Division in the Internet Age},
  journal   = {Annual Review of Economics},
  volume    = {11},
  pages     = {407--441},
  year      = {2019}
}

@inproceedings{Oh2019,
  author    = {Hoon Oh and Ariel D. Procaccia and Warut Suksompong},
  title     = {Fairly Allocating Many Goods with Few Queries},
  booktitle = {Proceedings of the 33rd AAAI Conference on Artificial Intelligence (AAAI)},
  pages     = {2141--2148},
  year      = {2019},
  organization = {AAAI Press}
}

@inproceedings{segalhalevi2018,
  author       = {Erel Segal-Halevi},
  title        = {Fairly Dividing a Cake after Some Parts Were Burnt in the Oven},
  booktitle    = {Proceedings of the 17th International Conference on Autonomous Agents and Multiagent Systems (AAMAS)},
  pages        = {1276--1284},
  year         = {2018},
  organization = {International Foundation for Autonomous Agents and Multiagent Systems}
}

@article{Stromquist1980,
  author  = {Walter Stromquist},
  title   = {How to Cut a Cake Fairly},
  journal = {The American Mathematical Monthly},
  volume  = {87},
  number  = {8},
  pages   = {640--644},
  year    = {1980},
  doi     = {10.1080/00029890.1980.11971518}
}

@article{Suksompong2019,
  author    = {Warut Suksompong},
  title     = {Fairly Allocating Contiguous Blocks of Indivisible Items},
  journal   = {Discrete Applied Mathematics},
  volume    = {260},
  pages     = {227--236},
  year      = {2019}
}

@article{Suksompong2021,
  author    = {Warut Suksompong},
  title     = {Constraints in Fair Division},
  journal   = {ACM SIGecom Exchanges},
  volume    = {19},
  number    = {2},
  pages     = {46--61},
  year      = {2021},
  doi       = {10.1145/3505156.3505162}
}

@incollection{Thomson2016,
  author       = {William Thomson},
  title        = {Introduction to the Theory of Fair Allocation},
  booktitle    = {Handbook of Computational Social Choice},
  editor       = {Felix Brandt and Vincent Conitzer and Ulle Endriss and J\'er\^ome Lang and Ariel D. Procaccia},
  chapter      = {11},
  pages        = {261--283},
  publisher    = {Cambridge University Press},
  year         = {2016}
}

\newpage
\appendix

\section{The case of subdivided claws with all branches of even size}
\label{app:subdivided_claw}

The last remaining case in the proof of Theorem~\ref{thm:binary} is the case where~$T$ is a subdivided claw with all branches of even size, and $n_1 \geqslant k \geqslant \bad{T}+2$. In this case, $n$ is odd. Proposition~\ref{prop:bad_set_at_least_n/2} ensures that $\bad{T}+2 \geqslant \left\lfloor\frac{n}{2}\right\rfloor$. Note that, if $k \geqslant \left\lceil\frac{n}{2}\right\rceil$, then the exact same proof as before applies. Thus, we cas assume that $k = \left\lfloor\frac{n}{2}\right\rfloor$.

First, assume that $V_0 \neq \emptyset$. The reason why the previous proof does not apply directly is because after the partitioning step we have more than~$k$ subtrees. Indeed, for a subdivided claw with all branches of even size, the partitioning step will always create $\left\lfloor \frac{n}{2} \right\rfloor$ subtrees of size~$2$, and one subtree of size~$1$, so $k+1$ subtrees in total. Let us explain how to adapt the proof by modifying this partitioning step. Our goal is to partition~$T$ into $k=\left\lfloor\frac{n}{2}\right\rfloor$ subtrees, one of size~$3$ having at least one vertex in~$V_0$, and the $k-1$ other subtrees having size~$2$. Note that, if we have such a partition, the rest of the proof (namely, the merging and splitting steps) is exactly the same. Let us now explain how to obtain such a partition. We first partition~$T$ into three paths, two that are branches (which thus have even size), and one that is a branch plus the branching node (which thus has odd size), with the condition that the path of odd size intersects~$V_0$. This is possible, because it suffices to group the branching node with a branch that intersects~$V_0$ (or with any branch if it itself belongs to~$V_0$). Then, we partition each of the two paths of even size into subtrees of size~$2$, and we partition the path of odd size into one subtree of size~$3$ which contains a vertex in~$V_0$ and subtrees of size~$2$. See Figure~\ref{fig:partitioning_subdivided_claw} for an illustration of this new partitioning step.

\begin{figure}[h!]
    \centering
    \begin{tikzpicture}[x=0.75pt,y=0.75pt,yscale=-0.8,xscale=0.8]
%uncomment if require: \path (0,443); %set diagram left start at 0, and has height of 443

%Straight Lines [id:da7452319231041531] 
\draw    (438.5,204.5) -- (464.5,219) ;
%Straight Lines [id:da6417171643058979] 
\draw    (308.5,58.5) -- (308.5,83) ;
%Straight Lines [id:da16643090381783898] 
\draw    (308.5,34) -- (308.5,58.5) ;
%Straight Lines [id:da4933172811624199] 
\draw    (230.5,175.5) -- (204.5,190) ;
%Straight Lines [id:da7863930459823071] 
\draw    (412.5,190) -- (438.5,204.5) ;
%Straight Lines [id:da5663772715512712] 
\draw    (386.5,175.5) -- (412.5,190) ;
%Straight Lines [id:da784378839541415] 
\draw    (282.5,146.5) -- (256.5,161) ;
%Straight Lines [id:da25324530029438097] 
\draw    (256.5,161) -- (230.5,175.5) ;
%Straight Lines [id:da3654905828681142] 
\draw    (334.5,146.5) -- (360.5,161) ;
%Straight Lines [id:da7944807192983117] 
\draw    (360.5,161) -- (386.5,175.5) ;
%Straight Lines [id:da6664588072876044] 
\draw    (308.5,83) -- (308.5,107.5) ;
%Straight Lines [id:da2683520073232416] 
\draw    (308.5,132) -- (282.5,146.5) ;
%Straight Lines [id:da41832524535798854] 
\draw    (308.5,132) -- (334.5,146.5) ;
%Straight Lines [id:da919701070034562] 
\draw    (308.5,107.5) -- (308.5,132) ;
%Shape: Circle [id:dp009864038146723986] 
\draw  [fill={rgb, 255:red, 0; green, 0; blue, 0 }  ,fill opacity=1 ] (303.5,132) .. controls (303.5,129.24) and (305.74,127) .. (308.5,127) .. controls (311.26,127) and (313.5,129.24) .. (313.5,132) .. controls (313.5,134.76) and (311.26,137) .. (308.5,137) .. controls (305.74,137) and (303.5,134.76) .. (303.5,132) -- cycle ;
%Shape: Circle [id:dp34943681916132896] 
\draw  [fill={rgb, 255:red, 0; green, 0; blue, 0 }  ,fill opacity=1 ] (355.5,161) .. controls (355.5,158.24) and (357.74,156) .. (360.5,156) .. controls (363.26,156) and (365.5,158.24) .. (365.5,161) .. controls (365.5,163.76) and (363.26,166) .. (360.5,166) .. controls (357.74,166) and (355.5,163.76) .. (355.5,161) -- cycle ;
%Shape: Circle [id:dp8587189875016208] 
\draw  [fill={rgb, 255:red, 0; green, 0; blue, 0 }  ,fill opacity=1 ] (277.5,146.5) .. controls (277.5,143.74) and (279.74,141.5) .. (282.5,141.5) .. controls (285.26,141.5) and (287.5,143.74) .. (287.5,146.5) .. controls (287.5,149.26) and (285.26,151.5) .. (282.5,151.5) .. controls (279.74,151.5) and (277.5,149.26) .. (277.5,146.5) -- cycle ;
%Shape: Circle [id:dp016714777540657755] 
\draw  [fill={rgb, 255:red, 0; green, 0; blue, 0 }  ,fill opacity=1 ] (303.5,107.5) .. controls (303.5,104.74) and (305.74,102.5) .. (308.5,102.5) .. controls (311.26,102.5) and (313.5,104.74) .. (313.5,107.5) .. controls (313.5,110.26) and (311.26,112.5) .. (308.5,112.5) .. controls (305.74,112.5) and (303.5,110.26) .. (303.5,107.5) -- cycle ;
%Shape: Circle [id:dp12363243983230077] 
\draw  [fill={rgb, 255:red, 0; green, 0; blue, 0 }  ,fill opacity=1 ] (303.5,83) .. controls (303.5,80.24) and (305.74,78) .. (308.5,78) .. controls (311.26,78) and (313.5,80.24) .. (313.5,83) .. controls (313.5,85.76) and (311.26,88) .. (308.5,88) .. controls (305.74,88) and (303.5,85.76) .. (303.5,83) -- cycle ;
%Shape: Circle [id:dp3922996362707891] 
\draw  [fill={rgb, 255:red, 0; green, 0; blue, 0 }  ,fill opacity=1 ] (329.5,146.5) .. controls (329.5,143.74) and (331.74,141.5) .. (334.5,141.5) .. controls (337.26,141.5) and (339.5,143.74) .. (339.5,146.5) .. controls (339.5,149.26) and (337.26,151.5) .. (334.5,151.5) .. controls (331.74,151.5) and (329.5,149.26) .. (329.5,146.5) -- cycle ;
%Shape: Circle [id:dp41119792607417727] 
\draw  [fill={rgb, 255:red, 255; green, 255; blue, 255 }  ,fill opacity=1 ] (381.5,175.5) .. controls (381.5,172.74) and (383.74,170.5) .. (386.5,170.5) .. controls (389.26,170.5) and (391.5,172.74) .. (391.5,175.5) .. controls (391.5,178.26) and (389.26,180.5) .. (386.5,180.5) .. controls (383.74,180.5) and (381.5,178.26) .. (381.5,175.5) -- cycle ;
%Shape: Circle [id:dp9027363191987893] 
\draw  [fill={rgb, 255:red, 255; green, 255; blue, 255 }  ,fill opacity=1 ] (251.5,161) .. controls (251.5,158.24) and (253.74,156) .. (256.5,156) .. controls (259.26,156) and (261.5,158.24) .. (261.5,161) .. controls (261.5,163.76) and (259.26,166) .. (256.5,166) .. controls (253.74,166) and (251.5,163.76) .. (251.5,161) -- cycle ;
%Shape: Circle [id:dp6076396671224557] 
\draw  [fill={rgb, 255:red, 0; green, 0; blue, 0 }  ,fill opacity=1 ] (225.5,175.5) .. controls (225.5,172.74) and (227.74,170.5) .. (230.5,170.5) .. controls (233.26,170.5) and (235.5,172.74) .. (235.5,175.5) .. controls (235.5,178.26) and (233.26,180.5) .. (230.5,180.5) .. controls (227.74,180.5) and (225.5,178.26) .. (225.5,175.5) -- cycle ;
%Shape: Circle [id:dp8153283627173313] 
\draw  [fill={rgb, 255:red, 0; green, 0; blue, 0 }  ,fill opacity=1 ] (407.5,190) .. controls (407.5,187.24) and (409.74,185) .. (412.5,185) .. controls (415.26,185) and (417.5,187.24) .. (417.5,190) .. controls (417.5,192.76) and (415.26,195) .. (412.5,195) .. controls (409.74,195) and (407.5,192.76) .. (407.5,190) -- cycle ;
%Shape: Circle [id:dp21415531788082287] 
\draw  [fill={rgb, 255:red, 0; green, 0; blue, 0 }  ,fill opacity=1 ] (433.5,204.5) .. controls (433.5,201.74) and (435.74,199.5) .. (438.5,199.5) .. controls (441.26,199.5) and (443.5,201.74) .. (443.5,204.5) .. controls (443.5,207.26) and (441.26,209.5) .. (438.5,209.5) .. controls (435.74,209.5) and (433.5,207.26) .. (433.5,204.5) -- cycle ;
%Shape: Circle [id:dp43678200988483273] 
\draw  [fill={rgb, 255:red, 0; green, 0; blue, 0 }  ,fill opacity=1 ] (199.5,190) .. controls (199.5,187.24) and (201.74,185) .. (204.5,185) .. controls (207.26,185) and (209.5,187.24) .. (209.5,190) .. controls (209.5,192.76) and (207.26,195) .. (204.5,195) .. controls (201.74,195) and (199.5,192.76) .. (199.5,190) -- cycle ;
%Shape: Circle [id:dp639928809614127] 
\draw  [fill={rgb, 255:red, 255; green, 255; blue, 255 }  ,fill opacity=1 ] (303.5,58.5) .. controls (303.5,55.74) and (305.74,53.5) .. (308.5,53.5) .. controls (311.26,53.5) and (313.5,55.74) .. (313.5,58.5) .. controls (313.5,61.26) and (311.26,63.5) .. (308.5,63.5) .. controls (305.74,63.5) and (303.5,61.26) .. (303.5,58.5) -- cycle ;
%Shape: Circle [id:dp8317349399646871] 
\draw  [fill={rgb, 255:red, 255; green, 255; blue, 255 }  ,fill opacity=1 ] (303.5,34) .. controls (303.5,31.24) and (305.74,29) .. (308.5,29) .. controls (311.26,29) and (313.5,31.24) .. (313.5,34) .. controls (313.5,36.76) and (311.26,39) .. (308.5,39) .. controls (305.74,39) and (303.5,36.76) .. (303.5,34) -- cycle ;
%Shape: Circle [id:dp9142074999299326] 
\draw  [fill={rgb, 255:red, 0; green, 0; blue, 0 }  ,fill opacity=1 ] (459.5,219) .. controls (459.5,216.24) and (461.74,214) .. (464.5,214) .. controls (467.26,214) and (469.5,216.24) .. (469.5,219) .. controls (469.5,221.76) and (467.26,224) .. (464.5,224) .. controls (461.74,224) and (459.5,221.76) .. (459.5,219) -- cycle ;
%Shape: Ellipse [id:dp7770296014181216] 
\draw  [line width=1.5]  (308.5,23.13) .. controls (314.05,23.13) and (318.55,33.48) .. (318.55,46.25) .. controls (318.55,59.02) and (314.05,69.37) .. (308.5,69.37) .. controls (302.95,69.37) and (298.45,59.02) .. (298.45,46.25) .. controls (298.45,33.48) and (302.95,23.13) .. (308.5,23.13) -- cycle ;
%Shape: Ellipse [id:dp47595863247788417] 
\draw  [line width=1.5]  (308.5,72.13) .. controls (314.05,72.13) and (318.55,82.48) .. (318.55,95.25) .. controls (318.55,108.02) and (314.05,118.37) .. (308.5,118.37) .. controls (302.95,118.37) and (298.45,108.02) .. (298.45,95.25) .. controls (298.45,82.48) and (302.95,72.13) .. (308.5,72.13) -- cycle ;
%Shape: Ellipse [id:dp8638548076102939] 
\draw  [line width=1.5]  (240.63,169.93) .. controls (243.7,175.48) and (235.84,185.72) .. (223.07,192.8) .. controls (210.3,199.88) and (197.45,201.12) .. (194.37,195.57) .. controls (191.3,190.02) and (199.16,179.78) .. (211.93,172.7) .. controls (224.7,165.62) and (237.55,164.38) .. (240.63,169.93) -- cycle ;
%Shape: Ellipse [id:dp42932469218275515] 
\draw  [line width=1.5]  (292.63,140.93) .. controls (295.7,146.48) and (287.84,156.72) .. (275.07,163.8) .. controls (262.3,170.88) and (249.45,172.12) .. (246.37,166.57) .. controls (243.3,161.02) and (251.16,150.78) .. (263.93,143.7) .. controls (276.7,136.62) and (289.55,135.38) .. (292.63,140.93) -- cycle ;
%Shape: Ellipse [id:dp5437754016500211] 
\draw  [line width=1.5]  (344.63,152.07) .. controls (341.55,157.62) and (328.7,156.38) .. (315.93,149.3) .. controls (303.16,142.22) and (295.3,131.98) .. (298.37,126.43) .. controls (301.45,120.88) and (314.3,122.12) .. (327.07,129.2) .. controls (339.84,136.28) and (347.7,146.52) .. (344.63,152.07) -- cycle ;
%Shape: Ellipse [id:dp6999392534052765] 
\draw  [line width=1.5]  (474.63,224.57) .. controls (471.55,230.12) and (458.7,228.88) .. (445.93,221.8) .. controls (433.16,214.72) and (425.3,204.48) .. (428.37,198.93) .. controls (431.45,193.38) and (444.3,194.62) .. (457.07,201.7) .. controls (469.84,208.78) and (477.7,219.02) .. (474.63,224.57) -- cycle ;
%Shape: Ellipse [id:dp460956414307953] 
\draw  [line width=1.5]  (424.29,196.45) .. controls (421.21,202) and (401.8,197.12) .. (380.93,185.55) .. controls (360.06,173.98) and (345.63,160.1) .. (348.71,154.55) .. controls (351.79,149) and (371.2,153.88) .. (392.07,165.45) .. controls (412.94,177.02) and (427.37,190.9) .. (424.29,196.45) -- cycle ;

\end{tikzpicture}

    \caption{The partitioning step for subdivided claws with all branches of even size, if $V_0 \neq \emptyset$. Vertices in~$V_0$ are colored white, vertices in~$V_1$ are colored black. This partitioning step gives us a partition of~$T$ into $\left\lfloor\frac{n}{2}\right\rfloor-1$ subtrees of size~$2$, and one subtree of size~$3$ that intersects~$V_0$.}
    \label{fig:partitioning_subdivided_claw}
\end{figure}
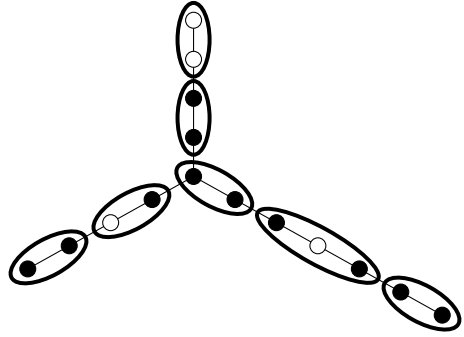

Finally, if $V_0 = \emptyset$, we can proceed as follows: we partition each branch into subtrees of size~$2$ (it is possible because all branches have even size), and we add the branching node to one of its adjacent subtrees. Thus, there are $\left\lfloor\frac{n}{2}\right\rfloor - 1$ subtrees of size~$2$, one subtree of size~$3$, and all vertices are in~$V_1$, so this allocation is \EFO.

\section{Deferred reduction from Section~\ref{sec:complexity}}
\label{app:NP_reduction}

\begin{lemma}
    The \textsc{3-Partition} problem reduces to the \textsc{Outer Connected Fair Divsion of Graphs} problem restricted to the blow-up of subdivided stars for common binary additive valuation functions.
\end{lemma}
\begin{proof}
    Given a multiset $S$ of positive integers, the \textsc{3-Partition} problem consists in deciding whether there exists a partition of the set $S$ into triplets that all have the same sum. 
    Let $m$ be the number of items in $S$. We construct a graph as follows. We start with the complete graph $K_m$ of size $m$, and assign each vertex infinite value (or some large constant, such as the sum of items in $S$). We create $m$ chains of length equal to the respective items in $S$, assigning each vertex value 1, and adding a leaf of value zero to both ends of every chain. Finally, we attach the (end of) chains to the clique, via a matching, i.e. each chain is adjacent to a unique vertex of $K_m$. This construction is illustrated in Figure \ref{fig:3-partition_star_blown_up_clique}. To divide this graph among $m/3$ agents, any agent must get exactly 3 central vertices to respect envyfreeness, and their corresponding chains to respect connectivity. Since all non-cut vertices of each such allocated bundle have value zero, we get that the corresponding partition of $S$ consists of triplets of the exact same value.
    \begin{figure}[h!]
        \centering
        
\tikzset{every picture/.style={line width=0.75pt}} %set default line width to 0.75pt        

\begin{tikzpicture}[x=0.75pt,y=0.75pt,yscale=-0.7,xscale=0.7]

%Shape: Circle [id:dp5896792246620907] 
\draw [red]  (247.75,36.75) .. controls (247.75,33.02) and (250.77,30) .. (254.5,30) .. controls (258.23,30) and (261.25,33.02) .. (261.25,36.75) .. controls (261.25,40.48) and (258.23,43.5) .. (254.5,43.5) .. controls (250.77,43.5) and (247.75,40.48) .. (247.75,36.75) -- cycle ;
%Shape: Circle [id:dp5871451767829492] 
\draw [red]  (199,36.75) .. controls (199,33.02) and (202.02,30) .. (205.75,30) .. controls (209.48,30) and (212.5,33.02) .. (212.5,36.75) .. controls (212.5,40.48) and (209.48,43.5) .. (205.75,43.5) .. controls (202.02,43.5) and (199,40.48) .. (199,36.75) -- cycle ;
%Shape: Circle [id:dp7949190575832291] 
\draw [red]  (224,64.75) .. controls (224,61.02) and (227.02,58) .. (230.75,58) .. controls (234.48,58) and (237.5,61.02) .. (237.5,64.75) .. controls (237.5,68.48) and (234.48,71.5) .. (230.75,71.5) .. controls (227.02,71.5) and (224,68.48) .. (224,64.75) -- cycle ;
%Shape: Circle [id:dp7789939184112632] 
\draw   (224,124.75) .. controls (224,121.02) and (227.02,118) .. (230.75,118) .. controls (234.48,118) and (237.5,121.02) .. (237.5,124.75) .. controls (237.5,128.48) and (234.48,131.5) .. (230.75,131.5) .. controls (227.02,131.5) and (224,128.48) .. (224,124.75) -- cycle ;
%Shape: Circle [id:dp4966921739875768] 
\draw   (282.75,99.75) .. controls (282.75,96.02) and (285.77,93) .. (289.5,93) .. controls (293.23,93) and (296.25,96.02) .. (296.25,99.75) .. controls (296.25,103.48) and (293.23,106.5) .. (289.5,106.5) .. controls (285.77,106.5) and (282.75,103.48) .. (282.75,99.75) -- cycle ;
%Straight Lines [id:da5301469635878038] 
\draw    (212.5,36.75) -- (289.5,93) ;
%Straight Lines [id:da5004173001477068] 
\draw    (254.5,43.5) -- (230.75,118) ;
%Straight Lines [id:da5567637509508881] 
\draw    (205.75,43.5) -- (230.75,118) ;
%Straight Lines [id:da5824554376248212] 
\draw    (230.75,71.5) -- (230.75,118) ;
%Shape: Circle [id:dp8665910154690021] 
\draw   (165.25,99.75) .. controls (165.25,96.02) and (168.27,93) .. (172,93) .. controls (175.73,93) and (178.75,96.02) .. (178.75,99.75) .. controls (178.75,103.48) and (175.73,106.5) .. (172,106.5) .. controls (168.27,106.5) and (165.25,103.48) .. (165.25,99.75) -- cycle ;
%Shape: Circle [id:dp0693772706560406] 
\draw  [fill={rgb, 255:red, 0; green, 0; blue, 0 }  ,fill opacity=1 ] (165.25,133.25) .. controls (165.25,129.52) and (168.27,126.5) .. (172,126.5) .. controls (175.73,126.5) and (178.75,129.52) .. (178.75,133.25) .. controls (178.75,136.98) and (175.73,140) .. (172,140) .. controls (168.27,140) and (165.25,136.98) .. (165.25,133.25) -- cycle ;
%Shape: Circle [id:dp06539254111806747] 
\draw  [fill={rgb, 255:red, 0; green, 0; blue, 0 }  ,fill opacity=1 ] (165.25,161.75) .. controls (165.25,158.02) and (168.27,155) .. (172,155) .. controls (175.73,155) and (178.75,158.02) .. (178.75,161.75) .. controls (178.75,165.48) and (175.73,168.5) .. (172,168.5) .. controls (168.27,168.5) and (165.25,165.48) .. (165.25,161.75) -- cycle ;
%Shape: Circle [id:dp7607567940372177] 
\draw  [fill={rgb, 255:red, 0; green, 0; blue, 0 }  ,fill opacity=1 ] (165.25,190.25) .. controls (165.25,186.52) and (168.27,183.5) .. (172,183.5) .. controls (175.73,183.5) and (178.75,186.52) .. (178.75,190.25) .. controls (178.75,193.98) and (175.73,197) .. (172,197) .. controls (168.27,197) and (165.25,193.98) .. (165.25,190.25) -- cycle ;
%Shape: Circle [id:dp6374850797805277] 
\draw  [fill={rgb, 255:red, 0; green, 0; blue, 0 }  ,fill opacity=1 ] (165.25,218.75) .. controls (165.25,215.02) and (168.27,212) .. (172,212) .. controls (175.73,212) and (178.75,215.02) .. (178.75,218.75) .. controls (178.75,222.48) and (175.73,225.5) .. (172,225.5) .. controls (168.27,225.5) and (165.25,222.48) .. (165.25,218.75) -- cycle ;
%Shape: Circle [id:dp19915224191300396] 
\draw  [fill={rgb, 255:red, 0; green, 0; blue, 0 }  ,fill opacity=1 ] (224,153.25) .. controls (224,149.52) and (227.02,146.5) .. (230.75,146.5) .. controls (234.48,146.5) and (237.5,149.52) .. (237.5,153.25) .. controls (237.5,156.98) and (234.48,160) .. (230.75,160) .. controls (227.02,160) and (224,156.98) .. (224,153.25) -- cycle ;
%Shape: Circle [id:dp5031771576183209] 
\draw  [fill={rgb, 255:red, 0; green, 0; blue, 0 }  ,fill opacity=1 ] (282.75,128.25) .. controls (282.75,124.52) and (285.77,121.5) .. (289.5,121.5) .. controls (293.23,121.5) and (296.25,124.52) .. (296.25,128.25) .. controls (296.25,131.98) and (293.23,135) .. (289.5,135) .. controls (285.77,135) and (282.75,131.98) .. (282.75,128.25) -- cycle ;
%Shape: Circle [id:dp0814981827247242] 
\draw  [fill={rgb, 255:red, 0; green, 0; blue, 0 }  ,fill opacity=1 ] (282.75,156.75) .. controls (282.75,153.02) and (285.77,150) .. (289.5,150) .. controls (293.23,150) and (296.25,153.02) .. (296.25,156.75) .. controls (296.25,160.48) and (293.23,163.5) .. (289.5,163.5) .. controls (285.77,163.5) and (282.75,160.48) .. (282.75,156.75) -- cycle ;
%Shape: Circle [id:dp2885501638126292] 
\draw   (165.25,247.25) .. controls (165.25,243.52) and (168.27,240.5) .. (172,240.5) .. controls (175.73,240.5) and (178.75,243.52) .. (178.75,247.25) .. controls (178.75,250.98) and (175.73,254) .. (172,254) .. controls (168.27,254) and (165.25,250.98) .. (165.25,247.25) -- cycle ;
%Shape: Circle [id:dp2737020922632992] 
\draw   (224,181.75) .. controls (224,178.02) and (227.02,175) .. (230.75,175) .. controls (234.48,175) and (237.5,178.02) .. (237.5,181.75) .. controls (237.5,185.48) and (234.48,188.5) .. (230.75,188.5) .. controls (227.02,188.5) and (224,185.48) .. (224,181.75) -- cycle ;
%Shape: Circle [id:dp18052289285942447] 
\draw   (282.75,185.25) .. controls (282.75,181.52) and (285.77,178.5) .. (289.5,178.5) .. controls (293.23,178.5) and (296.25,181.52) .. (296.25,185.25) .. controls (296.25,188.98) and (293.23,192) .. (289.5,192) .. controls (285.77,192) and (282.75,188.98) .. (282.75,185.25) -- cycle ;
%Straight Lines [id:da21729365425199887] 
\draw    (172,106.5) -- (172,240.5) ;
%Straight Lines [id:da11551261154740977] 
\draw    (289.5,93) -- (230.75,71.5) ;
%Straight Lines [id:da27215669182497615] 
\draw    (254.5,43.5) -- (289.5,93) ;
%Straight Lines [id:da6348654185045882] 
\draw    (172,93) -- (205.75,43.5) ;
%Straight Lines [id:da9643397004663753] 
\draw    (247.75,36.75) -- (172,93) ;
%Straight Lines [id:da09248834896338232] 
\draw    (172,93) -- (230.75,71.5) ;
%Straight Lines [id:da30222139250620195] 
\draw    (230.75,131.5) -- (230.75,175) ;
%Straight Lines [id:da5919260794121167] 
\draw    (289.5,106.5) -- (289.5,178.5) ;
%Straight Lines [id:da28184022954293486] 
\draw [red]   (212.5,36.75) -- (230.75,58) ;
%Straight Lines [id:da8404862093254301] 
\draw [red]  (247.75,36.75) -- (230.75,58) ;
%Straight Lines [id:da8349749340389181] 
\draw [red]  (212.5,36.75) -- (247.75,36.75) ;
\end{tikzpicture}
        
        \caption{Subdivided star with center vertex blown-up by a clique of size $m=3$ (in red), with vertices of value 0 and 1 filled in white and black respectively. This corresponds to the instance of the \textsc{3-Partition} problem with $S = \{4,1,2\}$.}
        \label{fig:3-partition_star_blown_up_clique}
    \end{figure}
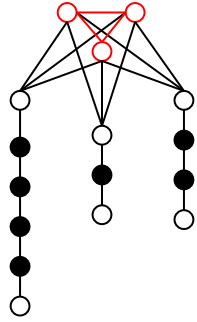
\end{proof}

\end{document}